\documentclass[11pt]{article}
\usepackage[utf8]{inputenc}
\usepackage[a4paper]{geometry}
\usepackage[draft]{todonotes}
\usepackage[english]{babel}
\usepackage{libertine}
\usepackage{amsmath, amsthm}
\usepackage{amsfonts}

\usepackage[psamsfonts]{amssymb}
\usepackage{mathrsfs}
\usepackage{textcomp}
\usepackage{mathabx}
\usepackage{pst-node}
\usepackage{braket}
\usepackage{tikz-cd}

\usepackage{graphicx}
\usepackage{enumitem}
\usepackage{bm}
\usepackage{caption}
\usepackage{subcaption}
\usepackage{framed}
\usepackage{fancybox}
\usepackage{dsfont}
\usepackage{tikz}
\usepackage{mathtools,xparse}
\usepackage{hyperref}
\hypersetup{
    colorlinks=true, 
    linktoc=all,     
    linkcolor=blue,  
}
\usepackage{scalerel,stackengine}
\stackMath

\newcommand\reallywidehat[1]{%
\savestack{\tmpbox}{\stretchto{%
  \scaleto{%
    \scalerel*[\widthof{\ensuremath{#1}}]{\kern-.6pt\bigwedge\kern-.6pt}%
    {\rule[-\textheight/2]{1ex}{\textheight}}
  }{\textheight}%
}{0.5ex}}%
\stackon[1pt]{#1}{\tmpbox}%
}
\DeclarePairedDelimiter{\norm}{\lVert}{\rVert}
\NewDocumentCommand{\normL}{ s O{} m }{%
  \IfBooleanTF{#1}{\norm*{#3}}{\norm[#2]{#3}}_{L_2(\Omega)}%
}

\newtheorem{theorem}{Theorem}[section]
\newtheorem{lemma}[theorem]{Lemma}
\newtheorem{assumption}{Assumption}
\newtheorem{prop}[theorem]{Proposition}
\newtheorem{corollary}{Corollary}[theorem]
\theoremstyle{definition}

\newtheorem{definition}[theorem]{Definition}

\usepackage{framed}
\usepackage{comment}
\iftrue    
\excludecomment{detail}
\else
\newenvironment{detail}{
    \begin{leftbar}
    }
    {
   \end{leftbar}
   \vspace{1ex}
   }
\fi

\newcommand{\bb}{\mathbb}

\newcommand{\cl}{\mathcal}

\newcommand{\m}{\mathrm}

\newcommand{\bnorm}{\norm[\bigg]}
\newcommand{\Bnorm}{\norm[\Bigg]}

\newcommand{\rmi}{\mathrm{i}}
\newcommand{\rme}{\mathrm{e}}

\newcommand{\supp}{\mathrm{supp}}
\newcommand{\toinfty}{\xrightarrow[n\to\infty]{} }

\theoremstyle{remark}
\newtheorem{remark}[theorem]{Remark}

\newcommand{\vp}{V_{\mathrm{per}}}
\newcommand{\hp}{H_{\mathrm{per}}}
\newcommand{\hpx}{H_{\mathrm{per},\xi}}
\newcommand{\rp}{\gamma_{\mathrm{per}}}
\newcommand{\ep}{q_{\mathrm{per}}}
\newcommand{\epx}{\overline{q_{\mathrm{per}}}}
\newcommand{\rpo}{\gamma_{\mathrm{per}}^{-}}
\newcommand{\rpx}{\gamma_{\mathrm{per},\xi}}
\newcommand{\rpxo}{\gamma_{\mathrm{per},\xi}^{-}}
\newcommand{\mup}{\mu_{\mathrm{per}}}
\newcommand{\mups}{\mu_{\mathrm{per}}}
\newcommand{\mupsm}{\mu_{\mathrm{per},\mathrm{sym}}}
\newcommand{\vl}{V_{\mathrm{per},L}}
\newcommand{\vr}{V_{\mathrm{per},R}}
\newcommand{\hl}{H_{\mathrm{per},L}}
\newcommand{\hr}{H_{\mathrm{per},R}}

\newcommand{\hc}{H_{\chi}}

\newcommand{\rol}{\rho_{\mathrm{per},L}}
\newcommand{\ror}{\rho_{\mathrm{per},R}}
\newcommand{\gl}{\gamma_{\mathrm{per},L}}

\newcommand{\lpg}{L_{\mathrm{per},x}^2\left(\Gamma\right)}

\newcommand\Z{{\mathbb Z}}
\newcommand\R{{\mathbb R}}

\newcommand{\gS}{{\mathfrak{S}}}
\newcommand{\br}{\mathbf r}
\newcommand{\Spx}{\mathscr{S}_{\mathrm{per},x}}
\newcommand{\F}{\mathscr{F}}
\newcommand{\B}{\mathscr{B}}
\numberwithin{equation}{section}
\newlength\dlf  

\begin{document}

\title{Mean--field stability for the junction of quasi 1D systems with Coulomb interactions}
\author{Ling-Ling CAO\\
\small Universit\'{e} Paris-Est, CERMICS (ENPC), F-77455 Marne-la-Vall\'{e}e}
\maketitle
\begin{abstract}
Junctions appear naturally when one studies surface states or transport properties of quasi one dimensional materials such as carbon nanotubes, polymers and quantum wires. These materials can be seen as 1D systems embedded in the 3D space.
In this article, we first establish a mean--field description of reduced Hartree--Fock type for a 1D periodic system in the 3D space (a quasi 1D system), the unit cell of which is unbounded. With mild summability condition, we next show that a quasi 1D system in its ground state can be described by a mean--field Hamiltonian. We also prove that the Fermi level of this system is always negative. A junction system is described by two different infinitely extended quasi 1D systems occupying separately half spaces in 3D, where Coulombic electron-electron interactions are taken into account and without any assumption on the commensurability of the periods. We prove the existence of the ground state for a junction system, the ground state is a spectral projector of a mean--field Hamiltonian, and the ground state density is unique.
\end{abstract}
\tableofcontents
\section{Introduction}
\subsection{Physical background and mathematical models}
Atomic junctions of quasi 1D systems appear for instance when studying the surface states of one-dimensional (1D) crystals~\cite{AERTS19601063,PhysRev.56.317}, quantum thermal transport in nanostructures~\cite{Wang2008}, and p-n junctions~\cite{Nature_2014, Nature_Pospischil} which are the foundation of the modern semiconductor electronic devices. Besides, electronic transport in carbon nanotubes~\cite{RevModPhys.87.703} and in molecular wires~\cite{MolecularWire}, which recently attracted a lot of interest, is often modeled by the junction of two semi-infinite systems with different chemical potentials. In recent years, studies of various quantum Hall effects and topological insulators focussed attention on 2D materials, see~\cite{RevModPhys.82.3045} and references therein. These 2D materials often possess periodicity in one dimension and can therefore be reduced to quasi 1D materials by momentum representation in the periodic direction~\cite{Hatsugai1993b}. Furthermore, when studying edge states properties (see~\cite{Hatsugai1993b, Avila2013} and references therein) of 2D materials, they can be seen as a junction with the vacuum.

Real world materials are often described by periodic~\cite{CATTO2001687, Cances2008} or ergodically periodic~\cite{CancesLahbabiLewin2013} systems in mathematical modeling. In this article we consider a junction of two different quasi 1D periodic systems without any assumption on the commensurability of the periods. Generally speaking, there are two regimes for the junction of two different periodic systems: when the chemical potentials of the underlying periodic systems are separated by some occupied bands (non-equilibrium regime, see Fig.~\ref{junction_1}), and when the chemical potential are in a common spectral gap (equilibrium regime, see Fig.~\ref{junction_2}). The non-equilibrium regime models a persistent (non-perturbative) current in the junction system~\cite{Bruneau2015, Bruneau2016,Bruneau2016II,Cornean2012}, while the equilibrium regime can model either the ground state of the junction material or the presence of perturbative current in the linear response regime~\cite{Cornean2008}. In this article we consider the equilibrium regime, and only briefly comment on the non-equilibrium regime in Section~\ref{refSec}, as the study of this situation requires different techniques. 

The most prominent feature of quasi 1D materials is the presence of strong electron-electron interactions due to low screening effect~\cite{BrusL2010, BrusL2014} as electrons interact through the 3D space. For finite systems, one can use a $N$-body Schr\"{o}dinger model to describe the electron-electron interactions. Nevertheless, this is impossible for infinite systems. Mean--field theory is a good candidate for infinite systems: it consists in replacing the $N$-body interactions by a $1$-body interaction with an effective average field, leading to a quasi-particle description of the system. However, mean--field models are rarely available for quasi 1D periodic systems, as periodic systems are often considered either in the 3D space (see~\cite{LIEB197722, catto1998mathematical} for Thomas--Fermi type models and~\cite{CATTO2001687} for Hartree--Fock type models), or strictly in a 1D geometry (see~\cite{Blanc2002} for Thomas--Fermi type models). To our knowledge, the work on Thomas--Fermi type model~\cite{blanc2000} for polymers is the only literature available for a 1D periodic system with interactions through the 3D space. Furthermore, non-periodic infinite systems are difficult to handle mathematically, as they do not possess any symmetry, hence the usual Bloch decomposition of periodic systems~\cite{ReeSim4, CATTO2001687} is not applicable, and the definition of the ground state energy needs to be examined~\cite{BlancLeBrisLions2003}. 

In this article, we establish a mean-field model to describe the junction of two different quasi 1D periodic systems (see Fig. \ref{Junction_pic}) in the 3D space with Coulomb interactions, under the framework of the reduced Hartree--Fock~\cite{Solovej1991} (rHF) description. Remark that the rHF model is strictly convex in the density, and can be seen as a good approximation of Kohn-Sham LDA model~\cite{KS65, ANANTHARAMAN20092425, lewin19}, which is widely used in condensed matter physics. This non-linear model can be employed to describe the junction of two nanotubes, or a more realistic model of junction of two quasi 1D crystals for electronic structure calculations. It can be further explored to study the linear response with respect to different Fermi levels between two semi-infinite chains: recall that the famous Landauer--B\"{u}ttiker~\cite{Landauer, ButtikerLandauer} formalism for electronic (thermal) transport which is based on the lead-device-lead description, can be seen as the junction of two different quasi 1D systems (leads) with different chemical (thermal) potentials, and the device as a perturbation of this junction. Remark also that p-n junctions of carbon nanotubes without external battery~\cite{PhysRevLett_nanotube, LeeJ_nanotube} correspond to the equilibrium regime, and can thus be described by the model we consider. Futhermore, our model can also be easily adapted to describe 1D dislocation problems in the 3D space, while the linear 1D dislocation problems have been studied in~\cite{Korotyaev2000, Korotyaev2003} and some generalizations have been provided for higher dimensional systems~\cite{ Dohnal2011,HEMPEL2011166, Hempel2014}. 

\subsection{Summary of main results}
The organization of this article and the main results are as follows: in Section~\ref{SecRHFPeriodic} we consider a quasi 1D periodic system, which is described by nuclei arranged periodically alongside the $x$-axis (see Fig.~\ref{Junction_sym}) with electrons occupying the 3D space, as it is a building block for the junction system. 
\begin{figure}[h!]
\centering
\includegraphics[width=0.6\textwidth ]{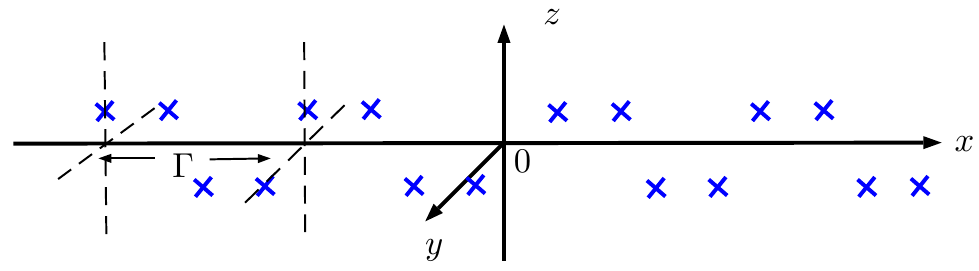}
\caption{An example of nuclei configuration of a quasi 1D periodic system.}
\label{Junction_sym}
\end{figure}
We define a periodic rHF energy functional~\eqref{defEnergyfunc} by taking into account the real Coulomb interactions in the 3D space. In Theorem~\ref{thmPeriodicExistence1} we show that this rHF functional admits minimizers, and that the ground state electronic density is unique. Remark that this is different from~\cite{CATTO2001687, Cances2008} as the system is periodic only in the $x$-direction, the unit cell $\Gamma$ being unbounded so that additional compactness proofs are needed when dealing with the ground state problem. With a mild summability condition~\eqref{mildIntegrabilityDensity} on the unique density of minimizers, we are able to obtain a mean--field Hamiltonian $\hp = -\frac{1}{2}\Delta +\vp$ to describe a quasi 1D periodic system, where $\vp$ is the mean--field potential that tends to $0$ in the $\bm r:= (y,z)$ direction. In Theorem~\ref{thmPeriodicExistence2} we prove that the Fermi level $\epsilon_F$ of the quasi 1D system, which represents the highest energy attainable by electrons under this quasi-particle description, is always negative. We also prove that the unique minimizer is a spectral projector of the mean--field Hamiltonian~$\rp = \mathds 1_{(-\infty,\epsilon_F]}\left(\hp\right) $.

In Section~\ref{SecJunc}, \begin{figure}[h!]
\centering
\includegraphics[width=0.6\textwidth ]{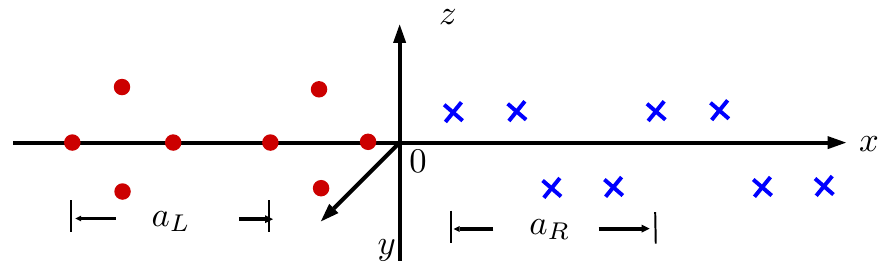}
\caption{Nuclei configuration of the junction system with period $a_{L}$ on $(-\infty,0]\times\bb R^2$ and $a_{R}$ on $(0,+\infty)\times\bb R^2$.}
\label{Junction_pic}
\end{figure}under certain symmetry assumptions on the nuclear densities $ \mu_{\mathrm{per},L }$ and $\mu_{\mathrm{per},R}$ of two different quasi 1D periodic systems, the junction system is described by considering the following nuclear configuration (see Fig.~\ref{Junction_pic}): 
$$\mu_{J}:= \mathds{1}_{x\leq 0}\cdot \mu_{\mathrm{per},L}+\mathds{1}_{x >0}\cdot \mu_{\mathrm{per},R} + v,$$
where $v$ describes how the junction is initiated. We aim at establishing a quasi-particle description of this infinitely extended junction system with Coulomb interactions and show the existence of ground state. As we do not assume any commensurability of periods of the two quasi 1D systems, the junction system does not possess any translation-invariant symmetry. The main idea is to establish a well-suited reference system based on the linear combination of periodic systems, and use perturbative techniques which have been widely used for mean-field type models~\cite{HAINZL2005TheMA, Hainzl2005, Hainzl2009, Cances2008, frank2013} to justify the construction. More precisely, we define a reference Hamiltonian $\hc = \chi^2  \hl + (1-\chi^2)\hr$ with $\chi$ a smooth cut-off function approximating $ \mathds{1}_{x\leq 0}$, where $\hl$ and $\hr$ are the mean--field Hamiltonians of the quasi 1D periodic systems. Denote by $\sigma(A) = \sigma_{\mathrm{disc}}(A)\bigcup \sigma_{\mathrm{ess}}(A)$ the spectrum of $A$, where $\sigma_{\mathrm{disc}}(A)$ (resp. $\sigma_{\mathrm{ess}}(A)$) denotes the discrete (resp. essential) spectrum of an operator $A$. Denote also by $\sigma_{\mathrm{ac}}(A)$ the purely absolutely continuous spectrum of $A$. We first show in Proposition~\ref{spectrapPpHchi} that $$\sigma_{\mathrm{ess}}\left(\hc\right) = \sigma_{\mathrm{ess}}\left(\hl\right)\bigcup  \sigma_{\mathrm{ess}}\left(\hr\right),\quad\sigma_{\mathrm{ess}}\left(\hc\right)\bigcap\, (-\infty,0] \subseteq \sigma_{\mathrm{ac}}\left(\hc\right).$$ This implies that i) the essential spectrum of the reference Hamiltonian is independent of the cut-off function $\chi$, ii) the linear junction preserves the scattering channels of the underlying systems, since the purely absolutely continuous spectrum of Hamiltonian has not been modified, hence the linear junction can be used to study the electronic conductance with the Landauer-B\"{u}ttiker formalism (see for example~\cite{Bruneau2015, Bruneau2016,Bruneau2016II} as well as the discussion following Proposition~\ref{spectrapPpHchi}).

After introducing a reference state $ \gamma_{\chi}:=\mathds 1_{(-\infty,\epsilon_F)}(\hc)$, we show in Proposition~\ref{propExist} that the electronic density $\rho_{\chi}$ is close to the linear combination of the underlying periodic electronic densities, and that the difference with these reference densities decays exponentially fast. The quasi-particle description of the nonlinear junction state can be constructed by considering $$\gamma_{J}  = \gamma_{\chi} + Q_{\chi},$$ where $Q_{\chi}$ is a trial density matrix which encodes the nonlinear effects of the junction system. Following the idea developed in~\cite{Cances2008}, we associate $Q_\chi$ with some minimization problem in Proposition~\ref{PropExistMinimizer}, and denote by $\overline{Q}_{\chi}$ a minimizer. 
We prove in Theorem~\ref{thm2} that 
$$\rho_{\gamma_{J}}  = \rho_{\chi} + \rho_{\overline{Q_{\chi}}} \,\,\text{  is independent of $\chi$. }  $$
This implies that the ground state of the junction system with Coulomb interactions exists and its density is independent of the choice of the reference state, see Corollary~\ref{junction_gs}.
\section{A reduced Hartree--Fock description of quasi 1D periodic systems}
\label{SecRHFPeriodic}
In this section, we give a mathematical description of a quasi 1D periodic system in the framework of the reduced Hartree-Fock (rHF) approach. In Section~\ref{PreliminarySec}, we introduce some mathematical preliminaries. In Section~\ref{secRHF1Dper} we construct a periodic rHF energy functional for a quasi 1D system.

Let us first introduce some notation. Unless otherwise specified, the functions on $\R^d$ considered in this article are complex-valued. Elements of $\R^3$ are denoted by $\bm x= (x,\bm r)$, where $x \in \R$. For a given separable Hilbert space $\mathfrak{H}$, we denote by $\mathcal{L}(\mathfrak{H})$ the space of bounded linear operators acting on $\mathfrak{H}$, by $\mathcal{S}(\mathfrak{H})$ the space of bounded self-adjoint operators acting on $\mathfrak{H}$, and by $\mathfrak{S}_p(\mathfrak{H})$ the Schatten class of operators acting on $\mathfrak{H}$. For $1 \leq p < \infty$, a compact operator $A$ belongs to $\mathfrak{S}_p(\mathfrak{H})$ if and only if $\norm{A}_{\mathfrak{S}_p}:= (\m{Tr}(|A|^p))^{1/p} <\infty$. Operators in $\mathfrak{S}_1(\mathfrak{H})$ and $\mathfrak{S}_2(\mathfrak{H})$ are respectively called trace-class and Hilbert--Schmidt. If $A \in \mathfrak{S}_1(L^2(\R^d))$, there exists a unique function $\rho_A \in L^1(\R^d)$ such that
\[
\forall \phi \in L^\infty(\R^d), \quad \m{Tr}(A\phi) = \int_{\R^d} \rho_A \phi.
\]
The function $\rho_A$ is called the density of the operator $A$. If the integral kernel $A(\br,\br')$ of $A$ is continuous on $\R^d \times \R^d$, then $\rho_A(\br)=A(\br,\br)$ for all $\br \in \R^d$. This relation still stands in some weaker sense for a generic trace-class operator.

An operator $A \in \mathcal{L}(L^2(\R^d))$ is called locally trace-class if the operator $\varrho A \varrho$ is trace-class for any $\varrho \in C^\infty_c(\R^d)$. The density of a locally trace-class operator $A\in \mathcal{L}(L^2(\R^d))$ is the unique function $\rho_A \in L^1_{\rm loc}(\R^d)$ such that
\[
\forall \phi \in C^\infty_c(\R^d), \quad \m{Tr}(A\phi) = \int_{\R^d} \rho_A \phi.
\]
Let $\mathscr{S}(\bb R^d)$ be the Schwartz space of rapidly decreasing functions on $\bb R^d$, and $\mathscr{S}'(\bb R^d)$ the space of tempered distributions on $\bb R^d$. We denote by $\widehat{\phi}$ (resp. $\widecheck{\phi}$) the Fourier transform (resp. inverse Fourier transform) on $\mathscr{S}'(\bb R^d)$, with the following normalization: $$\forall \phi\in L^1(\bb R^d),\qquad\widehat{\phi}(\zeta):=\frac{1}{(2\pi)^{d/2}}\int_{\bb R^d} \phi(x)\,\rme^{-\rmi \zeta x}\,dx,\quad \widecheck{\phi}(x) := \frac{1}{(2\pi)^{d/2}}\int_{\bb R^d}\phi(\zeta )\,\rme^{\rmi \zeta x}\,d\zeta .$$
The normalization ensures that the Fourier transform defines a unitary operator on $L^2(\bb R^d)$. 

\subsection{Mathematical preliminaries}
\label{PreliminarySec}
We first introduce a decomposition of the operator which is $\bb Z$-translation invariant in the $x$-direction based on the partial Bloch transform. In order to describe the 1D periodic system in the 3D space, we next introduce a mixed Fourier transform. We also introduce a Green's function which is periodic only in the $x$-direction. Finally we introduce the kinetic energy space of density matrices and Coulomb interactions for quasi 1D systems.
\paragraph{Bloch transform in the $x$-direction.} For $k \in \bb Z$, we denote by $\tau_k^x$ the translation operator in the $x$-direction acting on $ L_{\mathrm{loc}}^2(\bb R^3)$:
$$
\forall u\in L_{\mathrm{loc}}^2(\bb R^3),  \quad (\tau_k^x u)(\cdot,\bm r) = u(\cdot - k,\bm r) \quad \mbox{ for a.a. } \bm r \in \bb R^2.
$$
An operator $A$ on $L^2(\bb R^3)$ is called $\bb Z$-translation invariant in the $x$-direction if it commutes with $\tau_k^x$ for all $k \in \bb Z$. In order to decompose operators which are $\bb Z$-translation invariant in the $x$-direction, let us without loss of generality choose a unit cell $\Gamma :=[-1/2,1/2)\times \mathbb{R}^2$, and introduce the $L^p$ spaces and $H^1$ spaces of functions which are $1$-periodic in the $x$-direction: for $1\leq p\leq +\infty$,
\begin{align*}
L_{\mathrm{per},x}^p\left(\Gamma\right)&:= \left\{u \in L_{\mathrm{loc}}^p(\bb R^3)\left|\,\norm{u}_{L^p(\Gamma)}<+\infty,\, \tau_k^x u= u, \forall k\in \bb Z\right.\right\},\\
 H_{\mathrm{per},x}^1\left(\Gamma\right)&:= \left\{u \in L_{\mathrm{per},x}^2(\Gamma)\left|\, \nabla u \in \left(L_{\mathrm{per},x}^2(\Gamma)\right)^3 \right.\right\}.
\end{align*}
Let us also introduce the following constant fiber direct integral of Hilbert spaces \cite{ReeSim4}: $$L^2(\Gamma^*; \lpg):= \int_{\Gamma^*}^{\bigoplus}\lpg\, \frac{d\xi}{2\pi},$$
with the base $\Gamma^* :=[-\pi,\pi)\times \{0\}^2\equiv [-\pi,\pi)$. The partial Bloch transform $\B$ is a unitary operator from $L^2(\bb R^3)$ to $L^2(\Gamma^*;\lpg)$, defined on the dense subspace of $C_c^{\infty}(\bb R^3)$ of $L^2(\bb R^3)$:
$$\forall(x,\bm r)\in \Gamma, \, \forall \xi\in \Gamma^* ,\quad \left(\B\phi\right)_{\xi}(x,\bm r):= \sum_{k\in \bb Z}\rme^{-\rmi (x+k)\xi}\phi(x+k,\bm r).$$
Its inverse is given, for $f_\bullet = \left(f_{\xi}\right)_{\xi\in\Gamma^*}$ by
$$\forall k\in\bb Z,\,\text{ for a.a. }(x,\bm r)\in \Gamma,\quad \left(\B^{-1}f_{\bullet}\right)(x+k,\bm r) := \int_{\Gamma^*}\rme^{\rmi (k+x)\xi}f_{\xi}(x,\bm r)\,\frac{d\xi}{2\pi}. $$
The partial Bloch transform has the property that any operator $A$ on $L^2(\bb R^3)$ which commutes with $\tau_k^x$ for $k\in \bb Z$ is decomposed by $\B$: for any $A\in \mathcal{L}(L^2(\bb R^3))$ such that $\tau_k^xA=A\tau_k^x$, there exists $A_\bullet \in L^\infty(\Gamma^*;\mathcal{L}(\lpg))$ such that for all $u \in L^2(\R^3)$,
$$
(\,\B(Au))_{\xi} = A_\xi(\,\B u)_\xi \quad \mbox{ for a.a. } \xi \in \Gamma^*.
$$
We hence use the following notation for the decomposition of an operator $A$ which is $\bb Z$-translation invariant in the $x$-direction:
\[
A=\,\B^{-1}\left( \int_{\Gamma^*}^{\oplus }A_{\xi}\,\frac{d\xi}{2\pi}\right)\,\B.
\]
In addition, $\norm{A}_{\mathcal{L}(L^2(\bb R^3))} = \norm[\big]{\norm{A_\bullet}_{\mathcal{L}(\lpg)}}_{L^{\infty}(\Gamma^*)}$. In particular, if $A$ is positive and locally trace-class, then for almost all $\xi \in \Gamma^*$, $A_\xi$ is locally trace-class. The densities of these operators are related by the formula
\begin{equation}
\rho_A(\bm x) = \frac{1}{2\pi} \int_{\Gamma^*} \rho_{A_\xi}(\bm x) \, d\xi.
\label{def_density_dmatrix}
\end{equation}
If $A$ is a (not necessarily bounded) self-adjoint operator such that $\tau_k^x(A+\rmi)^{-1}=(A+\rmi)^{-1}\tau_k^x$ for all $k \in \bb Z$, then $A$ is decomposed by $\,\mathcal{U}$ (see~\cite[Theorems XIII.84 and~XIII.85]{ReeSim4}). In particular, denoting by $\Delta$ the Laplace operator acting on $L^2(\bb R^3)$, the kinetic energy operator $-  \frac12\Delta$ on $L^2(\R^3)$ is decomposed by $\B$ as follows:
\begin{equation}
- \frac12\Delta = \B^{-1}\left(\int_{\Gamma^*}- \frac12\Delta_{\xi} \,\frac{d\xi}{2\pi}\right)\B,\quad -\Delta_{\xi}= (-\rmi \nabla_{\xi})^2= (\rmi\partial_x-\xi)^2-\Delta_{\bm r}, 
\label{kineticOpDecomp}
\end{equation}
where $\Delta_{\bm r} $ is the Laplace operator acting on $L^2(\bb R^2)$.

\paragraph{Mixed Fourier transform.} The mixed Fourier transform consists of a Fourier series transform in the $x$-direction and an integral Fourier transform in the $\bm r$-direction. Denote by $\Spx(\Gamma)$ the space of functions that are $C^{\infty}$ on $\bb R^3$ and $\Gamma$-periodic, decaying faster than any power of $|\bm r|$ when $|\bm r|$ tends to infinity, as well as their derivatives. Denote by $\Spx'(\Gamma)$ the dual space of $\Spx(\Gamma)$. The mixed Fourier transform is the unitary transform $\F:\lpg \to \ell^2\left(\bb Z, L^2(\bb R^2)\right) $ defined on the dense subspace $ \Spx(\Gamma)$ of $\lpg$ by:
\begin{equation}
\forall \phi\in \Spx(\Gamma),\, \forall (n, \bm k)\in  \bb Z\times\bb R^2 ,\quad \F\phi(n, \bm k) := \frac{1}{2\pi}\int_{\Gamma}\phi(x, \bm r)\,\rme^{-\rmi (2\pi nx +\bm k\cdot\bm r)} \,dx\,d\bm r.
\label{FourierTS}
\end{equation}
Its inverse is given by, 
$$\forall \left(\psi_{n}(\bm k)\right)_{n\in \bb Z, \bm k\in \R^2} \in \ell^2\left(\bb Z; L^2(\R^2)\right),\quad \F^{-1}\psi(x, \bm r) := \frac{1}{2\pi}\sum_{n\in \bb Z}\int_{\bb R^2}\psi_n( \bm k)\,\rme^{\rmi (2\pi nx +\bm k\cdot\bm r)} \,d\bm k.$$
Note that $\F$ can be extended from $\Spx'(\Gamma) $ to $\mathscr{S}'(\bb R^3)$. One can easily see that $\F $ is an isometry from $\lpg$ to $\ell^2\left(\bb Z, L^2(\bb R^2)\right) $ in the following sense:
 \begin{equation}
\forall f, g\in \lpg,\quad \int_{\Gamma}\overline{f (x,\bm r)}g(x,\bm r)\,dx\,d\bm r  = \sum_{n\in \bb Z}\int_{\bb R^2}\overline{\F f(n,\bm k)}\F g(n,\bm k)\,d\bm k.
\label{Isometry_F}
 \end{equation}
Moreover, it is easy to verify that for $f, g\in \lpg$,
\begin{equation}
\F\left( f \star_{\Gamma}g\right) = 2\pi\left(\F f\right) \left(\F g\right),
\label{eq:convolutionF}
\end{equation}
where $\left( f \star_{\Gamma} g\right) (\bm x) := \int_{\Gamma}f(\bm x-\bm x') g(\bm x') \,d\bm x' $. As an application of the mixed Fourier transform, let us introduce a Kato--Seiler--Simon type inequality~\cite{Seiler1975} for the operator $-\rmi \nabla_{\xi}=\left(-\rmi \partial_x + \xi, -\rmi \partial_{\bm r}\right)$ for all $\xi\in\Gamma^*$, which will be repeatedly used in the proofs.
\begin{lemma}
Fix $\xi\in\Gamma^*$. Let $2\leq p \leq  +\infty$ and $f,g\in L_{\mathrm{per},x}^p(\Gamma)$. Then
\begin{equation}
\lVert f(-\rmi \nabla_{\xi})g\rVert_{\mathfrak{S}_p\left(\lpg\right)}\leq (2\pi)^{-2/p} \lVert g\rVert_{L_{\mathrm{per},x}^p(\Gamma)}\left(\sum_{n\in \bb Z} \norm{f\left(2\pi n+\xi, \cdot\right)}_{L^p(\bb R^2)}^{p}\right)^{1/p}, 
\label{KSSGamma}
\end{equation}
for any $2\leq p< \infty$ and $$\lVert f(-\rmi \nabla_{\xi})g\rVert \leq \lVert g\rVert_{L_{\mathrm{per},x}^{\infty}(\Gamma)}\sup_{n\in \bb Z} \norm{f\left(2\pi n+\xi, \cdot \right)}_{L^{\infty}(\bb R^2)},$$
when $p=+\infty$.
\label{KSSGammalemma}
\end{lemma}
The proof of this lemma can be read in Section~\ref{KssGammaLemmaSec}.
\paragraph{Periodic Green's function.}
We introduce a 3D Green's function which is $1$-periodic in the $x$-direction in the same spirit as in~\cite{blanc2000, LIEB197722}.
\begin{definition}[Periodic Green's function] For $(x,\bm r)\in \bb R^3$, the periodic Green's function is defined as
\begin{equation}
G(x, \bm r) =  -2\log (|\bm r|)  +\widetilde{G} (x,\bm r),\quad \widetilde{G} (x,\bm r):= 4 \sum_{n\geq 1}K_0\left(2\pi n |\bm r|\right)\cos\left(2\pi n x\right),
\label{defGreenFunc}
\end{equation}
where $K_0(\alpha):= \int_0^{+\infty} \rme^{ -\alpha\cosh(t)} dt$ is the modified Bessel function of the second kind. 
\end{definition}
The following lemma summarizes the properties of the periodic Green's function defined in (\ref{defGreenFunc}).
\begin{lemma} 
\begin{enumerate}
\item The Green's function $G(x, \bm r)$ defined in~\eqref{defGreenFunc} satisfies the following Poisson's equation:
\[
-\Delta G(x,\bm r) = 4\pi  \sum_{n \in \mathbb{Z}}\delta_{(x,\bm r) = (n,  0)}  \in \mathscr{S}'(\bb R^3),
\]
where $\delta_a\in \mathscr{S}'(\bb R^d)$ is the Dirac distribution at $a\in \bb R^d$. Moreover $G\in \Spx'(\Gamma)$ and 
\begin{equation}
\F(G)(n,\bm k ) = \frac{2}{4\pi^2n^2+|\bm k|^2} \in  \mathscr{S}'(\bb R^3).
\label{Poisson_Eq}
\end{equation}
\item The function $\widetilde{G}$ defined in~\eqref{defGreenFunc} belongs to $ L_{\mathrm{per},x}^p(\Gamma) $ for $1\leq p<2$ and satisfies $\int_{\Gamma} \widetilde{G} \equiv 0$. Moreover, there exist positive constants $d_1$ and $d_2$ such that $|\widetilde{G} (\cdot, \bm r)|\leq d_1\frac{\rme^{- 2\pi |\bm r|}}{\sqrt{|\bm r|}}$ when $|\bm r| \to +\infty$, and $|\widetilde{G} (\cdot, \bm r)|\leq \frac{d_2}{|\bm r|}$ when $|\bm r|\to 0$, uniformly with respect to $x$. Finally, the function $\widetilde{G}(x,\bm r) $ can also be written as 
\begin{equation}
\widetilde{G}(x,\bm r) = \sum_{n\in \bb Z}\left(\frac{1}{\sqrt{(x-n)^2+|\bm r|^2}} - \int_{-1/2}^{1/2}\frac{1}{\sqrt{(x-y-n)^2+|\bm r|^2}}\,dy\right).
\label{blancGreens}
\end{equation}
\end{enumerate}
\label{lemma1}
\end{lemma}
The proof of this lemma can be read in Section \ref{lemma1Sec}. 
\paragraph{One-body density matrices and kinetic energy space.}
In mean-field models, electronic states can be described by one-body density matrices (see e.g. \cite{Cances2008, frank2013}). Recall that for a finite system with $N$ electrons, a density matrix is a trace-class self-adjoint operator $\gamma\in \mathcal{S}(L^2(\bb R^3)) \cap \gS_1(L^2(\R^3))$ satisfying the Pauli principle $0\leq \gamma\leq 1$ and the normalization condition $\m{Tr}(\gamma)= \int_{\R^3} \rho_\gamma = N$. The kinetic energy of $\gamma$ is given by $\m{Tr}(-\frac{1}{2}\Delta \gamma):= \frac 12 \m{Tr}(|\nabla|\gamma|\nabla|)$ (see~\cite{CATTO2001687, Cances2008,CANCES2012887}).

Consider a 1D periodic system in the 3D space, where atoms are arranged periodically in the $x$-direction with unit cell $\Gamma$ and first Brillouin zone $\Gamma^*$. Since the rHF model is strictly convex in the density~\cite{Solovej1991}, we do not expect any spontaneous symmetry breaking. Therefore the electronic state of this quasi 1D system will be described by a one-body density matrix which commutes with the translations $\left\{\tau_k^x\right\}_{k\in \bb Z}$, hence is decomposed by the partial Bloch transform $\B$. In view of the decomposition~\eqref{kineticOpDecomp}, we define the following admissible set of one-body density matrices, which guarantees that the number of electrons per unit cell and the kinetic energy per unit cell are finite:
\begin{equation}
\label{kineticEnergyset}
\mathcal{P}_{\mathrm{per},x}: = \left\{\gamma \in \mathcal{S}(L^2(\mathbb{R}^3)) \left| \,0\leq\gamma\leq 1,\, \forall k \in \mathbb{Z}, \tau_k^x\gamma =\gamma \tau_k^x, \int_{\Gamma^*}\m{Tr}_{L_{\mathrm{per},x}^2}\left(\sqrt{1-\Delta_{\xi}}\,\gamma_{\xi}\,\sqrt{1-\Delta_{\xi}}\right)\,d\xi <\infty\right \}\right.,
\end{equation}
where 
\begin{equation}
\gamma =\B^{-1}\left(\int_{\Gamma^*}\gamma_{\xi}\,\frac{d\xi}{2\pi}\right)\B.
 \label{BlochDecomp}
\end{equation}
For any $\gamma\in\mathcal{P}_{\mathrm{per},x}$, it is easy to see that $\rho_{\gamma}\in L_{\mathrm{per},x}^1(\Gamma)$. Moreover, a Hoffmann-Ostenhof type inequality~\cite{PhysRevA.16.1782} can also be deduced from \cite[Equation~(4.42)]{CATTO2001687}:
\begin{equation}
\int_{\Gamma}\left|\nabla \sqrt{\rho_{\gamma}}\right|^2 \leq \int_{\Gamma^*}\m{Tr}_{L_{\mathrm{per},x}^2}\left(-\Delta_{\xi}\gamma_{\xi}\right)\frac{d\xi}{2\pi}.
\label{HOinequality}
\end{equation}
Therefore $\sqrt{\rho_{\gamma}}$ is in $H_{\mathrm{per},x}^1(\Gamma)$ hence in $ L_{\mathrm{per},x}^6(\Gamma)$ by Sobolev embeddings, so that $\rho_{\gamma}\in L_{\mathrm{per},x}^p(\Gamma)$ for $1\leq p\leq 3$ by an interpolation argument.
\paragraph{Coulomb interactions.}
Recall that the Coulomb interaction energy of charge densities $f$ and $g$ belonging to $ L^{6/5}(\bb R^3)$ can be written in real and reciprocal space as:
$$D(f,g) := \int_{\bb R^3}\int_{\bb R^3} \frac{f(\bm x)g(\bm x')}{\left|\bm x-\bm x'\right|}\,d\bm x\,d\bm x' =4\pi \int_{\bb R^3} \frac{\overline{\widehat{f}(\bm k)}\widehat{g}(\bm k)}{|\bm k|^2}.$$
In order to describe Coulomb interactions in the reciprocal space for a quasi 1D periodic system, we gather the results obtained in~\eqref{Isometry_F}, \eqref{eq:convolutionF} and~\eqref{Poisson_Eq}, and define the Coulomb interaction energy per unit cell for charge densities $f,g$ belonging to $\Spx(\Gamma)$ as:
\begin{equation}
 D_{\Gamma}(f,g):= 4\pi \sum_{n\in \bb Z}\int_{\bb R^2}\frac{\overline{\F(f)}(n,\bm k)\F(g)(n,\bm k)}{|\bm k |^2+4\pi^2 n^2}\,d\bm k.
\label{DG_Coulomb_interactions}
\end{equation}
It is easy to see that $D_{\Gamma}(\cdot,\cdot)$ is a positive definite bilinear form on $\Spx(\Gamma)$. Let us introduce the Coulomb space for the 1D periodic system in the 3D space as
\begin{equation}
\mathcal{C}_{\Gamma}:=\left\{ f\in \Spx'(\Gamma)\,\left| \, \forall n\in \bb Z,\F(f)(n,\cdot)\in L_{\mathrm{loc}}^1(\bb R^2), D_{\Gamma}(f,f)<+\infty \right.\right\},
\label{CoulombInterSpc}
\end{equation}
which is a Hilbert space endowed with the inner product $D_{\Gamma}(\cdot,\cdot)$.
\begin{remark}
\label{remark_neutral}
Remark that charge densities in $\mathcal{C}_{\Gamma}$ are neutral in some weak sense. Indeed, for $f\in \mathcal{C}_{\Gamma}\bigcap L_{\mathrm{per},x}^1(\Gamma)$, the condition $\int_{\bb R^2}\frac{\left|\F(f)(0,\bm k)\right|^2}{|\bm k |^2}\,d\bm k <+\infty$ implies that $\F(f)(0,\bm 0)=\int_{\Gamma}f(x, \bm r)\,dx\, d\bm r  =0$.
\end{remark}
\subsection{Reduced Hartree--Fock description for a quasi 1D periodic system}
\label{secRHF1Dper}
Based on the kinetic energy space and Coulomb interactions defined in the previous section, we construct here a rHF energy functional for a quasi 1D periodic system which is $1$-periodic only in the $x$-direction. We show that its ground state is given by the solution of some minimization problem. Denote by $Z\in \bb N^*$ the total nuclear charge in each unit cell. For the sake of technical reasons we model the nuclear density of a quasi 1D system by a smooth function (smeared nuclei) which is $1$-periodic in the $x$-direction
$$\mup(x, \bm r) = \sum_{n\in \bb Z}Z\,m(x-n,\bm r),$$
where $m(x,\bm r)$ is a non-negative $C_c^{\infty}(\Gamma)$ function such that $\int_{\bb R^3} m = 1$. In particular $\int_{\Gamma}\mup = Z $. 

For any trial density matrix $\gamma $ which commutes with the translations $\tau_{k}^x$ in the $x$-direction, the periodic rHF energy functional for a quasi 1D system associated with the nuclear density $\mup$ is defined as:
\begin{equation}
\mathcal{E}_{\mathrm{per},x}(\gamma):=\frac{1}{2\pi}\int_{\Gamma^*}\m{Tr}_{\lpg}\left(-\frac{1}{2}\Delta_{\xi}\gamma_{\xi}\right)d\xi+\frac{1}{2}D_{\Gamma}\left(\rho_{\gamma}-\mup,\rho_{\gamma}-\mup\right).
\label{defEnergyfunc}
\end{equation}
Let us introduce the following set of admissible density matrices for this rHF energy functional, which guarantees that the kinetic energy and Coulomb interaction energy per-unit cell are finite:
$$ \mathcal{F}_{\Gamma}:= \left\{\gamma\in \mathcal{P}_{\mathrm{per},x}\left|\, \rho_{\gamma}-\mup\in \mathcal{C}_{\Gamma}\right.\right\},$$
where $ \mathcal{P}_{\mathrm{per},x}$ is the kinetic energy space defined in~\eqref{kineticEnergyset} and $\mathcal{C}_{\Gamma}$ is the Coulomb space defined in~\eqref{CoulombInterSpc}.
\begin{lemma}
\label{F_GammaNotEmpty}
The set $\mathcal{F}_{\Gamma}$ is not empty. Moreover, for any $\gamma\in \mathcal{F}_{\Gamma}$,
\begin{equation}
\label{chargeNeutral}
\int_{\Gamma}\rho_{\gamma} = \int_{\Gamma}\mup.
\end{equation}
\end{lemma}
The proof of Lemma~\ref{F_GammaNotEmpty} relies on an explicit construct of an element in $ \mathcal{F}_{\Gamma}$, and can be read in Section~\ref{F_GammaNotEmptySec}. 

The periodic rHF ground state energy (per unit cell) of a quasi 1D system can then be written as the following minimization problem:
\begin{equation}
\label{pb1}
I_{\mathrm{per}}=\inf \left\{\mathcal{E}_{\mathrm{per},x}(\gamma)  ;\,\gamma\in \mathcal{F}_{\Gamma}\, \right\},
\end{equation}
The minimization problem similar to~\eqref{pb1} under the Thomas-Fermi type models has been studied in~\cite{blanc2000}, where the authors proved the uniqueness of the minimizers, and justified the model by a thermodynamic limit argument. For a 3D periodic crystal, the minimization problem~\eqref{pb1} has been examined in \cite{CATTO2001687}, where the authors showed the existence of minimizers and the uniqueness of the density of the minimizers. The characterization of the minimizers is given in~\cite[Theorem 1]{Cances2008}: the minimizer is unique and is a spectral projector satisfying a self-consistent equation. The following theorem provides similar results for a quasi 1D system: we show that the minimizer of \eqref{pb1} exists, and that the density of the minimizers is unique. Let us emphasize that the unit cell of a quasi 1D system is an unbounded domain $\Gamma$, hence we need to deal with the possible escaping of electrons in the $\bm r$-direction, a situation which needs not be considered for bounded unit cells as in~\cite{CATTO2001687,Cances2008}.
\begin{theorem}[Existence of rHF ground state]
The minimization problem (\ref{pb1}) admits a minimizer $\rp$ with density $\rho_{\rp}$ belonging to $ L_{\mathrm{per},x}^p(\Gamma)$ for $1\leq p\leq 3$. Besides, all the minimizers share the same density. 
\label{thmPeriodicExistence1}
\end{theorem}
The proof of Theorem~\ref{thmPeriodicExistence1} relies on a classical variational argument, and can be read in Section~\ref{thmPeriodicExistenceSec1}. 

In order to treat the junction of quasi 1D systems in Section~\ref{SecJunc}, it is useful to define and study the mean--field potential $\vp$ generated by the ground state electronic density $\rho_{\rp}$ and the nuclear density $\mup$. It is also critical to obtain some decay estimates of $\vp$ in the $\bm r$-direction. However, for $\vp$ satisfying Poisson's equation $-\Delta \vp = 4\pi (\rho_{\rp}-\mup)$, the $L^p$ integrability of $\rho_{\rp}$ obtained in Theorem~\ref{thmPeriodicExistence1} does not imply the decay of the mean-field potential $\vp$ in the $\bm r$-direction, given that the Green's function defined in~\eqref{defGreenFunc} has $\log$-growth in the $\bm r$-direction. Moreover, the uniform bound given by the energy functional~\eqref{defEnergyfunc} does not provide any $L^p$ bounds or decay property of $\vp$. In this perspective, we introduce the following assumption on $\rho_{\rp}$. Remark that this assumption, which we call  ``summability condition" is common when treating 2D Poisson's equation~\cite[Theorem 6.21]{lieb2001analysis}.
\begin{assumption}
\label{assumption_density}
The unique ground state density $\rho_{\rp}$ of the problem~\eqref{pb1} satisfies
\begin{equation}
\label{mildIntegrabilityDensity}
\int_{\Gamma}|\bm r| \rho_{\rp}(x,\bm r)\,dx\,d\bm r <+\infty.
\end{equation}
\end{assumption}
With this mild summability condition~\eqref{mildIntegrabilityDensity} on $\rho_{\rp}$, we prove in Theorem~\ref{thmPeriodicExistence2} that the highest attainable energy (Fermi level) of electrons for a quasi 1D system in its ground state is always negative. This coincides with the physical reality: the additional summability condition on the density is sufficient to guarantee that the mean-field potential tends to $0$ in the $\bm r$-direction. If the Fermi level is non-negative, electrons can escape to infinity in the $\bm r$-direction, decreasing the energy of the system, hence the system is not at ground state. Furthermore, we are able to characterize the unique minimizer as a spectral projector of the mean-field Hamiltonian. We comment on Assumption~\ref{assumption_density} in Remark~\ref{remark_fail_proof_summability}.
\begin{theorem}[Properties of the rHF ground state with summability condition on the density] Assume that Assumption~\ref{assumption_density} holds for the unique ground state density $\rho_{\rp}$ of the minimization problem~\eqref{pb1}.
\begin{enumerate}
\item \textbf{(The integrability of mean-field potential.)} 
The mean-field potential $$\vp:= (\rho_{\rp} -\mups) \star_{\Gamma}G$$ belongs to $L_{\mathrm{per},x}^p(\Gamma)$ for $2< p \leq +\infty$. Moreover, $\vp$ is continuous and tends to zero in the $\bm r$-direction. 
\item \textbf{(Spectral properties of the mean-field Hamiltonian.)} The mean-field Hamiltonian
\begin{equation}
\hp=\B^{-1}\left(\int_{\Gamma^*}\hpx\,\frac{d\xi}{2\pi}\right)\B= -\frac{1}{2}\Delta + \vp,\quad \hpx := -\frac{1}{2}\Delta_{\xi}+\vp,
\label{Hper}
\end{equation}
is a self-adjoint operator acting on $L^2(\bb R^3)$ with domain $H^2\left(\bb R^3\right)$ and form domain $H^1\left(\bb R^3\right)$. There exists $N_H\in \bb N^*$ which can be finite or infinite, and a sequence $\left\{\lambda_n(\xi)\right\}_{\xi\in \Gamma^*,\,1\leq n\leq N_H}$ such that 
$$\sigma_{\mathrm{ess}}\left(\hpx\right) = [0,+\infty),\qquad \sigma_{\mathrm{disc}}\left(\hpx\right) =\bigcup_{1\leq n \leq N_H} \lambda_n\left(\xi \right) \subset [-\norm{\vp}_{L^{\infty}}, 0).$$ 
Moreover, the following spectral decomposition holds:
\begin{equation} 
\sigma\left(\hp\right)=\sigma_{\mathrm{ess}}\left(\hp\right)= \bigcup_{\xi\in\Gamma^*}\sigma(\hpx),\quad \bigcup_{\xi\in\Gamma^*}\sigma_{\mathrm{disc}}(\hpx) \subseteq \sigma_{\mathrm{ac}}\left(\hp\right).
\label{spectral_prop_H}
\end{equation}
In particular, $[0,+\infty) \subset \sigma_{\mathrm{ess}}(\hp)$.
\item \textbf{(The Fermi level is always negative.)} The energy level counting function $$\quad F(\kappa): \kappa \mapsto \frac{1}{|\Gamma^*|}\int_{\Gamma^*} \m{Tr}_{\lpg}\left(\mathds 1_{\left(-\infty,\kappa \right]}\left(\hpx\right)\right)d\xi=\frac{1}{|\Gamma^*|}\sum_{n=1}^{N_H}\int_{\Gamma^*} \mathds 1\left(\lambda_{n}(\xi )\leq \kappa \right)d\xi $$
is continuous and non-decreasing on $ \left(-\infty, 0\right]$. The following inequality always holds: $$   N_H = F(0)\geq \int_{\Gamma} \mups , $$
which means that there are always enough negative energy levels for the electrons. Moreover, there exists $\epsilon_F< 0$ called Fermi level (chemical potential) such that $F(\epsilon_F) = \int_{\Gamma} \mups =Z $, which represents the highest attainable energy level by electrons, and can be interpreted as the Lagrange multiplier associated with the charge neutrality condition~\eqref{chargeNeutral}.
\item \textbf{(The unique minimizer is a spectral projector.)} The minimizer of the problem~\eqref{pb1} is unique and satisfies the following self-consistent equation:
\begin{equation}
\gamma_{\mathrm{per}} = \mathds{1}_{(-\infty,\epsilon_F]}(\hp) = \B^{-1}\left(\int_{\Gamma^*}\rpx\,\frac{d\xi}{2\pi}\right)\B,\quad \rpx:=  \mathds{1}_{(-\infty,\epsilon_F]}(\hpx).
\label{selfconsist1}
\end{equation} 
Furthermore, there exist positive constants $C_{\epsilon_F}$ and $\alpha_{\epsilon_F}$ which depend on the Fermi level $\epsilon_F$, such that 
\begin{equation}
\label{expodecay_density}
0\leq \rho_{\rp} (x,\bm r) \leq  C_{\epsilon_F}\rme^{-\alpha_{\epsilon_F} |\bm r|}.
\end{equation}
\end{enumerate}
\label{thmPeriodicExistence2}
\end{theorem}
The proof of Theorem~\ref{thmPeriodicExistence2} can be read in Section~\ref{thmPeriodicExistenceSec2}.
\begin{remark}
As the unit cell of the 1D system in the 3D space is an unbounded domain $\Gamma$, the decomposed mean-field Hamiltonian $\hpx$ does not have a compact resolvent, which is a significant difference compared to the situation considered in~\cite{CATTO2001687,Cances2008}.
\end{remark} 
\begin{remark}
\label{remark_fail_proof_summability}
Let us comment on Assumption~\ref{assumption_density}. Remark that the exponential decay of the density~\eqref{expodecay_density} implies the summability condition~\eqref{mildIntegrabilityDensity}. However, we were not able to directly prove~\eqref{mildIntegrabilityDensity}. This failure is mainly due to the lack of \textit{a priori} summability bounds for the density matrices in $\mathcal{F}_{\Gamma}$. One might argue that we can add the condition~\eqref{mildIntegrabilityDensity} to the definition of $\mathcal{F}_{\Gamma}$. However, the set $\mathcal{F}_{\Gamma}$ with the condition~\eqref{mildIntegrabilityDensity} is not closed for the usual weak-$*$ topology when considering a minimizing sequence of~\eqref{defEnergyfunc}. Another attempt is to use a Schauder fixed-point algorithm as in~\cite{lions1987,cances:hal-00186926} to prove that~\eqref{selfconsist1} admits a solution. The most crucial step is to guarantee that there are enough negative bound states to meet the charge neutrality constraint~\eqref{chargeNeutral}. The number of bound states is controlled by the decay rate of potentials. With exponentially decaying densities we can show that~\cite[Lemma 2.5]{blanc2000} there exists $C\in \R^+$ such that $\left| \vp(\cdot,\bm r)\right| \leq C |\bm r|^{-1}$. Nevertheless this condition is not sufficient to guarantee that the number of bound states is sufficient, as the critical decay rate for numbers of bound states to be finite or infinite is $-|\bm r|^{-2}$~\cite[Theorem XIII.6]{ReeSim4}. In other words, we do not have a uniform bound over the Fermi level $\epsilon_F$ at each fixed-point iteration. On the other hand, the summability condition~\eqref{mildIntegrabilityDensity} is a sufficient but probably not a necessary condition for the negativity of the Fermi level and the characterization of the minimizers. The main difficult is to control the decay of the mean-field potential $\vp$ in the $\bm r$-direction by just controlling the nuclear density $\mups$, given that the Green's function defined in~\eqref{defGreenFunc} has $\log$-growth in the $\bm r$-direction. Furthermore, different decay scenarios of $\vp$ in the $\bm r$-direction lead to different characterizations of the spectrum of the Hamiltonian $\hp$: if $\vp$ is bounded from below, and positive with $\log$-growth when $|\bm r|\to\infty$, one can show that the spectrum of $\hpx$ is purely discrete and the spectrum of $\hp$ has a band structure. The Fermi level of the system could be positive in this case. We are not able to prove the above statements without Assumption~\ref{assumption_density}.
\end{remark}
In order to describe the junction of quasi 1D systems, more specifically to guarantee that the Coulomb energy generated by the perturbative state is finite, the integrability of the mean-field potential provided in~Theorem~\ref{thmPeriodicExistence2} is not sufficient in view of Lemma~\ref{eta_chi_Lemma} below. In order to make use of this result, let us introduce a class of nuclear densities such that the $x$-averaged density is rotationally invariant in the $\bm r$-direction: 
$$
\mupsm(x, \bm r) = \sum_{n\in \bb Z}Z\,m_{\mathrm{s}}(x-n,\bm r),
$$
where $m_{\mathrm{s}}(x,\bm r)$ is a non-negative $C_c^{\infty}(\Gamma)$ function such that $\int_{\bb R^3} m_{\mathrm{s}} = 1$. Moreover, there exists $m_{\mathrm{sym}}(|\bm r|) \in C_c^{\infty}(\R^2)$ such that
\begin{equation}
\forall \bm r\in \R^2,\quad \int_{-1/2}^{1/2}\quad m_{\mathrm{s}}(x, \bm r)\,dx  \equiv m_{\mathrm{sym}}(|\bm r|). 
\label{absenceDipole}
\end{equation}
\begin{lemma}
\label{symmetryPotentialLemma}
Suppose that Assumption~\ref{assumption_density} holds. Under the symmetry condition~\eqref{absenceDipole} on the nuclear density $\mup$, all the results of Theorem~\ref{thmPeriodicExistence2} hold for the minimization problem~\eqref{pb1}. Besides, the mean--field potential $\vp $ belongs to $L_{\mathrm{per},x}^p(\Gamma)$ for $1<p\leq +\infty$.
\end{lemma}
The proof of this lemma can be read in Section~\ref{symmetryPotentialLemmaSec}. The nuclei of many actual materials can be modeled with a smear nuclear density satisfying the condition~\eqref{absenceDipole}: for instance nanotubes and polymers with rotational symmetry in the $\bm r$-direction.
\section{Mean--field stability for the junction of quasi 1D systems}
\label{SecJunc}
In this section, we construct a rHF model for the junction of two different quasi 1D periodic systems. The junction system is described by periodic nuclei satisfying the symmetry condition~\eqref{absenceDipole} with different periodicities and possibly different charges per unit cell, occupying separately the left and right half spaces (\textit{i.e.}, $(-\infty, 0]\times \bb R^2$ and $(0,+\infty) \times \bb R^2$), see Fig.~\ref{Junction_pic}. We do not assume any commensurability of the different periodicities. The junction system is therefore~\textit{a priori} no longer periodic, making it impossible to define the periodic rHF energy. Inspired by perturbative approaches when treating infinitely extended systems~\cite{0305-4470-38-20-014, Hainzl2005, HAINZL2005TheMA, Hainzl2009, Cances2008}, the idea is to find an appropriate reference state which is close enough to the actual one. Section \ref{mathdesJuncSec} gives a mathematical description of the junction system. Section \ref{refSec} is devoted to a rigorous construction of a reference Hamiltonian $\hc$ and a reference one-particle density matrix defined as a spectral projector of $\hc$. In Section~\ref{perturbativeSec} we construct a perturbative state, which encodes the non-linear effects due to the electron-electron interaction in the rHF approximation, and associate the ground state energy of this perturbative state to some minimization problem in Section~\ref{ResultSec}. 
\subsection{Mathematical description of the junction system}
\label{mathdesJuncSec}
Consider two quasi 1D periodic systems with periods $a_{L}>0$ and $a_{R}>0$. The unit cells are respectively denoted by $\Gamma_{L}:= [-\frac{a_{L}}{2},\frac{a_{L}}{2})\times \bb R^2$ and $\Gamma_{R}:= [-\frac{a_{R}}{2},\frac{a_{R}}{2})\times \bb R^2$ with their duals $\Gamma_{L}^*:= [-\frac{\pi}{a_{L}},\frac{\pi}{a_{L}})$ and $\Gamma_{R}^*:= [-\frac{\pi}{a_{R}},\frac{\pi}{a_{R}})$. We consider nuclear densities fulfilling the symmetry condition~\eqref{absenceDipole} and suppose that Assumption~\ref{assumption_density} holds for the ground state densities of both quasi 1D periodic systems. More precisely, let $m_{L}(x,\bm r)$ and $m_{R}(x,\bm r)$ be non-negative $C_c^{\infty}$ functions with supports respectively in $\Gamma_{L}$ and $\Gamma_{R}$ such that $\int_{\bb R^3} m_{L} = 1$ and $\int_{\bb R^3} m_{R} = 1$. Assume that there exist $m_{\mathrm{sym},L}(|\bm r|),m_{\mathrm{sym},R}(|\bm r|) \in C_c^{\infty}(\R^2)$ such that 
$$\forall \bm r\in \R^2,\quad \int_{-a_L/2}^{a_L/2}\quad m_{L}(x, \bm r)\,dx  \equiv m_{\mathrm{sym},L}(|\bm r|),\quad \int_{-a_R/2}^{a_R/2}\quad m_{R}(x, \bm r)\,dx  \equiv m_{\mathrm{sym},R}(|\bm r|). $$
Denoting by $Z_{L},Z_R\in \bb N\backslash \{0\} $ the total charges of the nuclei per unit cells, the smeared periodic nuclear densities are respectively defined as 
\begin{equation}
 \mu_{\mathrm{per},L }(x, \bm r) := \sum_{n\in \bb Z}Z_{L}\,m_{L}(x-a_L n,\bm r),\quad\mu_{\mathrm{per},R }(x, \bm r) := \sum_{n\in \bb Z}Z_{R}\,m_{R}(x-a_Rn,\bm r).
\label{nuclearDensityper}
\end{equation}
The periodic Green's functions with period $\Gamma_{L}$ and $\Gamma_{R}$ are separately defined as $$G_{a_{L}}(x,\bm r)
= a_{L}^{-1}G\left(\frac{x}{a_{L}},\bm r\right),\quad G_{a_{R}}(x,\bm r) = a_{R}^{-1}G\left(\frac{x}{a_{R}},\bm r\right),$$
where $G(\cdot)$ is the periodic Green's function defined in~\eqref{defGreenFunc}. One can easily verify that 
$$-\Delta G_{a_{L}}(x,\bm r) =4\pi  \sum_{n\in \mathbb{Z}}\delta_{(x,\bm r) = (a_L n, \bm 0)}  \in \mathscr{S}'(\bb R^3),\quad -\Delta G_{a_{R}}(x,\bm r) =4\pi  \sum_{n \in \mathbb{Z}}\delta_{(x,\bm r) = (a_R n, \bm 0)}  \in \mathscr{S}'(\bb R^3). $$
According to the results of Theorem~\ref{thmPeriodicExistence2}, the following self-consistent equations uniquely define the ground states density matrices associated with the periodic nuclear densities $\mu_{\mathrm{per},L}$ and $\mu_{\mathrm{per},R}$:  $$\gl:= \mathds 1_{(-\infty,\epsilon_{L}]}(\hl),\quad \hl:=-\frac{\Delta}{2}+\vl, \quad \vl:= \left(\rol-\mu_{\mathrm{per},L}\right)\star_{\Gamma_{L}} G_{a_{L}},$$
$$\gamma_{\mathrm{per},R}:= \mathds 1_{(-\infty,\epsilon_{R}]}(\hr),\quad \hr:=-\frac{\Delta}{2}+\vr,\quad \vr:= (\ror-\mu_{\mathrm{per},R})\star_{\Gamma_{R}} G_{a_{R}},$$
where the negative constants $\epsilon_L $ and $\epsilon_R$ are the Fermi levels of the quasi 1D systems. The junction of the quasi 1D systems are described by considering the following nuclear density configuration (see Fig.\ref{Junction_pic}):
\begin{equation}
\mu_{J}(x,\bm r):= \mathds{1}_{x\leq 0}\cdot \mu_{\mathrm{per},L}(x,\bm r) +\mathds{1}_{x >0}\cdot \mu_{\mathrm{per},R}(x,\bm r) + v(x,\bm r),
\label{junctionNuclear}
\end{equation}
where $v(x,\bm r)\in L^{6/5}(\bb R^3) $ describes how the junction switches between the underlying nuclear densities. The assumption $v\in L^{6/5}(\bb R^3) $ ensures that $D(v,v)<+\infty$. Recall that
$$ D(f,g) = \int_{\mathbb{R}^3}\int_{\mathbb{R}^3}\frac{f(x)g(y)}{|x-y|}dx\,dy= \frac{1}{4\pi}\int_{\mathbb{R}^3}\frac{\overline{\widehat{f}(k)}\widehat{g}(k)}{k^2}dk$$
describes the Coulomb interactions in the whole space.
Once one sets the nuclear configuration (\ref{junctionNuclear}), electrons are allowed to move in the 3D space. The \textit{infinite} rHF energy functional for the junction system associated with a test density matrix $\gamma_{J}$ formally reads
\begin{equation}
\mathcal{E}(\gamma_{J}) = \mathrm{Tr}\left(-\frac{1}{2}\Delta\gamma_{J}\right) + \frac{1}{2}D\left(\rho_{\gamma_{J}} - \mu_{J},\rho_{\gamma_{J}}- \mu_{J}\right).
\label{semi-infinitesystem}
\end{equation} Let us also introduce the Coulomb space $\mathcal{C}$ and its dual $\mathcal{C}'$ (Beppo-Levi space \cite{Cances2008}):
\begin{equation}
\mathcal{C} :=\left\{ \rho \in {\mathscr{S}'(\mathbb{R}^3)}\left|\, \widehat{\rho} \in L_{\m{loc}}^1(\R^3),\, D(\rho,\rho) <\infty \right.\right\},\quad \mathcal{C}' :=\left\{V\in L^6(\mathbb{R}^3)\mid \nabla V \in (L^2(\mathbb{R}^3))^3\right\}.
\label{CoulombSpace}
\end{equation}
\begin{figure}[h!]
\centering
\includegraphics[width=0.65\textwidth ]{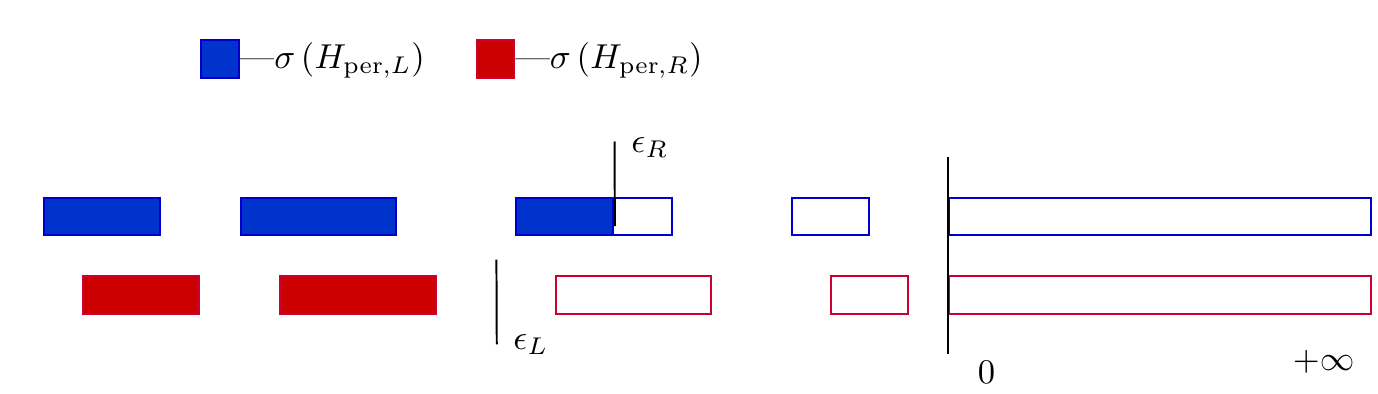}
\caption{Spectrum of $\hl$, $\hr$ in the non-equilibrium regime.}
\label{junction_1}
\end{figure}
\begin{figure}[h!]
\centering
\includegraphics[width=0.65\textwidth ]{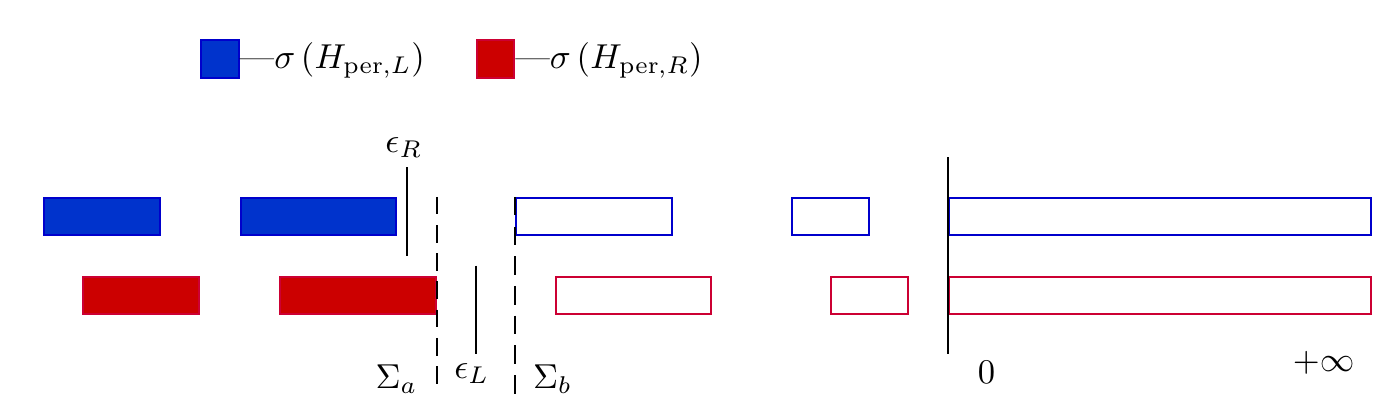}
\caption{Spectrum of $\hl$, $\hr$ in the equilibrium regime.}
\label{junction_2}
\end{figure}
Remark that the ground state energy of the junction system, if it exists, is infinite and there is no periodicity in this system, hence usual techniques which essentially consist in considering the energy per unit volume~\cite{CATTO2001687,Cances2008} are not applicable. We next define a reference system such that the difference between the junction system and the reference can be considered as a perturbation. This perturbative approach has been used in \cite{HAINZL2005TheMA, Hainzl2005, Hainzl2009, Cances2008} in various contexts. The next section is devoted to the rigorous mathematical construction of the reference state and its rHF energy functional.
\subsection{Reference state for the junction system}
\label{refSec}
In this section, we construct a reference Hamiltonian obtained by a linear combination of the periodic mean--field potentials $\vl$ and $\vr$. We prove the validity of this approach by showing that the density generated by this reference state is close to the linear combination of the periodic densities $\rol$ and $\ror$.
\paragraph{Hamiltonian of the reference state.} We introduce a class of smoothed cut-off functions. For $\bm x\in \bb R^3$, consider:
\begin{equation}
\mathcal{X} :=\left\{\chi\in C^{2}(\bb R^3)\left|\,0\leq \chi \leq 1;\, \chi(\bm x) = 1\text{ if } \bm x \in \left(-\infty, -\frac{a_{L}}{2}\right]\times \bb R^2;\,  \chi(\bm x) = 0 \text{ if } \bm x\in \left[\frac{a_{R}}{2}, +\infty\right)\times \bb R^2 \right.\right\}. 
\label{chifunc}
\end{equation}
Fix $\chi \in \mathcal{X}$, let us introduce a reference potential $$V_{\chi}:=\chi^2\vl+ (1-\chi^2)\vr.$$
We will show in Section~\ref{ResultSec} that the choice of $\chi \in \mathcal{X}$ is irrelevant. By Theorem \ref{thmPeriodicExistence2} and Lemma~\ref{symmetryPotentialLemma} we know that $V_{\chi}$ belongs to $ L_{\mathrm{loc}}^p\left(\bb R, L^p(\bb R^2)\right) $ for $1< p\leq \infty $, is continuous in all directions and tends to zero in the $\bm r$-direction.
By the Kato--Rellich theorem (see for example \cite[Theorem 9.10]{Helffer}), there exists a unique self-adjoint operator
\begin{equation}
\hc := -\frac{1}{2}\Delta + V_{\chi}
\label{Hchi}
\end{equation}
on $L^2(\bb R^3)$ with domain $H^2(\bb R^3)$ and form domain $H^1(\bb R^3)$. We next show that the essential spectrum of the reference Hamiltonian $\hc$ is the union of the essential spectra of $\hl$ and $\hr$, which implies that the reference system does not change essentially the unions of possible energy levels of quasi periodic systems, and that there are no surface states which propagate along the junction surface in the $\bm r$-direction. Note that this is \textit{a priori} not obvious as the cut-off function $\chi$ is $\bm r$-translation invariant (hence not compact), therefore scattering states may occur at the junction surface and escape to infinity in the $\bm r$-direction. Standard techniques in scattering theory to prove this statement, such as Dirichlet decoupling~\cite{DEIFT1976218, Hempel2014}, are not applicable in our situation since the junction surface is not compact.
\begin{prop}[Spectral properties of the reference state $\hc$] For any $\chi\in\mathcal{X}$, the essential spectrum of $\hc$ defined in (\ref{Hchi}) satisfies $$\sigma_{\mathrm{ess}}(\hc) = \sigma_{\mathrm{ess}}\left(\hl\right)\bigcup  \sigma_{\mathrm{ess}}\left(\hr\right).$$
In particular, $ [0,+\infty)\subset\sigma_{\mathrm{ess}}(\hc)$ and $\sigma_{\mathrm{ess}}(\hc)$ does not depend on the shape of the cut-off function $\chi\in \mathcal{X}$ defined in~\eqref{chifunc}.
\label{spectrapPpHchi}
\end{prop}
The proof relies on an explicit construction of Weyl sequences associated with $\hc$ (see Section~\ref{spectrapPpHchiSec}). Remark that Proposition~\ref{spectrapPpHchi} also implies that the reference system essentially preserves the scattering channels of the underlying quasi 1D systems, since the scattering involves the purely absolutely continuous spectrum of a Hamiltonian (see for example~\cite{Exner2007, Bruneau2016,Bruneau2016II}). However, this does not exclude the existence of embedded eigenvalues in the essential spectrum, which may cause additional scattering channels~\cite{Frank03onthe, Frank04onthe}. We prove in Corollary~\ref{corollary_1} that the results in Proposition~\ref{spectrapPpHchi} still hold for the nonlinear junction.
\paragraph{Reference state as a spectral projector.}
Before constructing the reference state, let us discuss different regimes for junction system. From Theorem~\ref{thmPeriodicExistence2} we know that the chemical potentials (Fermi levels) $\epsilon_L $ and $\epsilon_R$ are negative. Introduce the energy interval $I_{\epsilon_F}:= \left[\min(\epsilon_L, \epsilon_R), \max(\epsilon_L, \epsilon_R)\right] $. In view of Proposition~\ref{spectrapPpHchi}, assume that the essential spectrum of $\hc$ below $0$ is purely absolutely continuous, the non-equilibrium regime (Fig.~\ref{junction_1}) corresponds to $$\sigma_{\mathrm{ac}}(\hc) \bigcap I_{\epsilon_F} \neq \emptyset. $$ 
In this regime, steady state currents occur and the Landauer-B\"{u}ttiker conductance can be calculated~\cite{Bruneau2015,Bruneau2016,Bruneau2016II}. When $\mu_{\mathrm{per},L}$ and $\mu_{\mathrm{per},R}$ are identical, the junction system becomes periodic with different chemical potentials $\epsilon_L $ and $\epsilon_R$ on the left and right half lines. In this case the Thouless conductance~\cite{Bruneau2015} can be defined and it is given by $$C_T\frac{\left|\sigma_{\mathrm{ac}}(\hc) \bigcap  I_{\epsilon_F} \right|}{ \left|I_{\epsilon_F} \right|} >0,$$
for some positive constant $C_T$. However it is not the aim of this article to discuss steady state currents for non-equilibrium systems. We instead consider the equilibrium regime (see Fig.~\ref{junction_2}) with the following assumption.
\begin{assumption}
The chemical potential $\epsilon_L$ and $\epsilon_R$ are in a common spectral gap $(\Sigma_a, \Sigma_b)$ (equilibrium regime, see Fig.~\ref{junction_2}), where $\Sigma_a$ is the maximum of the filled bands of $\hl$ and $\hr$, and $\Sigma_b$ is the minimum of the unfilled bands of $\hl$ and $\hr$.
\label{as:1}
\end{assumption}
\begin{figure}[h!]
\centering
\includegraphics[width=0.6\textwidth ]{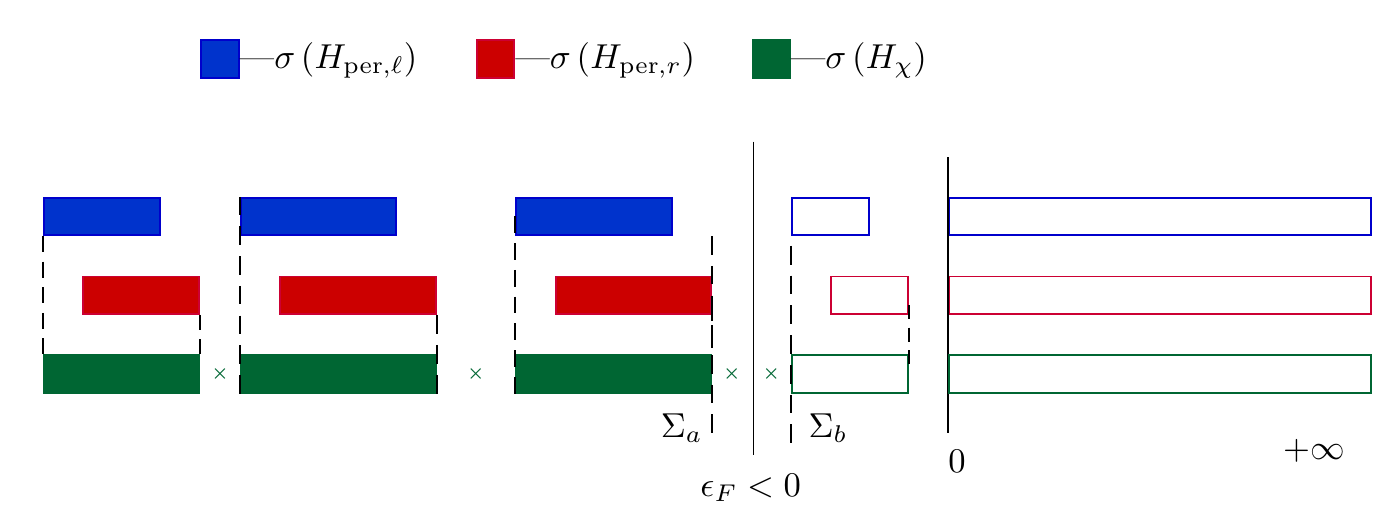}
\caption{Spectrum of $\hl$, $\hr$ and $\hc$ below $0$.}
\label{HchiSpectralFig}
\end{figure}
Assumption~\ref{as:1} guarantees that the Fermi level of the junction system lies in a spectral gap of $\hc$ in view of Proposition~\ref{spectrapPpHchi}, which is a common hypothesis~\cite{Cances2008,Hainzl2005,Hainzl2009} for 3D periodic insulating and semi-conducting systems. We make this assumption for simplicity. Remark that with approaches proposed in~\cite{energycost, frank2013,cao2018} it is possible to extend the results to metallic systems provided that the junction system is in its ground state and no steady state current occurs.

Let us without loss of generality choose the Fermi level $\epsilon_{F}=\max(\epsilon_L,\epsilon_R) = \sup I_{\epsilon_F}$ and define the reference state~$\gamma_{\chi}$ as the spectral projector associated with the states of $\hc$ below $\epsilon_F$:
\begin{equation}
 \gamma_{\chi}:=\mathds 1_{(-\infty,\epsilon_F)}(\hc).
 \label{referenceState}
\end{equation}
Remark that $\hc $ can have discrete spectrum in the gap $(\Sigma_a,\Sigma_b)$, with eigenvalues possibly accumulating at~$\Sigma_a$ and~$\Sigma_b$, and $\epsilon_{F}$ can also be an eigenvalue of $\hc$. The definition of~\eqref{referenceState} however excludes the possible bound states with energy~$\epsilon_F$. 

The following proposition shows that the density $\rho_{\chi}$ of $\gamma_\chi$ is well defined in $L_{\mathrm{loc}}^1(\bb R^3)$, and is close to the linear combination of the periodic densities $\rol$ and $\ror$, the difference decaying exponentially fast in the $x$-direction as $|x|\to\infty$. 
\begin{prop}[Exponential decay of density]
Under Assumption~\ref{as:1}, the spectral projector $\gamma_\chi$ is locally trace class, so that its density $\rho_{\chi}$ is well defined in $L_{\mathrm{loc}}^1(\bb R^3)$. Moreover,
$$\chi^2\rol+(1-\chi^2)\ror-\rho_{\chi} \in  L^p(\mathbb{R}^3) \text{  for  } 1<p\leq 2.$$ Furthermore, denote by $w_a$ the characteristic function of the unit cube centered at $a\in \bb R^3$. There exist positive constants $C$ and~$t$ such that for all $$ \alpha = (\alpha_x, 0,0) \in \bb R^3, \, \text{ with either  }\supp(w_\alpha)\subset (-\infty,a_{L}/2] \times \bb R^2 \text{ or }\, \supp(w_\alpha)\subset   [a_{R}/2,+\infty)\times \bb R^2,$$ it holds, $$ \int_{\mathbb{R}^3}\left|w_\alpha \left(\chi^2\rol+(1-\chi^2)\ror-\rho_{\chi}\right)w_\alpha\right| \leq C\rme^{-t|\alpha|}.$$
\label{propExist}
\end{prop}
The proof can be read in Section \ref{propExistsec}. 
\paragraph{Fictitious nuclear density of the reference state.}
The density $\rho_{\chi}$ associated with $\gamma_{\chi}$ is fixed once the Fermi level $\epsilon_F$ is chosen. We can therefore define a fictitious nuclear density $\mu_{\chi}$ by imposing that the total electronic density $\rho_{\chi}-\mu_{\chi}$ generates the potential $V_{\chi}$. The fictitious nuclear density $\mu_{\chi}$ is given by
\begin{equation}
-\Delta V_{\chi} = 4\pi \left(\rho_{\chi}-\mu_{\chi}\right),  \quad \mu_{\chi}:= \rho_{\chi}-\left( \chi^2(\rol-\mu_{\mathrm{per},L}) +(1-\chi^2)(\ror-\mu_{\mathrm{per},R}) + \eta_{\chi}\right),
\label{ReferencePotential2}
\end{equation}
where $\eta_\chi$ has compact support in the $x$-direction:
\begin{equation}
\eta_{\chi} := -\frac{1}{4\pi}\left(\partial_x^2(\chi^2)\left(\vl-\vr \right) +2\partial_x(\chi^2)\partial_x\left( \vl- \vr\right)\right).
\label{eta_chi_Def}
\end{equation}
Let us emphasize that the Poisson's equation~\eqref{ReferencePotential2} is defined on the whole space $\bb R^3$. 
\paragraph{The nuclear density of the junction is a fictitious nuclear density plus a perturbation.}
Once we have defined fictitious nuclear density, we can treat the difference between the real nuclear density of the junction system $\mu_{J}$ and the fictitious nuclear density $\mu_{\chi}$ as a perturbative nuclear density. By doing so we can define a finite renormalized energy with respect to the perturbative nuclear density. Note that this idea is similar to the definition of the defect state in \cite{Cances2008} for defects in crystals, and the polarization of the vacuum in the Bogoliubov--Dirac--Fock model \cite{HAINZL2005TheMA,Hainzl2005,Hainzl2009}. Introduce
\begin{equation}
\nu_{\chi} := \mu_{J}- \mu_{\chi} = \left(\mathds{1}_{x\leq  0}-\chi^2 \right)\left(\mu_{\mathrm{per},L}-\mu_{\mathrm{per},R}\right)+ \left(\chi^2\rol + (1-\chi^2)\ror-\rho_{\chi} \right)+ \eta_{\chi}+v.
\label{nuchi}
\end{equation}
In order to guarantee that the perturbative state has a finite Coulomb energy, we need $D(\nu_\chi,\nu_\chi)<+\infty$. A sufficient condition is that $\nu_{\chi}$ belongs to $L^{6/5}(\bb R^3)$. This motivates the following $L^p$-estimate on $\eta_{\chi}$.
\begin{lemma}
The function $\eta_\chi$ defined in~\eqref{eta_chi_Def} belongs to $L^p\left(\bb R^3\right)$ for $1< p < 6$.
\label{eta_chi_Lemma}
\end{lemma}
The proof can be read in Section~\ref{eta_chi_LemmaSec}. In view of Lemma~\ref{eta_chi_Lemma} and Proposition~\ref{propExist}, together with the fact that $ (\mathds{1}_{x\leq  0}-\chi^2 )\left(\mu_{\mathrm{per},L}-\mu_{\mathrm{per},R}\right)$ has compact support and $v$ belongs to $L^{6/5}\left(\bb R^3\right)$, it is easy to see that $\nu_{\chi}$ belongs to $L^{6/5}\left(\bb R^3\right)$, and hence to the Coulomb space $\mathcal{C}$ defined in~\eqref{CoulombSpace}. This means that the perturbative state generated by the nuclear density~$\nu_{\chi}$ has finite Coulomb energy. 
\begin{remark}
Remark that the integrability of $\vp$ provided by Lemma~\ref{symmetryPotentialLemma} is crucial to deduce Lemma~\ref{eta_chi_Lemma}.
\end{remark}
\subsection{Definition of the perturbative state}
\label{perturbativeSec}
In this section we define a perturbative state associated with the perturbative density $\nu_{\chi}$ following the ideas developed in~\cite{Cances2008}. We formally derive the rHF energy difference between the junction state $\gamma_{J} $ and the reference state $ \gamma_\chi$ by writing $\gamma_{J}  = \gamma_{\chi} + Q_\chi$ with $Q_{\chi}$ a trial density state. In view of~\eqref{semi-infinitesystem}, we formally have
\begin{equation}
\begin{aligned}
\mathcal{E}(\gamma_{J})-\mathcal{E}(\gamma_{\chi})&\overbrace{=}^{\mathrm{formally}} \mathrm{Tr}\left(-\frac{1}{2}\Delta\left(\gamma_{\chi} + Q_\chi\right)\right) + \frac{1}{2}D\left(\rho_{J} - \mu_{J},\rho_{J} - \mu_{J}\right)\\
&\qquad\qquad-\mathrm{Tr}\left(-\frac{1}{2}\Delta \gamma_{\chi}\right)- \frac{1}{2}D\left(\rho_{\chi} - \mu_{\chi},\rho_{\chi} - \mu_{\chi}\right)\\
&\,\, = \mathrm{Tr}\left(-\frac{1}{2}\Delta Q_{\chi}\right) + D(\rho_{\chi}-\mu_{\chi}, \rho_{Q_{\chi}})- D(\rho_{Q_{\chi}},\nu_{\chi})+ \frac{1}{2}D(\rho_{Q_{\chi}} ,\rho_{Q_{\chi}}) \\
&\qquad\qquad - D( \rho_{\chi}-\mu_{\chi}, \nu_{\chi})+\frac{1}{2}D(\nu_{\chi},\nu_{\chi})\\
& \,\,=\m{Tr}\left(\hc Q_{\chi}\right) - D(\rho_{Q_{\chi}}, \nu_{\chi})+\frac{1}{2}D(\rho_{Q_{\chi}} ,\rho_{Q_{\chi}}) - D( \rho_{\chi}-\mu_{\chi}, \nu_{\chi})+\frac{1}{2}D(\nu_{\chi},\nu_{\chi}).
\end{aligned}
\label{formalformallimit}
\end{equation}
We next give a mathematical definition of the terms in the last equality of~\eqref{formalformallimit}. We expect $Q_{\chi}$ to be a perturbation of the reference state~$\gamma_{\chi}$. More precisely, we expect $Q_{\chi}$ to be Hilbert-Schmidt. This is usually called the ``Shale-Stinespring" condition~\cite{ShaleStinespring}, see~\cite{LewinHDR, Solovej2007} for a detailed discussion. Moreover, we also expect the kinetic energy of $Q_{\chi}$ to be finite. Let $\Pi$ be an orthogonal projector on the Hilbert space $\mathfrak{H}$ such that both $\Pi$ and $\Pi^{\perp}:= 1-\Pi$ have infinite ranks. A self-adjoint compact operator $A$ on $\mathfrak{H}$ is said to be $\Pi$-trace class if $A\in \mathfrak{S}_2\left(\mathfrak{H}\right)$ and both $\Pi A\Pi$ and $\Pi^{\perp}A\Pi^{\perp}$ are in $\mathfrak{S}_1\left(\mathfrak{H}\right)$. For an operator $A$ we define its $\Pi$-trace as $$\mathrm{Tr}_{\Pi}(A):= \mathrm{Tr}\left(\Pi A\Pi\right) + \mathrm{Tr}\left(\Pi^{\perp}A\Pi^{\perp}\right),$$
and denote by $\mathfrak{S}_1^{\Pi}\left(\mathfrak{H}\right) $ the associated set of  $\Pi$-trace class operators. Since the reference state $\gamma_{\chi}$ defined in~\eqref{referenceState} is an orthogonal projector on $L^2(\bb R^3)$, we can define associated $\gamma_{\chi}$-trace class operators. For any trial density matrix $Q_{\chi}$, let us denote by $Q_{\chi}^{++}:= \gamma_{\chi}^{\perp}Q_{\chi}\gamma_{\chi}^{\perp}$ and $Q_{\chi}^{--} :=\gamma_{\chi}Q_{\chi}\gamma_{\chi}$, and introduce a Banach space of operators with finite $\gamma_{\chi}$-trace and finite kinetic energy as follows: 
\begin{align*}
\mathcal{Q}_{\chi}&:= \left\{Q_{\chi}\in \mathfrak{S}_1^{\gamma_{\chi}} (L^2(\bb R^3))\left|\, Q_{\chi}^* = Q_{\chi}, ~|\nabla |Q_{\chi} \in \mathfrak{S}_2(L^2(\bb R^3)) ,\right.\right.\\
&\qquad\qquad\qquad\qquad \left.\left. |\nabla |Q_{\chi}^{++}|\nabla |\in \mathfrak{S}_1(L^2(\bb R^3)), ~ |\nabla |Q_{\chi}^{--}|\nabla |\in \mathfrak{S}_1(L^2(\bb R^3))\right.\right\},
\end{align*}
equipped with its natural norm
$$ \left\lVert Q_{\chi} \right\lVert_{\mathcal{Q}_{\chi}}:= \lVert Q_{\chi}\rVert_{{\mathfrak{S}}_2} + \left\lVert Q_{\chi}^{++}\right\rVert_{{\mathfrak{S}}_1} +\left\lVert Q_{\chi}^{--}\right\rVert_{{\mathfrak{S}}_1} +\left\lVert| \nabla |Q_{\chi}\right\rVert_{{\mathfrak{S}}_2} +\left \lVert | \nabla |Q_{\chi}^{++}| \nabla |\right\rVert_{{\mathfrak{S}}_1} +\left\lVert | \nabla |Q_{\chi}^{--}| \nabla |\right\rVert_{{\mathfrak{S}}_1}.$$
By construction $\mathrm{Tr}_{\gamma_{\chi}}\left(Q_{\chi}\right)= \mathrm{Tr}\left(Q_{\chi}^{++}\right)+\mathrm{Tr}\left(Q_{\chi}^{--}\right)$. For $Q$ to be an admissible perturbation of the reference state $\gamma_{\chi}$, Pauli's principle requires that $0\leq \gamma_{\chi} + Q_{\chi}\leq 1$. Let us introduce the following convex set of admissible perturbative states:
$$\mathcal{K}_{\chi}:=\left\{Q_{\chi}\in \mathcal{Q}_{\chi} \mid - \gamma_{\chi}\leq Q_{\chi}\leq 1-\gamma_{\chi}\right\}.$$
Remark that $\mathcal{K}_{\chi}$ is not empty since it contains at least $0$. Note also that $\mathcal{K}_{\chi}$ is the convex hull of states in $\mathcal{Q}_{\chi}$ of the special form $\gamma- \gamma_{\chi}$, where $\gamma$ is an orthogonal projector~\cite{Cances2008}. Furthermore, for any $Q_{\chi}\in\mathcal{K}_{\chi}$ a simple algebraic calculation shows that $$Q_{\chi}^{++}\geq 0,\quad Q_{\chi}^{--}\leq 0,\quad 0\leq Q_{\chi}^2\leq Q_{\chi}^{++}-Q_{\chi}^{--}.$$ 
As mentioned in the previous section, the Fermi level $\epsilon_F$ can be an eigenvalue of $\hc$. Consider $N\in \bb N^*$ such that $\epsilon_F \in (\Sigma_{N,\chi},\Sigma_{N+1,\chi}]$, where $\Sigma_{N,\chi} < \Sigma_{N+1,\chi} $ are two eigenvalues of $\hc$ in the gap $(\Sigma_a,\Sigma_b)$, and let $\Sigma_{N,\chi} =\Sigma_a$ and $\Sigma_{N+1,\chi} = \Sigma_b$ whenever there is no such element. For any $\kappa \in (\Sigma_{N,\chi}, \epsilon_F)$, let us introduce the following rHF kinetic energy of a state $Q_{\chi}\in \mathcal{Q}_{\chi}$:
\[
\mathrm{Tr}_{\gamma_{\chi}}\left(\hc Q_{\chi}\right) := \mathrm{Tr}\left(\left|\hc- \kappa\right|^{1/2}\left(Q_{\chi}^{++}-Q_{\chi}^{--}\right)\left|\hc-\kappa\right|^{1/2}\right) + \kappa \mathrm{Tr}_{\gamma_{\chi}}\left(Q_{\chi}\right) .
 \]
By~\cite[Corollary 1]{Cances2008}, the above expression is independent of $\kappa \in  (\Sigma_{N,\chi}, \epsilon_F)$.
In view of the last line of~\eqref{formalformallimit} we introduce the following minimization problem 
\begin{equation}
E_{\kappa,\chi} = \inf_{Q_{\chi} \in \mathcal{K}_{\chi}}\left\{\mathcal{E}_{\chi}\left(Q_{\chi}\right)-\kappa\mathrm{Tr}_{\gamma_{\chi}}\left(Q_{\chi}\right)\right\},
\label{minimizationQ}
\end{equation}
where 
\begin{equation}
\mathcal{E}_{\chi}\left(Q_{\chi}\right):= \mathrm{Tr}_{\gamma_{\chi}}\left(\hc Q_{\chi}\right) - D\left(\rho_{Q_{\chi}},\nu_{\chi}\right)+  \frac{1}{2}D\left(\rho_{Q_{\chi}},\rho_{Q_{\chi}}\right).
\label{energyFunctional}
\end{equation}
\subsection{Properties of the junction system}
\label{ResultSec}
The following result shows that the minimization problem~\eqref{minimizationQ} is well posed and admits minimizers. 
\begin{prop}(\textbf{Existence of the perturbative ground state}) Assume that Assumption~\ref{as:1} holds. Then there exist minimizers for the problem~\eqref{minimizationQ}. There may be several minimizers, but they all share the same density. Moreover, any minimizer $\overline{Q}_{\chi}$ of (\ref{minimizationQ}) satisfies the following self-consistent equation:
\begin{equation}
\left\{ 
\begin{aligned}
\overline{Q}_{\chi} &= \mathds 1_{(-\infty, \epsilon_F)}(H_{\overline{Q}_{\chi}}) -\gamma_{\chi} +\delta_{\chi} ,\\
H_{\overline{Q}_{\chi}} &= \hc  + (\rho_{\overline{Q}_{\chi}}-\nu_{\chi})\star |\cdot|^{-1} ,
\end{aligned} 
\right.
\label{selfEquationQbar}
\end{equation}
where $\delta_{\chi}$ is a finite-rank self-adjoint operator satisfying $0\leq \delta_{\chi }\leq 1$ and $\mathrm{Ran}(\delta_{\chi}) \subseteq \mathrm{Ker}(H_{\overline{Q}_{\chi}}-\epsilon_F)$.
\label{PropExistMinimizer}
\end{prop}
The proof is a direct adaptation of several results obtained in~\cite{Cances2008}, see a short summary in Section~\ref{SecPropExistMinimizer} for completeness. Remark that $ (\rho_{\overline{Q}_{\chi}}-\nu_{\chi})\star |\cdot|^{-1}\in L^6(\R^3)$ by~\cite[Lemma 16]{CANCES2012887}, therefore $(1-\Delta)^{-1}(\rho_{\overline{Q}_{\chi}}-\nu_{\chi})\star |\cdot|^{-1}$ belongs to $\mathfrak{S}_6$ by the Kato--Seiler--Simon inequality~\eqref{KSS}, hence $(\rho_{\overline{Q}_{\chi}}-\nu_{\chi})\star |\cdot|^{-1}$ is $-\Delta$-compact hence $\hc$-compact by the boundedness of $V_{\chi}$, leaving the essential spectrum unchanged. Therefore in view of Proposition~\ref{spectrapPpHchi}, the following corollary holds.
\begin{corollary}
\label{corollary_1}
For any $\chi\in\mathcal{X}$, and $H_{\overline{Q}_{\chi}}$ solution of~\eqref{selfEquationQbar}, it holds $$\sigma_{\mathrm{ess}}\left(H_{\overline{Q}_{\chi}}\right)= \sigma_{\mathrm{ess}}\left(\hl\right)\bigcup  \sigma_{\mathrm{ess}}\left(\hr\right), \quad \sigma_{\mathrm{ess}}\left(H_{\overline{Q}_{\chi}}\right)\bigcap \,(-\infty, 0] \subseteq \sigma_{\mathrm{ac}}\left(H_{\overline{Q}_{\chi}}\right).$$
In particular, $ [0,+\infty)\subset\sigma_{\mathrm{ess}}\left(H_{\overline{Q}_{\chi}}\right)$ and $\sigma_{\mathrm{ess}}\left(H_{\overline{Q}_{\chi}}\right)$ does not depend on the shape of the cut-off function $\chi\in \mathcal{X}$ defined in~\eqref{chifunc}.
\end{corollary}
The result of Proposition~\ref{PropExistMinimizer} can be interpreted as follows: given a cut-off function $\chi$ belonging to the class $\mathcal{X}$ defined in~\eqref{chifunc}, we can construct a reference state $\gamma_{\chi} $ and a perturbative ground state $\overline{Q}_{\chi}$, the sum of which forms the ground state of the junction system. However it is artificial to introduce cut-off functions $\chi$ since there are infinitely many possible choices. In view of~\eqref{junctionNuclear}, the ground state of the junction system should not depend on the choice of cut-off functions. The following theorem shows that the electronic density of the junction system is indeed independent of the choice of the cut-off function $\chi$.
\begin{theorem}[\textbf{Independence of the reference state and uniqueness of ground state density}]
The ground state density of the junction system with nuclear density defined in~\eqref{junctionNuclear} under the rHF description is independent of the choice of the cut-off function $\chi\in \mathcal{X}$, \textit{i.e.}, the total electronic density $\rho_{J} = \rho_{\chi} + \rho_{Q_{\chi}}$ is independent of $\chi$, where $\rho_{\chi}$ is the density associated with the spectral projector $\gamma_{\chi}$ defined in~\eqref{referenceState}, and $\rho_{Q_{\chi}}$ is the unique density associated with the solution $Q_{\chi}$ of the minimization problem~\eqref{selfEquationQbar}.
\label{thm2}
\end{theorem}
The proof can be read in Section \ref{thm2Sec}. Theorem~\ref{thm2} and Proposition~\ref{PropExistMinimizer} together imply that 
\begin{corollary}
\label{junction_gs}
The ground state of the junction system~\eqref{junctionNuclear} is of the form
$$\mathds 1_{(-\infty, \epsilon_F)}\left(\hc  + (\rho_{\overline{Q}_{\chi}}-\nu_{\chi})\star |\cdot|^{-1} \right) + \delta_\chi, \quad 0\leq \delta_{\chi} \leq 1,\quad \mathrm{Ran}(\delta_\chi) \subseteq \mathrm{Ker}\left(\hc  + (\rho_{\overline{Q}_{\chi}}-\nu_{\chi})\star |\cdot|^{-1}-\epsilon_F\right),$$ 
and its density is independent of the choice of $\chi$.
\end{corollary}
Remark that an extension to junctions of 2D materials may be done by similar constructions as above, see~\cite{cao} for more details.
\section*{Acknowledgements}
I would like to express my deep gratitude to \'{E}ric Canc\`{e}s and Gabriel Stoltz for many useful discussions and advice for the article, as well as their critical readings of the manuscript.
\pagebreak
\section{Proofs of the results}
In order to simplify the notation, in Section~\ref{KssGammaLemmaSec} to~\ref{eta_chi_LemmaSec} when treating the quasi 1D periodic system we denote by $\mathfrak{S}_p$ the Schattern class $\mathfrak{S}_p(\lpg)$ for $1\leq p \leq +\infty$. Unless otherwise specified, starting from Section~\ref{spectrapPpHchiSec} we use $\mathfrak{S}_p$ instead of $\mathfrak{S}_p(L^2(\bb R^3))$ for the proofs of the junction system. 

First of all, let us recall the following Kato--Seiler--Simon (KSS) inequality:
\begin{lemma}(\text{\cite[Lemma 2.1]{Seiler1975}})
Let $2\leq p\leq  \infty$. For $g,f$ belonging to $L^p(\bb R^3)$, the following inequality holds:
\begin{equation}
\lVert f(-\rmi \nabla)g(x)\rVert_{\mathfrak{S}_p\left(L^2(\bb R^3)\right)}\leq (2\pi)^{-3/p}\lVert g\rVert_{L^p(\bb R^3)}\lVert f\rVert_{L^p(\bb R^3)}.
\label{KSS}
\end{equation}
\end{lemma}
\subsection{Proof of Lemma~\ref{KSSGammalemma}}
\label{KssGammaLemmaSec}
The proof is an easy adaptation of the proof of the classical Kato--Seiler--Simon inequality~\eqref{KSS} by replacing the Fourier transform with the mixed Fourier transform $\F$. Let us prove separately~\eqref{KSSGamma} for $p=2$ and $p=+\infty$. The conclusion then follows by an interpolation argument. We use the following kernel representation during the proofs. For $\bm x = (x,\bm r)$ and $\bm y = (y, \bm r')$ belonging to $\Gamma$, symbolic calculus shows that the Schwartz kernel $K_{f,\xi}\left((x,\bm r),(y, \bm r')\right)$ of the operator $f(-\rmi \nabla_{\xi})$ acting on $\lpg$ formally reads
\begin{equation}
\begin{aligned}
K_{f,\xi} \left((x,\bm r),(y, \bm r')\right) &= \frac{1}{2\pi}\left(\F^{-1}\circ\tau_{-\xi}^x f\right)\left((x-y),(\bm r - \bm r')\right)\\
&=\frac{1}{4\pi^2}\sum_{n\in\bb Z}\int_{\bb R^2} \rme^{\rmi\left( 2\pi n(x-y)+ \bm k\cdot(\bm r-\bm r')\right)} f\left(2\pi n +\xi, \bm k\right)\, d\bm k.
\end{aligned}
\label{KernelRepresentation}
\end{equation}
Let $p=2$. In view of the isometry identity~\eqref{Isometry_F}, the convolution equality~\eqref{eq:convolutionF} and the kernel representation~\eqref{KernelRepresentation}, the following estimate holds
\begin{align*}
\norm{f(-\rmi \nabla_{\xi})&g}_{\mathfrak{S}_2}^2  = \frac{1}{4\pi^2}\int_{\Gamma\times \Gamma} \left|\left(\F^{-1}\circ\tau_{-\xi}^x f\right)(\bm x -\bm y)g(\bm y)\right|^2 \, d\bm x \,d\bm y\\
&\leq \frac{1}{4\pi^2}\int_{\Gamma} \left| g(\bm y)\right|^2\left(\int_{\Gamma}\left|\left(\F^{-1}\circ\tau_{-\xi}^x f\right)(\bm x -\bm y)\right|^2 \, d\bm x \right)\,d\bm y\\
&= \frac{1}{4\pi^2}\int_{\Gamma} \left| g(\bm y)\right|^2\sum_{n\in \bb Z}\left(\int_{\bb R^2}\left|f\left(2\pi n+\xi, \bm k\right) \right|^2\,d\bm k\right)\,d\bm y\\
& =\frac{1}{4\pi^2}\norm{g}_{\lpg}^2\sum_{n\in \bb Z}\norm{f\left( 2\pi n +\xi, \cdot\right)}_{L^2(\bb R^2)}^2,
\end{align*}
which proves~\eqref{KSSGamma} for $p=2$. For $p=+\infty$, it suffices to prove that for any $\phi \in \lpg $, 
$$\norm{f(-\rmi \nabla_{\xi})g\phi}_{\lpg} \leq \norm{g}_{L_{\mathrm{per},x}^{\infty}(\Gamma)}\sup_{n\in \bb Z} \norm{f\left((2\pi n+\xi, \cdot)\right)}_{L^{\infty}(\bb R^2)} \norm{\phi}_{\lpg}. $$
By arguments similar to those used when $p=2$, we obtain by the isometry~\eqref{Isometry_F} that
\begin{align*}
\norm{f(-\rmi \nabla_{\xi})&g\phi}_{\lpg}^2=\sum_{n\in \bb Z}\int_{\bb R^2}\left|\F\left(g\phi\right)(n, \bm k)\right|^2\left|f\left(\left(2\pi n +\xi\right)^2+\bm k^2\right)\right|^2 \,d\bm k \\
&\leq \sup_{n\in \bb Z} \norm{f\left((2\pi n+\xi, \cdot)\right)}_{L^{\infty}(\bb R^2)}^2 \norm{g\phi}_{\lpg}^2 \\
&\leq   \lVert g\rVert_{L_{\mathrm{per},x}^{\infty}(\Gamma)}^2\sup_{n\in \bb Z} \norm{f\left((2\pi n+\xi, \cdot)\right)}_{L^{\infty}(\bb R^2)}^2  \norm{\phi}_{\lpg}^2,
\end{align*}
which is~\eqref{KSSGamma} for $ p=+\infty$. Therefore, following the same interpolation arguments as in~\cite[Lemma 2.1]{Seiler1975} we obtain~\eqref{KSSGamma} for $2\leq p\leq +\infty$.

\subsection{Proof of Lemma \ref{lemma1}}
\label{lemma1Sec}
For $n\in \bb Z$, let us consider the 2D equation:
\[
-\Delta_{\bm r} G_n + 4\pi^2 n^2 G_n  = 2\pi \delta_{\bm r = 0}\quad \text{in   } \mathscr{S}'(\mathbb{R}^2). 
\]
It is well known (see for example~\cite{lieb2001analysis, Lahbabi2014}) that the solution of the above equation is 
\[
G_n(|\bm r|) =\left\{
\begin{aligned}
&-\log (|\bm r|), \quad n\equiv 0,\\
&K_0\left(2\pi |n||\bm r|\right),\quad  |n|\geq 1,
\end{aligned}
\right.
\]
where $K_0(\alpha):= \int_0^{+\infty} \rme^{ -\alpha\cosh(t)}\,dt$ is the modified Bessel function of the second kind. Therefore the Green's function $G(x,\bm r)$ defined in~\eqref{defGreenFunc} can be rewritten as 
\begin{equation}
\label{proof:Greenfunc}
G(x,\bm r) =2 \sum_{n\in \mathbb{Z}}\rme^{2\rmi \pi n x}G_n(\bm r) \in\Spx'(\Gamma).
\end{equation}
Applying the Laplacian operator to both sides, 
\[
-\Delta G(x, \bm r) =4\pi  \sum_{n\in \mathbb{Z} }\rme^{2\rmi \pi n x}\delta_{\bm r =0}  \in\Spx'(\Gamma).
\]
On the other hand, by the Poisson summation formula $ \sum_{n \in \mathbb{Z}}\delta_{x=n} =\sum_{n\in \mathbb{Z}}\rme^{2\rmi \pi n x}  \in \mathscr{S}'(\bb R)$, we conclude that the Green's function $G(x,\bm r)$ defined in~\eqref{defGreenFunc} satisfies $$-\Delta G(x, \bm r)  = 4\pi  \sum_{n\in \mathbb{Z} }\delta_{(x, \bm r )=(n, 0)} .$$
Taking the Fourier transform $\F$ on both sides of~\eqref{proof:Greenfunc} we obtain that
$$\F G (n,\bm k ) = \frac{2}{4\pi^2n^2+|\bm k|^2} \in \mathscr{S}'(\bb R^3).$$
Let us now give some estimates on $\widetilde{G}$ defined in~\eqref{defGreenFunc}. Recall that there exist two positive constants $C_0$ and $C_1$ such that~\cite{Duffin1971}
\[
0\leq K_0(\alpha)\leq  \left\{
\begin{aligned}
C_0 \left|\log(\alpha)\right|,\qquad &\text{ when } \alpha \leq 2\pi, \\
C_{1} \rme^{-\alpha}(\pi/2\alpha)^{1/2}, &\text{ when } \alpha >  2\pi.
\end{aligned}
\right.
\]
For $|\bm r|> 1$, it holds that
\begin{equation}
\begin{aligned}
\left|\widetilde{G} (x,\bm r)\right| &\leq 2C_1 \sum_{n =1}^{+\infty} \frac{\rme^{-2\pi n |\bm r|}}{\sqrt{n|\bm r|}}\leq \frac{2C_1}{\sqrt{|\bm r|}} \frac{\rme^{- 2\pi |\bm r|}}{1-\rme^{-2\pi |\bm r|}} \leq \frac{2C_1}{1-\rme^{-2\pi }} \frac{\rme^{- 2\pi |\bm r|}}{\sqrt{|\bm r|}}.
\end{aligned}
\label{estimateGreenFunc1}
\end{equation}
For $|\bm r| \leq 1$ fixed, there exists $N \geq 1$ such that $N \leq \frac{1}{ |\bm r|} < N+1$. In particular, for $n >N+1$ we have $2\pi n|\bm r| >2\pi$. There exists therefore a positive constant $C$ such that 
\begin{equation}
\begin{aligned}
\left|\widetilde{G} (x, \bm r)\right| &\leq 4 C_{0} \left|\sum_{n=1}^N \log\left(2\pi n |\bm r|\right)\right| + 2C_{1} \sum_{n=N+1}^{\infty} \frac{ \rme^{-2\pi n |\bm r|}}{ \sqrt{n|\bm r|}} \leq  4C_{0} \left|\int_1^{\frac{1}{ |\bm r|}} \log\left(2\pi t |\bm r|\right)dt \right| + 2C_{1}\int_{2\pi}^{\infty} \rme^{-t}\,dt \\
&=  \frac{4C_{0}}{ |\bm r|}\left| \log(2\pi)- 1 -  |\bm r|  + |\bm r|\log( 2 \pi |\bm r|)\right|+  2C_{1}\rme^{-2\pi}\leq  \frac{C}{|\bm r|}.
\end{aligned}
\end{equation}
Together with~\eqref{estimateGreenFunc1} we deduce that $\widetilde{G} (x, \bm r)\in L_{\mathrm{per},x}^p(\Gamma)$ for $ 1\leq p <2$. Note that for all $ \bm r \in \bb R^2\backslash{\{ 0\}}$, it holds $\int_{-1/2}^{1/2}\overline{G}(x,\bm r)\,dx \equiv 0$. Consider, for $\bm r\neq 0$, $$\overline{G}(x,\bm r) =  \sum_{n\in \bb Z}\left(\frac{1}{\sqrt{(x-n)^2+|\bm r|^2}} - \int_{-1/2}^{1/2}\frac{1}{\sqrt{(x-y-n)^2+|\bm r|^2}}\,dy\right).$$
From~\cite[Equation (1.8)]{blanc2000}, $$-\Delta \left(\overline{G}(x,\bm r) -2\log\left(\left|\bm r \right|\right)\right)= 4\pi  \sum_{k \in \mathbb{Z}}\delta_{(x,\bm r)=(k, 0)} \quad \in \mathscr{S}'(\bb R^3),$$
with $\overline{G}(x,\bm r) = \mathcal{O}(\frac{1}{|\bm r|})$ when $|\bm r| \to \infty$ by~\cite[Lemma 2.2]{blanc2000}. Denoting by $u(x,\bm r)= \widetilde{G}(x,\bm r)-\overline{G}(x,\bm r)$ we therefore obtain that  $-\Delta u(x,\bm r)\equiv 0$. As $u(x,\bm r)$ belongs to $L_{\mathrm{loc}}^1(\bb R^3)$, by Weyl's lemma for the Laplace equation we obtain that $u(x,\bm r)$ is $C^{\infty}(\bb R^3)$. On the other hand, by the decay properties of $\widetilde{G}$ and $\overline{G}$, we deduce that $\left|u(\cdot,\bm r)\right| \to 0$ when $|\bm r| \to \infty$ uniformly in $x$, hence by the maximum modulus principle for harmonic functions we can conclude that $u\equiv 0$, hence $\widetilde{G}(x,\bm r)=\overline{G}(x,\bm r)$.
\subsection{Proof of Lemma~\ref{F_GammaNotEmpty}}
\label{F_GammaNotEmptySec}
We prove this lemma by an explicit construction of a density matrix belonging to $\mathcal{F}_{\Gamma} $.
Consider a cut-off function $\varrho \in C_c^{\infty}(\Gamma)$ such that $0\leq \varrho \leq 1$ and $\int_{\Gamma} \varrho^2 = 1$, and define $\varrho_{\mathrm{per}} = \sum_{n\in \bb N}\rho(\cdot-n)$. Let $\omega\geq 0$ be a parameter to be made precise later. Define 
\begin{equation}
\gamma_{\omega} =\B^{-1}\left(\int_{\Gamma^*}\gamma_{\omega,\xi}\,\frac{d\xi}{2\pi}\right)\B,\quad \gamma_{\omega,\xi}=A_{\omega,\xi}A_{\omega,\xi}^*,\quad A_{\omega,\xi}:=
\mathds 1_{[0,\omega]}\left( -\Delta_{\xi}\right)\varrho_{\mathrm{per}} .
\label{test_state_F}
\end{equation}
It is easy to see that $0\leq \gamma_{\omega}\leq 1$, and that $\tau_{k}^x\gamma_{\omega}  =\gamma_{\omega}\tau_{k}^x$ for all $k\in \bb Z$ by construction. Let us prove that the kinetic energy per unit cell of $\gamma_{\omega}$ is finite. Denote by $F_{\xi}(n,\bm k)=  \sqrt{\left(2\pi n +\xi\right)^2+\bm k^2}$. By the Kato--Seiler--Simon type inequality~\eqref{KSSGamma},
\begin{align*}
\int_{\Gamma^*}&\m{Tr}_{L_{\mathrm{per},x}^2}\left(\sqrt{1-\Delta_{\xi}}\,\gamma_{\omega,\xi}\,\sqrt{1-\Delta_{\xi}}\right)d\xi  = \int_{\Gamma^*}\norm{\sqrt{1-\Delta_{\xi}} A_{\omega,\xi}}_{\mathfrak{S}_2}^2 \,d\xi \\
&\leq  \frac{1}{4\pi^2}\int_{\Gamma^*} \norm{\varrho_{\mathrm{per}} }_{\lpg}^2\left(\sum_{n\in\bb Z}\bnorm{\sqrt{1+F_{\xi}^2(n,\cdot)}\mathds 1_{[0,\omega]}\left( F_{\xi}^2(n,\cdot)\right) }_{L^2(\bb R^2)}^2\right)d\xi < +\infty.
\end{align*}
The last estimate follows by the condition $\left(2\pi n +\xi\right)^2+\bm k^2 \leq \omega $ implies that the sum on $n$ is finite and the integration on $\bm k$ occurs in a compact domain. Hence $\gamma_{\omega} $ belongs to $\mathcal{P}_{\mathrm{per},x}$.
Let us now show that there exists $\omega_*\geq 0$ such that $\rho_{\gamma_{\omega_*}}-\mups\in \mathcal{C}_{\Gamma}$. It is easy to see that the density $\rho_{\gamma_{\omega} }$ is smooth and compactly supported in $\Gamma$ by definition. Moreover, in view of the kernel representation~\eqref{KernelRepresentation}, the kernel $K_{A_{\omega,\xi}}$ of the operator $A_{\omega,\xi}$ is $$K_{A_{\omega,\xi}} \left((x,\bm r),(y, \bm r')\right) = \frac{1}{4\pi^2}\sum_{n\in\bb Z}\int_{\bb R^2} \rme^{\rmi\left( 2\pi n(x-y)+ \bm k\cdot(\bm r-\bm r')\right)} \mathds 1_{[0,\omega]}\left(F_{\xi}^2(n,\bm k)\right) \varrho_{\mathrm{per}}(y, \bm r')\, d\bm k.$$
Remark that the non-negative function $\omega \mapsto \int_{\Gamma^*} \left(\sum_{n\in \bb Z} \int_{\bb R^2}\mathds 1_{[0,\omega]}\left(F_{\xi}^2(n,\bm k)\right)\,d\bm k\right)d\xi$ is monotonic non-decreasing in $\omega$, equals $0$ when $\omega = 0$ and tends to $+\infty$ when $\omega \to +\infty$. Hence there exists $\omega_* > 0$ such that
\begin{equation}
\begin{aligned}
&\int_{\Gamma}\rho_{\gamma_{\omega_*}} = \frac{1}{2\pi}\int_{\Gamma^*} \norm{ A_{\omega_*,\xi}}_{\mathfrak{S}_2}^2\,d\xi =  \frac{1}{2\pi}\int_{\Gamma^*} \norm{K_{A_{\omega,\xi}}}_{L_{\m{per},x}^2(\Gamma\times\Gamma)}^2 \,d\xi\\
&=\frac{1}{(2\pi)^3} \int_{\Gamma^*} \int_{\Gamma}\int_{\bb R^2}\sum_{n\in\bb Z} \mathds 1_{[0,\omega_*]}\left(F_{\xi}^2(n,\bm k)\right) \varrho_{\mathrm{per}}^2(y, \bm r)\,dy\,d\bm r\, d\bm k \,d\xi\\
&= \frac{1}{(2\pi)^3}\int_{\Gamma^*} \left(\sum_{n\in \bb Z} \int_{\bb R^2}\mathds 1_{[0,\omega_*]}\left(F_{\xi}^2(n,\bm k)\right)\,d\bm k\right)d\xi=\int_{\Gamma}\mups >0.
\end{aligned}
\label{decayConstuction}
\end{equation}
This condition is equivalent to $\F(\rho_{\gamma_{\omega_*}} -\mup )(0,\bm 0) = 0$. As $\bm k \mapsto \F(\rho_{\gamma_{\omega_*}} -\mup )(0,\bm k)$ is $C^1(\bb R^2)$ and bounded, the function $\bm k \mapsto |\bm k|^{-1}\F(\rho_{\gamma_{\omega_*}} -\mups )(0,\bm k) $ is in $ L_{\mathrm{loc}}^2(\bb R^2)$. In view of this, there exists a positive constant $C$ such that
\begin{equation}
\begin{aligned}
&\sum_{ n\in \bb Z}\int_{\bb R^2}  \frac{\left|\F(\rho_{\gamma_{\omega_*}} -\mups)(n,\bm k)\right|^2}{|\bm k |^2+4\pi^2 n^2}\,d\bm k 
\leq \int_{|\bm k|\leq 2\pi  }\frac{\left|\F(\rho_{\gamma_{\omega_*}} -\mups)(0,\bm k)\right|^2}{|\bm k|^2}\,d\bm k \\
&+\frac{1}{4\pi^2 }\left(\int_{|\bm k|>2\pi }\left|\F(\rho_{\gamma_{\omega_*}} -\mups)(0,\bm k)\right|^2\,d\bm k +\sum_{n\in \bb Z\backslash\{0\}}\int_{\bb R^2}\left|\F(\rho_{\gamma_{\omega_*}} -\mups)(n,\bm k)\right|^2\,d\bm k\right) \\
&\leq  C+\frac{1}{4\pi^2 }\int_{\Gamma}\left|\rho_{\gamma_{\omega_*}} -\mups\right|^2< +\infty.
\end{aligned}
\label{FourierTTestState}
\end{equation}
In view of the definition of the Coulomb energy~\eqref{DG_Coulomb_interactions}, we can therefore conclude that $$D_{\Gamma}(\rho_{\gamma_{\omega_*}} -\mups,\rho_{\gamma_{\omega_*}} -\mups)  < +\infty.$$
This concludes the proof that the state $\gamma_{\omega_*}\in \mathcal{F}_{\Gamma}$. Hence $\mathcal{F}_{\Gamma}$ is not empty. As any density $\rho_{\gamma}$ associated with $\gamma \in \mathcal{P}_{\mathrm{per},x}$ is integrable, we can conclude that ~\eqref{chargeNeutral} holds in view of Remark~\ref{remark_neutral}.
\subsection{Proof of Theorems~\ref{thmPeriodicExistence1}}
\label{thmPeriodicExistenceSec1}
Let us start by giving a convenient equivalent formulation of the minimization problems~\eqref{pb1}. The operator $-\frac{1}{2}\Delta_\xi$ is not invertible, but the operator $-\frac{1}{2}\Delta_{\xi}-\kappa$ is positive definite and $\left|-\frac{1}{2}\Delta_{\xi} - \kappa \right|^{-1} $ is bounded for any $\kappa <0$. Therefore in view of the charge neutrality constraint~\eqref{chargeNeutral}, we rewrite the periodic rHF energy functional~\eqref{defEnergyfunc} as follows:
\begin{align*}
\forall \gamma\in \mathcal{F}_{\Gamma},\quad \mathcal{E}_{\mathrm{per},x}(\gamma)&=\frac{1}{2\pi}\int_{\Gamma^*}\m{Tr}_{\lpg}\left(-\frac{1}{2}\Delta_{\xi}\gamma_{\xi}\right)d\xi+\frac{1}{2}D_{\Gamma}\left(\rho_{\gamma}-\mup,\rho_{\gamma}-\mup\right) \\
&=\overline{\cl{E}_{\mathrm{per},x,\kappa}}(\gamma)+\frac{\kappa}{2\pi}\int_{\Gamma^*}\m{Tr}_{\lpg}\left(\gamma_{\xi}\right)d\xi =\overline{\cl{E}_{\mathrm{per},x,\kappa}}(\gamma)+ \kappa\int_{\Gamma}\mup,
\end{align*}
with $$\overline{\cl{E}_{\mathrm{per},x,\kappa}}(\gamma):=\frac{1}{2\pi}\int_{\Gamma^*}\m{Tr}_{\lpg}\left(\left|-\frac{1}{2}\Delta_{\xi} - \kappa\right|^{1/2}\gamma_{\xi}\left|-\frac{1}{2}\Delta_{\xi} - \kappa \right|^{1/2}\right)d\xi+\frac{1}{2}D_{\Gamma}\left(\rho_{\gamma}-\mup,\rho_{\gamma}-\mup\right). $$
The parameter $\kappa$ can be interpreted as the Lagrangian multiplier associated with the charger neutrality constraint. Therefore by fixing $\kappa<0$, the minimization problem~\eqref{pb1} is equivalent to the problem
\begin{equation}
\label{pb1Reformulated}
 \inf \left\{\overline{\cl{E}_{\mathrm{per},x,\kappa}}(\gamma); \,\gamma\in \mathcal{F}_{\Gamma}\, \right\}.
\end{equation}
We prove the existence of minimizers and the uniqueness of the density of minimizers for the problem~\eqref{pb1Reformulated} (hence of~\eqref{pb1}) by considering a minimizing sequence, and show that there is no loss of compactness. This approach is rather classical for rHF type models~\cite{CATTO2001687,Cances2008,CancesLahbabiLewin2013, cao2018}. But in our case we need to be careful as electrons might escape to infinity in the $\bm r$-direction. We show that this is impossible thanks to the Coulomb interactions (Lemma~\ref{ConsistencyDensity}).
\paragraph{Weak convergence of the minimizing sequence.}
First of all it is easy to see that the functional $\overline{\cl{E}_{\mathrm{per},x,\kappa}}(\cdot)$ is well defined on the non-empty set $ \mathcal{F}_{\Gamma}$. Consider a minimizing sequence of $\overline{\cl{E}_{\mathrm{per},x,\kappa}}(\cdot)$ $$\left\{\gamma_{n} :=\B^{-1}\left(\int_{\Gamma^*}\gamma_{n,\xi}\,\frac{d\xi}{2\pi}\right)\B\right\}_{n\geq 1 }$$ on $\mathcal{F}_{\Gamma}$. There exists $C>0$ such that for all $n\geq 1$:
\begin{equation}
\begin{aligned}
&0\leq \int_{\Gamma^*}\m{Tr}_{\lpg}\left(\left|-\frac{1}{2}\Delta_{\xi}-\kappa\right|^{1/2}\gamma_{n,\xi}\left|-\frac{1}{2}\Delta_{\xi}-\kappa\right|^{1/2}\right)d\xi \leq C,\\
& 0\leq D_{\Gamma}\left(\rho_{\gamma_n}-\mups,\rho_{\gamma_n}-\mups\right)\leq C.
\end{aligned}
\label{uniformBoundedEnergy}
\end{equation}
The kinetic energy bound~\eqref{uniformBoundedEnergy} together with the inequality~\eqref{HOinequality} implies that the sequence $\left\{ \sqrt{\rho_{\gamma_n}}\right\}_{n\geq 1} $ is uniformly bounded in $H_{\mathrm{per},x}^1(\Gamma)$ hence in $ L_{\mathrm{per},x}^6(\Gamma)$ by Sobolev embeddings. Therefore for all $ n \in \bb N^*$, the density $\rho_{\gamma_n}$ belongs to $ L_{\mathrm{per},x}^p(\Gamma)$ for $1\leq p\leq 3$. On the other hand, for almost all $\xi\in \Gamma^*$, the operator $\gamma_{n,\xi}$ is a trace-class operator on $\lpg$. As $0\leq \gamma_{n,\xi}^2\leq \gamma_{n,\xi}\leq 1$, we obtain that
\begin{align*}
0\leq \int_{\Gamma^*}\norm{\gamma_{n,\xi}}_{\mathfrak{S}_2}^2\,d\xi=\int_{\Gamma^*}\m{Tr}_{\lpg}\left( \gamma_{n,\xi}^2\right)\,d\xi\leq \int_{\Gamma^*}\m{Tr}_{\lpg}\left( \gamma_{n,\xi}\right)\,d\xi  = 2\pi Z.
\end{align*}
This implies that the operator-valued function $\xi \mapsto \gamma_{n,\xi} $ is uniformly bounded in $L^2(\Gamma^*; \mathfrak{S}_2) $. Furthermore, the uniform boundedness $ 0\leq \gamma_{n,\xi}\leq 1$ also implies that the operator-valued function $\xi \mapsto \gamma_{n,\xi}$ belongs to $L^{\infty}(\Gamma^*; \mathcal{S}(\lpg))$. Combining these remarks with the uniform energy bound~\eqref{uniformBoundedEnergy}, we deduce that there exist (up to extraction):
\begin{equation}
\gamma =\B^{-1}\left(\int_{\Gamma^*}\gamma_{\xi}\,\frac{d\xi}{2\pi}\right)\B,\quad \overline{\rho_{\gamma}}\in L_{\mathrm{per},x}^p (\Gamma),\quad  \widetilde{\rho_{\gamma}} - \mups \in \mathcal{C}_{\Gamma},
\label{weakLimit}
\end{equation}
such that $\gamma_n \xrightharpoonup[]{*}\gamma  $ in the following sense: for any operator-valued function $\xi \mapsto U_{\xi}\in L^2(\Gamma^*; \mathfrak{S}_2) + L^1(\Gamma^*; \mathfrak{S}_1),$
\begin{equation}
\int_{\Gamma^*} \m{Tr}_{\lpg}\left(U_{\xi}\gamma_n \right) d\xi \toinfty \int_{\Gamma^*} \m{Tr}_{\lpg}\left( U_{\xi}\gamma\right) d\xi.
\label{faible1}
\end{equation}
The density $$\rho_{\gamma_n} \rightharpoonup  \overline{\rho_{\gamma}}\,\, \text{ weakly in } \,\,L_{\mathrm{per},x}^p (\Gamma) \,\text{ for } \,1< p\leq 3. $$ The total density $$\rho_{\gamma_n}  - \mups\rightharpoonup  \widetilde{\rho_{\gamma}} - \mups \,\text{ weakly in } \,\mathcal{C}_{\Gamma}.$$
The convergence~\eqref{faible1} is due to the fact that the predual of $L^{\infty}(\Gamma^*; \mathcal{S}(\lpg)) $ is $ L^1(\Gamma^*;  \mathfrak{S}_1)$, and that $L^2(\Gamma^*; \mathfrak{S}_2) $ is a Hilbert space. 
\begin{remark}
The convergence of $\gamma_n \xrightharpoonup[]{*}\gamma  $ in the sense of~\eqref{faible1} can also be reformulated as the convergence in the sense of the following weak-$*$ topology:
\begin{equation}
\label{weakTopoCstar}
 \B\,\gamma_n \,\B^{-1} \xrightarrow[n\to\infty]{} \B\, \gamma \,\B^{-1}\text{ for the weak-$*$ topology of } L^{\infty}\left(\Gamma^*; \mathcal{S}\left(\lpg\right)\right) \bigcap L^2\left(\Gamma^*; \mathfrak{S}_2\right) .
\end{equation}
\end{remark}
Denote by $\mathcal{D}_{\mathrm{per},x}(\Gamma)$ the functions which are $C^{\infty}$ on $\bb R$, $1$-periodic in the $x$-direction, and have compact support in the $\bm r$-direction. Denote by $\mathcal{D}_{\mathrm{per},x}'(\Gamma)$ the dual space of $\mathcal{D}_{\mathrm{per},x}(\Gamma)$. The following lemma guarantees that the densities obtained by different weak limit processes coincide. In particular there is no loss of compactness in the $\bm r$-direction when $|\bm r| \to \infty$.
\begin{lemma}[Consistency of densities]
Denote by $\rho_{\gamma}$ the density associated with the density matrix $\gamma$ obtained in the weak limit~\eqref{weakLimit}. Then $\rho_{\gamma}  = \overline{\rho_{\gamma}}  =\widetilde{\rho_{\gamma}}  $ in $\mathcal{D}_{\mathrm{per},x}'(\Gamma)$. In particular, $\rho_{\gamma} = \overline{\rho_{\gamma}} $ as elements in $L_{\mathrm{per},x}^p (\Gamma)$ for $1<p\leq 3$ and $\rho_{\gamma} -\mups  =\widetilde{\rho_{\gamma}} - \mups $ as elements in $\mathcal{C}_{\Gamma} $.
\label{ConsistencyDensity}
\end{lemma}
We postpone the proof to Section~\ref{SecConsistencyDensity}.
\paragraph{The state $\gamma$ is a minimizer.}
Let us first show that the kinetic energy of $\gamma$ obtained by the weak limit~\eqref{faible1} is finite. To achieve this, consider an orthonormal basis $\left\{e_i\right\}_{i\in \bb N}\subset H_{\mathrm{per},x}^1(\Gamma)$ of $L_{\mathrm{per},x}^2( \Gamma)$, and define the following family of operators for $N \in \bb N^*$:
$$M_{\xi}^N:= \left|-\frac{1}{2}\Delta_{\xi}-\kappa\right|^{1/2}\left(\sum_{i=1}^N \Ket{e_i}\Bra{e_i}\right)\left|-\frac{1}{2}\Delta_{\xi}-\kappa\right|^{1/2}.$$
An easy computation shows that for all $\xi\in \Gamma^*$, the operator $M_{\xi}^N$ belongs to $\mathfrak{S}_2$. Moreover, the function $\xi\mapsto M_{\xi}^N$ can be seen as an operator-valued function belonging to $L^2(\Gamma^*;\mathfrak{S}_2)$ as $\Gamma^*=[-\pi,\pi)$ is a finite interval. In view of the convergence~\eqref{faible1} and by choosing $U_{\xi}  = M_{\xi}^N$ we obtain that (recalling that $\m{Tr}(AB) = \m{Tr}(BA) $ when $A,B$ are Hilbert-Schmidt operators, and $\left|-\frac{1}{2}\Delta_{\xi}-\kappa\right|^{1/2}\gamma_{\xi} $ is Hilbert-Schmidt for almost all $\xi\in\Gamma^*$)
\begin{align*}
0\leq&\int_{\Gamma^*} \m{Tr}_{\lpg}\left(\left|-\frac{1}{2}\Delta_{\xi}-\kappa\right|^{1/2}\gamma_{\xi} \left|-\frac{1}{2}\Delta_{\xi}-\kappa\right|^{1/2}\left(\sum_{i=1}^N \Ket{e_i}\Bra{e_i}\right)\right) d\xi\\
&= \int_{\Gamma^*} \m{Tr}_{\lpg}\left( M_{\xi}^N\gamma_\xi\right) d\xi =\lim_{n\to\infty}\int_{\Gamma^*} \m{Tr}_{\lpg}\left(M_{\xi}^N\gamma_{n,\xi} \right) d\xi\\
&= \lim_{n\to\infty}\int_{\Gamma^*}\m{Tr}_{\lpg}\left(\left|-\frac{1}{2}\Delta_{\xi}-\kappa\right|^{1/2}\gamma_{n,\xi} \left|-\frac{1}{2}\Delta_{\xi}-\kappa\right|^{1/2}\left(\sum_{i=1}^N \Ket{e_i}\Bra{e_i}\right)\right) d\xi \\
&= \lim_{n\to\infty}\int_{\Gamma^*} \sum_{i=1}^N \left\langle e_i\left|\, \left|-\frac{1}{2}\Delta_{\xi}-\kappa\right|^{1/2}\gamma_{n,\xi} \left|-\frac{1}{2}\Delta_{\xi}-\kappa\right|^{1/2}\right| e_i \right\rangle\, d\xi \\
&\leq   \liminf_{n\to\infty}\int_{\Gamma^*} \m{Tr}_{\lpg}\left(\left|-\frac{1}{2}\Delta_{\xi}-\kappa\right|^{1/2}\gamma_{n,\xi} \left|-\frac{1}{2}\Delta_{\xi}-\kappa\right|^{1/2}\right) d\xi  \leq C',
\end{align*}
where the last step we have used the uniform energy bound~\eqref{uniformBoundedEnergy}. Therefore, by passing to the limit $N\to +\infty$, by Fatou's lemma we have
\begin{equation}
\begin{aligned}
0\leq \int_{\Gamma^*} \m{Tr}_{\lpg}&\left( \left|-\frac{1}{2}\Delta_{\xi}-\kappa\right|^{1/2}\gamma_{\xi} \left|-\frac{1}{2}\Delta_{\xi}-\kappa\right|^{1/2}\right) d\xi\\
&\qquad \leq \liminf_{n\to\infty}\int_{\Gamma^*} \m{Tr}_{\lpg}\left(\left|-\frac{1}{2}\Delta_{\xi}-\kappa\right|^{1/2}\gamma_{n,\xi} \left|-\frac{1}{2}\Delta_{\xi}-\kappa\right|^{1/2}\right) d\xi  \leq C'.
\end{aligned}
\label{liminfKineticEnergy}
\end{equation}
Remark that the bound~\eqref{liminfKineticEnergy} also implies that $\sqrt{\rho_{\gamma}}$ belongs to $H_{\mathrm{per},x}^1 (\Gamma)$ by the Hoffmann--Ostenhof inequality~\eqref{HOinequality}. Hence $\rho_{\gamma}\in L_{\mathrm{per},x}^p (\Gamma)$ for $1\leq p\leq 3$. Since $\rho_{\gamma} -\mups $ is an element in $\mathcal{C}_{\Gamma} $ by Lemma~\ref{ConsistencyDensity}, this implies that the charge is neutral by Remark~\ref{remark_neutral}. That is, 
$$\int_{\Gamma}\rho_{\gamma} =  \int_{\Gamma}\mups .$$
As $D_{\Gamma}(\cdot,\cdot)$ defines an inner product on $\mathcal{C}_{\Gamma}$, by the weak convergence~\eqref{weakLimit} of $\rho_{\gamma_n} - \mups $ to $ \widetilde{\rho_{\gamma}} - \mups $ in $\mathcal{C}_{\Gamma}$ and the consistency of densities given by Lemma~\ref{ConsistencyDensity}, we obtain that
\begin{equation}
D_{\Gamma}\left(\rho_{\gamma} -\mups ,\rho_{\gamma} -\mups \right) \leq \liminf_{n\to\infty}D_{\Gamma}\left(\rho_{\gamma_n} -\mups ,\rho_{\gamma_n} -\mups \right).
\label{liminfCoulombInteraction}
\end{equation}
In view of~\eqref{liminfKineticEnergy} and~\eqref{liminfCoulombInteraction}, we conclude that
\begin{align*}
\mathcal{E}_{\mathrm{per},x}(\gamma)\leq \liminf_{n\to\infty} \mathcal{E}_{\mathrm{per},x}(\gamma_n),  
\end{align*}
which shows that the state $\gamma$ obtained in~\eqref{weakLimit} is a minimizer of the problem~\eqref{pb1}. Let us prove that all minimizers share the same density: consider two minimizers $\overline{\gamma}_1$ and $\overline{\gamma}_2$. By the convexity of $\cl{F}_{\Gamma}$ it holds that $\frac{1}{2}\left(\overline{\gamma}_1+\overline{\gamma}_2\right) \in \cl{F}_{\Gamma}$. Moreover
\begin{align*}
\mathcal{E}_{\mathrm{per},x}\left(\frac{\overline{\gamma}_1+\overline{\gamma}_2}{2}\right) = \frac{1}{2}\mathcal{E}_{\mathrm{per},x}\left(\overline{\gamma}_1\right) +\frac{1}{2}\mathcal{E}_{\mathrm{per},x}\left(\overline{\gamma}_2\right) - \frac{1}{4}D_{\Gamma}\left(\rho_{\overline{\gamma}_1}-\rho_{\overline{\gamma}_2},\rho_{\overline{\gamma}_1}-\rho_{\overline{\gamma}_2}\right),
\end{align*}
which shows that $D_{\Gamma}\left(\rho_{\overline{\gamma}_1}-\rho_{\overline{\gamma}_2},\rho_{\overline{\gamma}_1}-\rho_{\overline{\gamma}_2}\right) \equiv 0$, hence all the minimizers of the problem~\eqref{pb1} share the same density. 
\subsection{Proof of Lemma~\ref{ConsistencyDensity}}
\label{SecConsistencyDensity}
\paragraph{Equality of $\widetilde{\rho_{\gamma}} $ and $\overline{\rho_{\gamma}}$. }
The proof follows ideas similar to the ones used for the proof of~\cite[Lemma 3.5]{cao2018} by considering a test function in $\mathcal{D}_{\mathrm{per},x}(\Gamma)$ and replacing the Fourier transform by the mixed Fourier transform defined in~\eqref{FourierTS}. Consider a test function $w\in \mathcal{D}_{\mathrm{per},x}(\Gamma)$. The weak convergence of $\rho_{\gamma_n} \rightharpoonup  \overline{\rho_{\gamma}}$ in $L_{\mathrm{per},x}^p (\Gamma)$ with $1<p\leq 3$ implies that
$$\left\langle \rho_{\gamma_n} - \mups,w\right\rangle_{\mathcal{D}_{\mathrm{per},x}',\mathcal{D}_{\mathrm{per},x}} \toinfty \left\langle \overline{\rho_{\gamma}} - \mups,w\right\rangle_{\mathcal{D}_{\mathrm{per},x}',\mathcal{D}_{\mathrm{per},x}}.$$
On the other hand, 
\begin{equation}
\begin{aligned}
\left\langle \rho_{\gamma_n} - \mups,w\right\rangle_{\mathcal{D}_{\mathrm{per},x}',\mathcal{D}_{\mathrm{per},x}}& =\int_{\Gamma}\left( \rho_{\gamma_n} - \mups \right)w = \sum_{n\in\bb Z}\int_{\bb R^2}\overline{\F\left( \rho_{\gamma_n} - \mups \right)(n,\bm k)} \F w(n,\bm k)\,d\bm k \\
&= 4\pi \sum_{n\in\bb Z}\int_{\bb R^2}\frac{\overline{\F\left( \rho_{\gamma_n} - \mups \right)(n,\bm k) }\F f(n,\bm k)}{4\pi^2n^2+|\bm k|^2}\, d\bm k.
\end{aligned}
\label{distributionToCoulomb}
\end{equation}
where $ f=-\frac{1}{4\pi} \Delta w$. Note that $f$ belongs to the Coulomb space $\mathcal{C}_{\Gamma}$ defined in~\eqref{CoulombInterSpc} since for all $n\in \bb Z$, $\F f(n,\cdot)\in L_{\mathrm{loc}}^1(\bb R^2)$, and 
\begin{align*}
D_{\Gamma}(f,f) &=4\pi \sum_{n\in \bb Z}\int_{\bb R^2}\frac{\left|\F f (n,\bm k)\right|^2}{|\bm k |^2+4\pi^2 n^2}\,d\bm k = \frac{1}{4\pi} \sum_{n\in \bb Z}\int_{\bb R^2}\left(|\bm k |^2+4\pi^2 n^2\right)\left|\F w (n,\bm k)\right|^2\,d\bm k\\
&=\frac{1}{4\pi} \sum_{n\in \bb Z}\int_{\bb R^2}\left|\F \left(\nabla w\right) (n,\bm k)\right|^2\,d\bm k = \frac{1}{4\pi}\int_{\Gamma}\left|\nabla w\right|^2 <+\infty.
\end{align*}
Therefore in view of~\eqref{distributionToCoulomb}, the convergence $D_{\Gamma}\left(\rho_{\gamma_n}  - \mups,f\right)  \toinfty D_{\Gamma}\left(\widetilde{\rho_{\gamma}} - \mups,f\right) $ implies that $$\left\langle \rho_{\gamma_n} - \mups,w\right\rangle \toinfty \left\langle \widetilde{\rho_{\gamma}} - \mups,w\right\rangle.$$
The uniqueness of limit in the sense of distribution allows us to conclude. 
\paragraph{Equality of $\rho_{\gamma}$ and $\overline{\rho_{\gamma}}$. }
Let us prove that $\rho_{\gamma}  = \overline{\rho_{\gamma}} $ in $\mathcal{D}_{\mathrm{per},x}'(\Gamma)$. The fact that $\rho_{\gamma_n} \rightharpoonup \overline{\rho_{\gamma}}$ weakly in $L_{\mathrm{per},x}^p(\Gamma)$ implies that
$$ \left\langle \rho_{\gamma_n},w\right\rangle \toinfty \left\langle \overline{\rho_{\gamma}},w\right\rangle.$$
Therefore it suffices to prove that the operator-valued function $\xi \mapsto w\gamma_{n,\xi} \in L^1\left(\Gamma^*;\mathfrak{S}_1\right) $ converges in the following sense:
\begin{equation}
\frac{1}{2\pi}\int_{\Gamma^*} \m{Tr}_{\lpg} \left(w\gamma_{n,\xi}\right)d\xi = \left\langle \rho_{\gamma_n},w\right\rangle \toinfty  \left\langle \rho_{\gamma},w\right\rangle = \frac{1}{2\pi}\int_{\Gamma^*} \m{Tr}_{\lpg} \left(w \gamma_{\xi}\right)d\xi.
\label{compacityOfMass}
\end{equation}
The weak convergence~\eqref{faible1} does not guarantee the above convergence since the function $w$ does not belong to any Schatten class. We prove~\eqref{compacityOfMass} by using the kinetic energy bound~\eqref{uniformBoundedEnergy}, which implies that the operator-valued function $$\xi\mapsto \left|-\frac{1}{2}\Delta_{\xi}-\kappa\right|^{1/2} \gamma_{n,\xi} =\left(\left|-\frac{1}{2}\Delta_{\xi}-\kappa\right|^{1/2} \gamma_{n,\xi} \left|-\frac{1}{2}\Delta_{\xi}-\kappa\right|^{1/2}\right)\left|-\frac{1}{2}\Delta_{\xi} - \kappa\right|^{-1/2}$$
is uniformly bounded in $L^1\left(\Gamma^*; \mathfrak{S}_1\right)$ as $\left|-\frac{1}{2}\Delta_{\xi} - \kappa\right|^{-1/2}$ is uniformly bounded with respect to $\xi\in\Gamma^*$. Moreover, the energy bounded~\eqref{uniformBoundedEnergy} also implies that $ \xi\mapsto \left|-\frac{1}{2}\Delta_{\xi}-\kappa\right|^{1/2}\sqrt{\gamma_{n,\xi}}$ is uniformly bounded in $L^2\left(\Gamma^*; \mathfrak{S}_2\right)$. Hence the operator-valued function $\xi\mapsto \left|-\frac{1}{2}\Delta_{\xi}-\kappa\right|^{1/2} \gamma_{n,\xi} =\left|-\frac{1}{2}\Delta_{\xi}-\kappa\right|^{1/2}\sqrt{\gamma_{n,\xi}}\sqrt{\gamma_{n,\xi}}$
is~uniformly bounded in $L^2\left(\Gamma^*; \mathfrak{S}_2\right)$ as $0\leq \gamma_{n,\xi}\leq 1$. Therefore $$\xi \mapsto \left|-\frac{1}{2}\Delta_{\xi}-\kappa\right|^{1/2} \gamma_{n,\xi}  \text{ is uniformly bounded in }L^1\left(\Gamma^*; \mathfrak{S}_1\right)\bigcap L^2\left(\Gamma^*; \mathfrak{S}_2\right), $$
hence in $L^q\left(\Gamma^*; \mathfrak{S}_q\right)$ for $1\leq q \leq 2$ by interpolation. Therefore, up to extraction, the following weak convergence holds: for any operator-valued function $\xi \mapsto W_{\xi} \in L^{q'}\left(\Gamma^*;\mathfrak{S}_{q'}\right)$ where $q'= \frac{q}{q-1}$ for $1<q\leq 2$, 
\begin{equation}
\int_{\Gamma^*} \m{Tr}_{\lpg}\left(W_{\xi} \left|-\frac{1}{2}\Delta_{\xi}-\kappa\right|^{1/2}\gamma_{n,\xi}\right) d\xi \toinfty  \int_{\Gamma^*} \m{Tr}_{\lpg}\left( W_{\xi}\left|-\frac{1}{2}\Delta_{\xi}-\kappa\right|^{1/2} \gamma_{\xi} \right) d\xi.
\label{SpConvergence}
\end{equation}
On the other hand by the inequality~\eqref{KSSGamma} we obtain that for any $\xi\in\Gamma^*$, 
\begin{align*}
\Bnorm{w\left|-\frac{1}{2}\Delta_{\xi} - \kappa\right|^{-1/2}}_{\mathfrak{S}_{q'}} &\leq \frac{1}{(2\pi)^{2/q'}}\norm{w}_{L_{\mathrm{per},x}^{q'}(\Gamma)}\left(\sum_{n\in \bb Z}\int_{\bb R^2}\frac{2^{q'/2}}{\left((2\pi n +\xi)^2+|\bm r|^2+2|\kappa|\right)^{q'/2}}\,d\bm r\right)^{1/q'}.
\end{align*}
Upon choosing for example $q'= 4$, the right hand side of the above quantity is finite. Therefore upon taking $W_{\xi} = w\left|-\frac{1}{2}\Delta_{\xi} - \kappa\right|^{-1/2}$ in~\eqref{SpConvergence} we obtain that
\begin{align*}
&\frac{1}{2\pi}\int_{\Gamma^*} \m{Tr}_{\lpg} \left(w\gamma_{n,\xi}\right)d\xi= \frac{1}{2\pi}\int_{\Gamma^*} \m{Tr}_{\lpg} \left(w\left|-\frac{1}{2}\Delta_{\xi} - \kappa\right|^{-1/2}\left|-\frac{1}{2}\Delta_{\xi}-\kappa\right|^{1/2}\gamma_{n,\xi}\right)d\xi\\
&\toinfty \frac{1}{2\pi} \int_{\Gamma^*} \m{Tr}_{\lpg}\left(w\left|-\frac{1}{2}\Delta_{\xi} - \kappa\right|^{-1/2}\left|-\frac{1}{2}\Delta_{\xi}-\kappa\right|^{1/2}\gamma_{\xi}\right)d\xi =  \frac{1}{2\pi}\int_{\Gamma^*} \m{Tr}_{\lpg} \left(w \gamma_{\xi}\right)d\xi.
\end{align*}
Hence~\eqref{compacityOfMass} holds. Therefore $\rho_{\gamma}  = \overline{\rho_{\gamma}} $ in $\mathcal{D}_{\mathrm{per},x}'(\Gamma)$, which concludes the proof of the lemma.
\subsection{Proof of Theorem~\ref{thmPeriodicExistence2}}
\label{thmPeriodicExistenceSec2}
We first define a mean-field Hamiltonian associated with the problem~\eqref{pb1}, and then show that the Fermi level is always negative. Moreover, the minimizer of~\eqref{pb1} is uniquely given by the spectral projector of the mean-field Hamiltonian. In the end we show that the density of the minimizer decays exponentially fast in the $\bm r$-direction. 
\paragraph{Properties of the mean--field potential and Hamiltonian.}
We begin with the definition of a mean-field potential and a mean-field Hamiltonian, and next study the spectrum of the mean-field Hamiltonian. Consider a minimizer $\rp$ of~\eqref{pb1} with the unique density $\rho_{\rp} \in L_{\mathrm{per},x}^p(\Gamma)$ where $1\leq p\leq 3$. Define the mean-field potential
\begin{align*}
 \vp&:= \ep \star_{\Gamma}G,\quad  \ep :=\rho_{\rp} -\mups.
\end{align*}
which is the solution of Poisson's equation $-\Delta  \vp = 4\pi\ep$. Let us prove that $\vp$ belongs to $L_{\mathrm{per},x}^{p}(\Gamma)$ for $ 2< p\leq +\infty$. 
As $\mup$ is smooth and has compact support in the $\bm r$-direction and $\rho_{\rp}$ belongs to $ L_{\mathrm{per},x}^p (\Gamma)$ for $1\leq p\leq 3$, hence $\F\ep(0,\cdot)$ belongs to $L^2(\R^2) \cap L^{\infty}(\R^2)\cap C^0(\R^2)$ by classical Fourier theory (see for example~\cite{ReeSim2} ). Moreover, as $\int_{\Gamma}|\bm r|\rho_{\rp}(x,\bm r) <+\infty$ by~\eqref{mildIntegrabilityDensity},
\begin{equation}
\label{majoration_partiale}
\begin{aligned}
\bnorm{\partial_{\bm k} \F\ep(0,\bm k)}_{L^{\infty}(\R^2)} &= \left|\int_{\Gamma}\rme^{-\rmi \bm r \cdot\bm k}\bm r\cdot\ep(x,\bm r)\,dx\,d\bm r \right| \\
&\leq \left|\int_{\Gamma}|\bm r|\left|\rho_{\rp} +\mups\right|(x,\bm r)\,dx \,d\bm r \right|<+\infty .
\end{aligned}
\end{equation}
This implies that $ \F\ep(0,\bm k)$ is $C^1(\R^2)$ and bounded. Remark also that $ \F\ep(0,\bm 0)=0$ by the charge neutrality and that $\F\ep(0,\cdot)$ belongs to $L^2(\R^2) \cap L^{\infty}(\R^2)\cap C^0(\R^2)$, hence for $1\leq \alpha <2$, 
\begin{equation}
\label{eq:FourierofV}
\begin{aligned}
\int_{\bb R^2} \left|\F\vp (0,\bm k)\right|^{\alpha}\,d\bm k&=\int_{\bb R^2}\frac{\left|\F\ep(0,\bm k)\right|^\alpha}{|\bm k|^{2\alpha}}\,d\bm k \\
&\leq  \int_{|\bm k| < 1}\frac{\left|\F\ep(0,\bm k)\right|^\alpha}{|\bm k|^{2\alpha}}\,d\bm k+ \left(\int_{|\bm k| \geq 1}\left|\F\ep(0,\bm k)\right|^{2\alpha} \,d\bm k\right)^{1/2} \left(\int_{|\bm k| \geq 1} \frac{1}{|\bm k|^{4\alpha}}\,d\bm k \right)^{1/2} <+\infty.
\end{aligned}
\end{equation}
Remark that the mixed Fourier transform $\F$ is an isometry from $\lpg$ to $ \ell^2\left(\Z, L^2\left(\R^2\right)\right)$ by~\eqref{Isometry_F}. On the other hand, $$ \forall \phi \in  \ell^1\left(\Z, L^1\left(\R^2\right)\right), \quad \norm{\F^{-1}\phi}_{L^{\infty}(\Gamma)} = \sup_{x\in\Gamma}\left| \frac{1}{2\pi}\sum_{n\in \bb Z}\int_{\bb R^2}\phi_n( \bm k)\,\rme^{\rmi (2\pi nx +\bm k\cdot\bm r)} \,d\bm k\right| \leq \frac{1}{2\pi } \norm{\phi}_{ \ell^1\left(\Z, L^1\left(\R^2\right)\right)},$$
By the Riesz–Thorin interpolation theorem (see for example~\cite{ReeSim2} and~\cite[Theorem 5.7]{lieb2001analysis}), we can deduce a Hausdorff-Young inequality for $\F^{-1}$: for $1\leq \alpha\leq 2$ there exists a constant $C_{\alpha}$ depending on $\alpha$ such that 
\begin{equation}
\label{interpolation_mixedF}
\forall \phi\in  \ell^\alpha\left(\Z, L^\alpha\left(\R^2\right)\right),\quad  \norm{\F^{-1}\phi}_{L_{\mathrm{per},x}^{\alpha'}(\Gamma)} \leq C_\alpha \norm{\phi}_{ \ell^\alpha\left(\Z, L^\alpha\left(\R^2\right)\right)}.
\end{equation}
where $\alpha' := \alpha/(\alpha-1)$. Hence in view of~\eqref{eq:FourierofV} and~\eqref{interpolation_mixedF}, for $1\leq \alpha <2$ and $2<\alpha' = \alpha/(\alpha-1)\leq +\infty$ there exist positive constants $C_{\alpha,1}$, $C_{\alpha,2}$ and $C_{\alpha,2}'$ such that 
\begin{equation}
\label{vpEstimateLP}
\begin{aligned}
\norm{\vp}_{L_{\mathrm{per},x}^{\alpha'}(\Gamma)}^{\alpha}&\leq C_\alpha^{\alpha}\norm{\F\vp}_{ \ell^\alpha\left(\Z, L^\alpha\left(\R^2\right)\right)}^{\alpha} = \frac{C_\alpha^{\alpha}}{2\pi } \sum_{n\in \Z}\int_{\R^2} \left|\F\vp (n,\bm k)\right|^{\alpha}\,d\bm k\\
& =\frac{C_\alpha^{\alpha}}{2\pi }\int_{\bb R^2} \left|\F\vp (0,\bm k)\right|^\alpha\,d\bm k +\frac{C_\alpha^{\alpha}}{2\pi } \sum_{n\neq 0} \int_{\bb R^2}\frac{\left|\F\ep(n,\bm k)\right|^{\alpha}}{\left(4\pi n^2+|\bm k|^2\right)^{\alpha}}\,d\bm k\\
&\leq C_{\alpha,1}+ \frac{C_\alpha^{\alpha}}{2\pi }\left( \sum_{n\neq 0} \int_{\bb R^2}\left|\F\ep(n,\bm k)\right|^{2\alpha} d\bm k\right)^{1/2}\left(\sum_{n\neq 0}\int_{\bb R^2}\frac{1}{(4\pi n^2+|\bm k|^2)^{2\alpha}}\,d\bm k\right)^{1/2}\\
&\leq C_{\alpha,1}+ C_{\alpha,2}\norm{\F\ep}_{ \ell^{2\alpha}\left(\Z, L^{2\alpha}\left(\R^2\right)\right)}^{\alpha} \leq C_{\alpha,1}+C_{\alpha,2}'\norm{\ep}_{L_{\mathrm{per},x}^{q}(\Gamma)}^{\alpha} <+\infty,
\end{aligned}
\end{equation}
where the last step we have used estimates similar to~\eqref{interpolation_mixedF} for $\F$ and the fact that $\ep$ belongs to $L_{\mathrm{per},x}^q(\Gamma)$ for  $4/3 <q :=2\alpha/(2\alpha -1) \leq 2$. Therefore $\vp$ belongs to $ L_{\mathrm{per},x}^{p}(\Gamma)$ for $ 2< p\leq +\infty$. By the elliptic regularity we know that $\vp$ belongs to the Sobolev space $W_{\mathrm{per},x}^{2,p}(\Gamma)$ for $ 2<p\leq 3$, where $W_{\mathrm{per},x}^{2,p}(\Gamma)$ is the space of functions, together with their gradients and hessians,  belong to $L_{\mathrm{per},x}^{p}(\Gamma)$. Approximating $\vp$ by functions in $\mathcal{D}_{\mathrm{per},x}(\Gamma)$, we also deduce that $\vp$ tends to $0$ when $|\bm r|$ tends to infinity. The mean-field potential $\vp$ defines a $-\Delta$-bounded operator on $L^2(\bb R^3)$ with relative bound zero, hence by the Kato--Rellich theorem (see for example \cite[Theorem 9.10]{Helffer}) we know that $\hp = -\frac{1}{2}\Delta +\vp$ uniquely defines a self-adjoint operator on $L^2(\bb R^3)$ with domain $H^2\left(\bb R^3\right)$ and form domain $H^1\left(\bb R^3\right)$. As $\hp$ is $\bb Z$-translation invariant in the $x$-direction,
$$\hp = \B^{-1}\left(\int_{\Gamma^*}\hpx\,\frac{d\xi}{2\pi}\right)\B,\quad \hpx := -\frac{1}{2}\Delta_{\xi}+\vp.$$
Note that the decomposed Hamiltonian $\hpx $ does not have a compact resolvent as $\Gamma $ is not a bounded domain. It is easy to see that $\sigma\left(-\Delta_{\xi}\right) =\sigma_{\mathrm{ess}}\left(-\Delta_{\xi}\right)  = \left[0,+\infty\right).$ On the other hand, by the inequality~\eqref{KSSGamma} we have
\begin{align*}
\bnorm{\vp\left(1-\Delta_{\xi}\right)^{-1}}_{\mathfrak{S}_2} &\leq \frac{1}{2\pi}\norm{\vp}_{\lpg}\left(\sum_{n\in \bb Z}\int_{\bb R^2}\frac{1}{\left((2\pi n+\xi)^2+|\bm k|^2+1\right)^2}\,d\bm k\right)^{1/2} < +\infty.
\end{align*}
In particular $\vp$ is a compact perturbation of $-\Delta_{\xi}$, and therefore introduces at most countably many eigenvalues below $0$ which are bounded from below by $-\norm{\vp}_{L^{\infty}}$. Denote by $\left\{\lambda_n(\xi)\right\}_{1\leq n \leq N_{H} }$ these (negative) eigenvalues for $N_{H}\in \bb N^*$ ($N_H$ can be finite or infinite). Then for all $\xi\in \Gamma^*$, $$\sigma_{\mathrm{ess}}(\hpx) =\sigma_{\mathrm{ess}}(-\Delta_{\xi})= [0,+\infty),\qquad \sigma_{\mathrm{disc}}(\hpx) = \bigcup_{1\leq n \leq N_H} \lambda_n\left(\xi \right) .$$ In view of the decomposition~\eqref{Hper}, a result of~\cite[Theorem XIII.85]{ReeSim4} gives the following spectral decomposition: $$\sigma_{\mathrm{ess}}\left(\hp\right)\supseteq \bigcup_{\xi\in\Gamma^*}\sigma_{\mathrm{ess}}(\hpx) =[0,+\infty) ,\quad \sigma_{\mathrm{disc}}\left(\hp\right) \subseteq \bigcup_{\xi\in\Gamma^*}\sigma_{\mathrm{disc}}(\hpx) = \bigcup_{\xi\in\Gamma^*}\bigcup_{1\leq n \leq N_H} \lambda_n\left(\xi \right) .$$ We also obtain from~\cite[Item (e) of Theorem XIII.85]{ReeSim4} that $$\lambda \in \sigma_{\mathrm{disc}}(\hp) \Leftrightarrow  \left\{\xi\in\Gamma^*\left|\, \lambda\in \sigma_{\mathrm{disc}}(\hpx)\right.\right\} \text{ has non-trivial Lebesgue measure. }$$
By the regular perturbation theory of the point spectra~\cite{kato1995} (see also~\cite[Section XII.2]{ReeSim4}) and the approach of Thomas~\cite[Lemma 1]{thomas1973}, we know that the eigenvalues $\lambda_n(\xi)$ below $0$ are analytical functions of $\xi$ and cannot be constant, so that $\left\{\xi\in\Gamma^*\left|\, \lambda\in \sigma_{\mathrm{disc}}(\hpx)\right.\right\}$ has trivial Lebesgue measure, and the  essential spectrum of $\hp$ below $0$ is purely absolutely continuous. As a conclusion,
 $$\sigma\left(\hp\right)=\sigma_{\mathrm{ess}}\left(\hp\right)= \bigcup_{\xi\in\Gamma^*}\sigma(\hpx).$$
\paragraph{The Fermi level is always negative.}
Let us prove that the inequality $N_H = F(0) \geq \int_{\Gamma} \mups $ always holds. The physical meaning of this statement is that the Fermi level of the quasi 1D system at ground state is always negative when the mean-field potential tends to $0$ in the $\bm r$-direction. We prove this by contradiction: assume that $F(0) <  \int_{\Gamma}\mups$, then we can always construct (infinitely many) states belonging to $\cl{F}_{\Gamma}$ with positive energies arbitrarily close to $0$ and they decrease the ground state energy of the problem~\eqref{pb1}. 

Let us first define a spectral projector representing all the states of $\hp$ below $0$: for any $\xi \in \Gamma^*$ and $\hpx$ defined in~\eqref{Hper}, define
$$ \rpo := \mathds 1_{(-\infty, 0]}( \hp ),\qquad \rpxo := \mathds 1_{(-\infty, 0]}( \hpx ).
$$
Therefore, 
\begin{equation}
N_H= F(0)= \frac{1}{|\Gamma^*|}\int_{\Gamma^*}\m{Tr}_{\lpg}\left(\rpxo\right)d\xi = \int_{\Gamma}\rho_{\rpo}.
 \label{F0}
\end{equation}
The inequality that $F(0) <  \int_{\Gamma}\mups$ implies that
\begin{equation}
 Z_{\mathrm{diff}} :=\int_{\Gamma}\mups - N_H  = \int_{\Gamma}\mups -\int_{\Gamma}\rho_{\rpo} \quad \in  \bb N^*.
 \label{DiffCharge}
\end{equation}
The condition~\eqref{DiffCharge} implies in particular that $N_H<+\infty$, \textit{i.e.}, that there are at most finitely many states below $0$. Let us construct states of $\hp$ with positive energies. These states belonging to $L^2(\bb R^3)$ approximate the plane waves of $\hp$ traveling in the $\bm r$-direction. For $R>0$, recall that $\mathfrak{B}_{R}$ is the ball in $\R^3$ centered at $0$. Consider a smooth function $t(x,\bm r)$ supported in $\mathfrak{B}_{1}$, equal to one in $\mathfrak{B}_{1/2}$ and such that $\norm{t}_{L^2(\bb R^3)} =1$. For $n\in \bb N^*$, let us define $$\psi_{n}(x, \bm r) := n^{-3/2} t\left(\frac{(x,\bm r) - \left(n^2, \left(n^2,n^2\right)\right)}{n}\right).$$
It is easy to see that $\psi_{n}$ belongs to $L^2(\R^3)$, converges weakly to $0$ when $n$ tends to infinity and $\norm{\psi_{n}}_{L^2(\R^3)} = 1$. Moreover, as $\vp$ tends to $0$ in the $\bm r$-direction, hence for any $\varepsilon>0$ there exists an integer $N_{\varepsilon}$ such that $\left|\vp(\cdot, (n^2,n^2))\right| \leq \varepsilon$ when $n\geq N_\varepsilon$. Denote by $\left\{\psi_{n,\xi}\right\}_{n\in \bb N^*,\xi\in\Gamma^*} $ the Bloch decomposition $\B$ in the $x$-direction (see Section~\ref{PreliminarySec} for the definition) of $\left\{\psi_{n}\right\}_{n\in \bb N^*}$ which belong to $\lpg$. For $n\geq N_\varepsilon$, it holds
\begin{equation}
\begin{aligned}
\bnorm{\hp\psi_{n}}_{L^2(\R^3)}^2 &= \frac{1}{2\pi}\int_{\Gamma^*} \bnorm{\hpx\psi_{n,\xi}}_{\lpg}^2 \,d\xi\\
&= \bnorm{ -n^{-7/2}\Delta t\left(\frac{\cdot- \left(n^2, \left(n^2,n^2\right)\right)}{n}\right) + \vp  \psi_{n}}_{L^2(\R^3)}^2 \leq 2\left(\frac{1}{n^4}  +\varepsilon^2\right).
\end{aligned} 
\label{WeylForScatteringStates}
\end{equation}
Remark that $\rpxo \hpx = \sum_{n=1}^{N_H}P_{\{\lambda_n(\xi)\}}(\hpx)$ is a compact operator, where $P_{\{\cdot\}}(\hpx)$ the spectral projector of $\hpx$. There exists an orthonormal basis $\left\{e_{n,\xi}\right\}_{n\geq 1}$ of $\lpg$ with elements in $H_{\mathrm{per},x}^1(\Gamma) $ such that $\rpxo \hpx e_{n,\xi } = \lambda_{n}(\xi)e_{n,\xi } $ for $1\leq n \leq N_H$, and $\rpxo\hpx e_{n,\xi }\equiv 0$ for $n > N_H$. Let us construct test density matrices composed by all the states of $\hp$ with negative energies and some states with positive energies. More precisely, for $N_0\in \bb N^*$ to be made precise later, consider a test density matrix 
$$  \gamma_{N_0}= \B^{-1}\left(\int_{\Gamma^*} \gamma_{N_0,\xi} \frac{d\xi}{2\pi}\right)\B ,$$
where \begin{align*}
\gamma_{N_0,\xi}& :=\rpxo+ \sum_{n=N_0+1}^{N_0+Z_{\mathrm{diff}}}(1-\rpxo)\left| \psi_{n,\xi}\left\rangle \right\langle \psi_{n,\xi}\right|\\
&=\sum_{ n=1}^{+\infty}\rpxo \left| e_{n,\xi}\left\rangle \right\langle e_{n,\xi}\right|   + \sum_{n=N_0+1}^{N_0+Z_{\mathrm{diff}}}(1-\rpxo)\left| \psi_{n,\xi}\left\rangle \right\langle \psi_{n,\xi}\right|.
\end{align*}
\begin{lemma}
\label{TestStatesADMLemma}
For any $N_0\in \bb N^*$, the state~$\gamma_{N_0}$ belongs to the admissible set~$\cl{F}_{\Gamma}$.
\end{lemma}
\begin{proof}
It is easy to see that $0\leq \gamma_{N_0} \leq 1$. Remark also that $\m{Ran}\left(\rpxo\right)= \m{Span}\left\{e_{n,\xi}\right\}_{1\leq n\leq N_H}$ for all $\xi\in \Gamma^*$. The density of $\gamma_{N_0}$ can be written as 
\begin{align*}
\rho_{\gamma_{N_0}} =  \frac{1}{2\pi}\int_{\Gamma^*} \sum_{n=1}^{N_H}\left|e_{n,\xi}\right|^2d\xi  +  \frac{1}{2\pi}\sum_{n=N_0+1}^{N_0+Z_{\mathrm{diff}}}\int_{\Gamma^*}   \left| \psi_{n,\xi}\right|^2 d\xi.
\end{align*}
The density $\rho_{\gamma_{N_0}}$ belongs to $L_{\mathrm{per}}^p(\Gamma) $ for $1\leq p \leq 3$ as $\left\{e_{n,\xi}\right\}_{n\geq 1}$ and $\left\{\psi_{n,\xi}\right\}_{n\geq 1}$ belong to $H_{\mathrm{per},x}^1(\Gamma)$. Besides, in view of~\eqref{F0}, 
\begin{equation}
\begin{aligned}
\int_{\Gamma}\rho_{\gamma_{N_0}}& =  \frac{1}{2\pi}\int_{\Gamma^*}\m{Tr}_{\lpg}\left(\rpxo\right)d\xi  + \frac{1}{2\pi}\sum_{n=N_0+1}^{N_0+Z_{\mathrm{diff}}}\int_{\Gamma}\int_{\Gamma^*}   \left| \psi_{n,\xi}\right|^2 d\xi\\
&=N_H + \sum_{n=N_0+1}^{N_0+Z_{\mathrm{diff}}} \int_{\R^3} |\psi_{n}|^2= N_H+ Z_{\mathrm{diff}} = \int_{\Gamma}\mups.
\end{aligned}
\label{NeutralTestState}
\end{equation}
A simple calculation shows that $|\nabla|\gamma_{N_0,\xi}|\nabla|$ is trace-class on $\lpg$. Hence $\gamma_{N_0}$ belongs to $\mathcal{P}_{\mathrm{per},x}$. Let us show that $\rho_{\gamma_{N_0}} - \mups$ belongs to $\mathcal{C}_{\Gamma}$. Following calculations similar to the ones leading to~\eqref{FourierTTestState}, we only need to prove that $\bm k \mapsto |\bm k|^{-1}\F(\rho_{\gamma_{N_0}} -\mups )(0,\bm k) $ is square-integrable near $\bm k =\bm 0$ since $\rho_{\gamma_{N_0}} - \mups$ belongs to $\lpg$. Remark that $\F(\rho_{\gamma_{N_0}} -\mups )(0, \bm 0) =\int_{\Gamma}\rho_{\gamma_{N_0}}- \mups  = 0$ and 
\begin{align*}
&\left|\partial_{\bm k}\F(\rho_{\gamma_{N_0}} -\mups )(0, \bm 0)\right| =\left|\int_{\Gamma}\bm r\left(\rho_{\gamma_{N_0}}- \mups \right)(x,\bm r)\,dx\,d\bm r\right| \\
&\leq \frac{1}{2\pi}\sum_{n=1}^{N_H}\int_{\Gamma^*} \int_{\Gamma}|\bm r|\left|e_{n,\xi}\right|^2(x,\bm r)\,dx\,d\bm r\,d\xi  +  \sum_{n=N_0+1}^{N_0+Z_{\mathrm{diff}}}\int_{\bb R^3}|\bm r|\left| \psi_{n}\right|^2(x,\bm r)\,dx\,d\bm r\\
&\qquad \qquad+ \int_{\Gamma}|\bm r|\mups(x,\bm r)\,dx\,d\bm r< +\infty,
\end{align*}
where we have used the fact that the eigenfunctions of $\hpx$ associated with negative eigenvalues decay exponentially (see~\cite[Theorem 3.4]{hislop1995introduction} and~\cite[Theorem 1]{combes1973}) so that $\int_{\Gamma}|\bm r|\left|e_{n,\xi}\right|^2(x,\bm r)\,dx\,d\bm r <+\infty$ for~$1\leq n \leq N_H$ and $\xi\in\Gamma^*$, and the fact that $\left\{\psi_{n}\right\}_{N_0+1\leq n \leq N_0+Z_{\mathrm{diff}}}$ have compact support in the $\bm r$-direction by definition. Therefore $\F(\rho_{\gamma_{N_0}} -\mups )(0,\bm k)$ is $C^1$ near $\bm k = \bm 0$. The conclusion then follows by arguments similar to those leading to~\eqref{FourierTTestState} in Section~\ref{F_GammaNotEmptySec}.
\end{proof}
Lemma~\ref{TestStatesADMLemma} implies that we can construct many admissible states in $\cl{F}_{\Gamma}$ by varying $N_0$. Let us show that we can always find $N_0$ such that $\gamma_{N_0}$ decreases the ground state energy of~\eqref{pb1} if $N_H< \int_{\Gamma}\mups$.

Given a minimizer $\overline{\gamma} $ of~\eqref{pb1}, simple expansion of the energy functional around minimal shows that $\overline{\gamma}$ also minimizes the functional (see~\cite{Cances2008}) $$\gamma \mapsto \int_{\Gamma^*}\m{Tr}_{\lpg}\left(\hpx\gamma_\xi\right)d\xi$$ on $\mathcal{F}_{\Gamma}$. Therefore, given $N_0\in \bb N^*$ we have
\begin{equation}
\begin{aligned}
0&\leq \int_{\Gamma^*}\m{Tr}_{\lpg}\left(\hpx\left(\gamma_{N_0,\xi}-\overline{\gamma_\xi}\right)\right) d\xi  =\int_{\Gamma^*}\m{Tr}_{\lpg}\left(\rpxo\hpx\left(\gamma_{N_0,\xi}-\overline{\gamma_\xi}\right)\right) d\xi\\
&\qquad\qquad +\int_{\Gamma^*}\m{Tr}_{\lpg}\left((1-\rpxo)\hpx\left(\gamma_{N_0,\xi}-\overline{\gamma_\xi}\right)\right) d\xi \\
&=M+\int_{\Gamma^*}\m{Tr}_{\lpg}\left((1-\rpxo)\hpx \gamma_{N_0,\xi}\right) d\xi - \int_{\Gamma^*}  \m{Tr}_{\lpg}\left((1-\rpxo)\hpx\overline{\gamma_\xi}\right) d\xi,
\end{aligned}
\label{PositiveTestEIneq}
\end{equation}
where, since $0\leq \overline{\gamma_{\xi}}\leq 1$ and $\left\{\lambda_n(\cdot) \right\}_{1\leq n \leq N_H}<0$, 
\begin{equation}
\begin{aligned}
M&:=\int_{\Gamma^*}\m{Tr}_{\lpg}\left(\rpxo\hpx\left(\gamma_{N_0,\xi}-\overline{\gamma_\xi}\right)\right) d\xi=\int_{\Gamma^*}\sum_{ n=1}^{N_H}\lambda_{n}(\xi)\left\langle e_{n,\xi}\left| 1- \rpxo\overline{\gamma_{\xi}}\right| e_{n,\xi}\right\rangle  d\xi \\
&=\int_{\Gamma^*}\sum_{ n=1}^{N_H}\lambda_{n}(\xi)\left\langle e_{n,\xi}\left| 1- \overline{\gamma_{\xi}}\right| e_{n,\xi}\right\rangle  d\xi \leq 0.
\label{PositiveTestPart1}
\end{aligned}
\end{equation}
In view of~\eqref{WeylForScatteringStates} and by a Cauchy-Schwarz inequality, we deduce that, for $N_0\geq N_{\varepsilon}$: 
\begin{equation}
\begin{aligned}
&\int_{\Gamma^*}\m{Tr}_{\lpg}\left((1-\rpxo)\hpx \gamma_{N_0,\xi}\right) d\xi \\
&=  \int_{\Gamma^*}\sum_{n=N_0+1}^{N_0+Z_{\mathrm{diff}}} \sum_{m=1}^{+\infty} \left\langle e_{m,\xi}\left| \hpx\left| \psi_{n,\xi}\left\rangle \right\langle \psi_{n,\xi}\right.\right| e_{m,\xi}\right\rangle  d\xi\\
&\leq \int_{\Gamma^*}\sum_{n=N_0+1}^{N_0+Z_{\mathrm{diff}}} \left(\sum_{m=1}^{+\infty} \left| \left\langle\psi_{n,\xi}\left.\right| e_{m,\xi}\right\rangle\right|^2\right)^{1/2}\left(\sum_{m=1}^{+\infty} \left| \left\langle e_{m,\xi}\left| \hpx\right. \left| \psi_{n,\xi}\right\rangle \right. \right|^2\right)^{1/2}d\xi\\
&= \int_{\Gamma^*}\sum_{n=N_0+1}^{N_0+Z_{\mathrm{diff}}} \norm{\psi_{n,\xi}}_{\lpg} \bnorm{\hpx\psi_{n,\xi}}_{\lpg} d\xi  \\
&\leq 2\pi \sum_{n=N_0+1}^{N_0+Z_{\mathrm{diff}}}\left(\frac{1}{2\pi} \int_{\Gamma^*}\norm{\psi_{n,\xi}}_{\lpg}^2d\xi\right)^{1/2}\left(\frac{1}{2\pi} \int_{\Gamma^*} \bnorm{\hpx\psi_{n,\xi}}_{\lpg}^2 d\xi \right)^{1/2} \\
&\leq  2\sqrt{2}\pi \sum_{n=N_0+1}^{N_0+Z_{\mathrm{diff}}} \left(\frac{1}{n^4}  +\varepsilon^2\right)^{1/2} \leq 2\sqrt 2\pi Z_\mathrm{diff}\left(\frac{1}{N_0^4} +\varepsilon^2\right)^{1/2}.
\end{aligned}
\label{ScatterTestEnergy}
\end{equation}
Moreover, by definition of $\rpo$
\begin{equation}
\begin{aligned}
& \int_{\Gamma^*}\m{Tr}_{\lpg}\left((1-\rpxo)\hpx \overline{\gamma_\xi}\right)  = \int_{\Gamma^*}\m{Tr}_{\lpg}\left(\left|\hpx\right|^{1/2}  (1-\rpxo)\overline{\gamma_\xi}(1-\rpxo)\left|\hpx\right|^{1/2} \right) \geq 0.
\end{aligned}
\label{LastScatterEnergy}
\end{equation}
We distinguish in the inequality~\eqref{PositiveTestPart1} the cases $M \equiv 0$ or $M<0$. When $M\equiv 0$, the inequality~\eqref{PositiveTestPart1} implies that $\overline{\gamma_\xi} \rpxo = \rpxo $ for almost all $\xi\in\Gamma^*$. In view of the the inequalities~\eqref{ScatterTestEnergy} and~\eqref{LastScatterEnergy}, the inequality~\eqref{PositiveTestEIneq} implies that, 
\begin{equation}
\forall N_0\geq N_\varepsilon,\quad 0\leq \int_{\Gamma^*}\m{Tr}_{\lpg}\left((1-\rpxo)\hpx\overline{\gamma_\xi}\right) d\xi  \leq 2\sqrt 2 \pi Z_\mathrm{diff}\left(\frac{1}{N_0^4} +\varepsilon^2\right)^{1/2}.
\label{FirstTestIneqE}
\end{equation}
By letting $N_0$ tend to infinity, it is easy to deduce that $(1-\rpxo )\overline{\gamma_\xi} = 0$ for almost all $\xi\in\Gamma^*$. Together with the fact that $\overline{\gamma_\xi} \rpxo = \rpxo $ we deduce that $\overline{\gamma_\xi}\equiv \rpxo$ for almost all $\xi\in\Gamma^*$. In view of~\eqref{F0} and~\eqref{DiffCharge}, by the charge neutrality we obtain that 
$$ Z_{\mathrm{diff}} =\int_{\Gamma}\mups - N_H=\int_{\Gamma}\rho_{\overline{\gamma}} -  N_H =  \frac{1}{|\Gamma^*|}\int_{\Gamma^*}\m{Tr}_{\lpg}\left(\rpxo\right)d\xi- N_H \equiv 0.$$
Hence $ \int_{\Gamma}\mups =N_H  = F(0)$. This also implies that the minimizer of the problem~\eqref{pb1} equals to $\rpo$ when $N_H= F(0)=\int_{\Gamma}\mups$ and $Z_{\m{diff}}\equiv 0$. When $M<0$ and $Z_{\mathrm{diff}}\neq 0$, we can always find $\varepsilon >0$ and $N_0\geq N_{\varepsilon}$ such that $ 2\sqrt 2\pi Z_\mathrm{diff}\left(\frac{1}{N_0^4} +\varepsilon^2\right)^{1/2}\leq -M/2$. In view of the the inequalities~\eqref{ScatterTestEnergy} and~\eqref{LastScatterEnergy}, the inequality~\eqref{PositiveTestEIneq} implies that
$$\forall N_0\geq N_\varepsilon,\quad 0\leq \int_{\Gamma^*}\m{Tr}_{\lpg}\left((1-\rpxo)\hpx\overline{\gamma_\xi}\right) d\xi  \leq M/2 < 0,$$
which leads to contradiction. We can finally conclude that $F(0) \geq  \int_{\Gamma}\mups$, so that the Fermi level of the quasi 1D system is always non-positive. In the following paragraph we show that the Fermi level can be chosen to be strictly negative.
\paragraph{Form of the minimizer and decay of the density of minimizers.}
We have already shown that if $N_H = F(0)\equiv \int_{\Gamma}\mups $ then the 1D system has a unique minimizer which is equal to $\rpo$. This also implies that for almost all $\xi\in\Gamma^*$, the operator $\hpx$ has $N_H$ strictly negative eigenvalues below $0$, therefore we can always choose the Fermi level $\epsilon_F  \in\left( \max_{\xi\in \Gamma^*}\lambda_{N_H}(\xi), 0\right) $. If $ F(0)>\int_{\Gamma} \mups$, it is clear that there exists $\epsilon_F<0$ such that $F(\epsilon_F) =\int_{\Gamma}\mups$ as $F(\kappa)$ is a non-decreasing function on $(-\infty,0]$ with range in $[0,F(0)]$. The form of the minimizer and the uniqueness is a direct adaptation of \cite[Theorem 1]{Cances2008} by using a spectral projector decomposition similar to~(A.2) of \cite[Theorem 1]{Cances2008}, that is, the unique minimizer can be written as
\begin{equation}
\begin{aligned}
\rp &=\mathds{1}_{(-\infty,\epsilon_F]}(\hp)= \B^{-1}\left(\int_{\Gamma^*}\rpx\,\frac{d\xi}{2\pi}\right)\B= \B^{-1}\left(\int_{\Gamma^*}\sum_{ n=1}^{N_H}\mathds 1(\lambda_n(\xi)\leq \epsilon_F) \left| e_{n,\xi}\left\rangle \right\langle e_{n,\xi}\right|  \right)\B,
\end{aligned}
\label{FormMinimzer}
\end{equation}
where $\rpx:=  \mathds{1}_{(-\infty,\epsilon_F]}(\hpx)$. The Fermi level $\epsilon_F< 0$ can be considered as the Lagrange multiplier associated with the charge neutrality condition 
$$F(\epsilon_F) = \int_{\Gamma}\rho_{\rp} = \int_{\Gamma} \mups.$$
\begin{detail}
{\color{red}{The following detail is: change to another Bloch transform to define a Kato-analytical family operators with constant domain in $L_{\mathrm{per}}^2$ and shows that indeed the eigenvalues are analytic in $\xi$, and cannot be constant, therefore $\left|\left\{\xi\in\Gamma^*\left|\, \lambda\in \sigma_{\mathrm{disc}}(\hpx)\right.\right\}\right| \equiv 0$.  }}
Introduce $\bb Z$-periodic $L^p$ spaces $$L_{\mathrm{per}}^p(\Gamma):= \left\{\phi\in  L_{\mathrm{loc}}^p(\mathbb{R}, L^p(\mathbb{R}^2))\left|\, \phi(x,\cdot)= \phi(x+k,\cdot)\text{ , $\forall k\in \bb Z$}\right.\right\},$$
endowed with norm 
$\norm{\phi}_{L_{\mathrm{per}}^p(\Gamma)} := \left(\int_{\Gamma}|\phi|^p\right)^{1/p}$ for $1\leq p<+\infty$ and $\norm{\phi}_{L_{\mathrm{per}}^{\infty}(\Gamma)} := \text{ess-sup }|\phi|$. Define $\widetilde{L^2(\Gamma^*;\lpg)}:= \int_{\Gamma^*}\lpg\,d\xi.$
We introduce an alternative Bloch decomposition $\widetilde{\B}: L^2(\bb R^3)\to  \widetilde{L^2(\Gamma^*;\lpg)} $, which is an isometry $$\forall \phi \in \mathscr{S}(\bb R^3),\, \text{ for a.e. }(x,\bm r)\in \Gamma, \, \xi \in \Gamma^* \quad \left(\widetilde{\B}\phi\right)_{\xi}(x,\bm r) = \left(\B\phi\right)_{\xi}(x,\bm r)\cdot \rme^{-\rmi x\xi}= \sum_{k\in \bb Z}\rme^{-\rmi (k+x)\xi}\phi(x+k,\bm r).$$
For $\xi \mapsto f_{\xi}\in \widetilde{L^2(\Gamma^*;\lpg)}$, the inverse of $\widetilde{\B}$ is defined to be
$$\forall k\in\bb Z,\,\text{ for a.e. }(x,\bm r)\in \Gamma,\quad \left(\widetilde{\B}^{-1}f_{\xi}\right)(x+k,\bm r) =\left(\B^{-1}f_{\xi}\rme^{\rmi x\xi}\right)(x+k,\bm r)= \int_{\Gamma^*}\rme^{\rmi (k+x)\xi}f_{\xi}(x,\bm r)\frac{d\xi}{2\pi}. $$
With $\widetilde{\B}$, we have an alternative decomposition 
$$\hp = \widetilde{\B}^{-1}\left(\int_{\Gamma^*}\widetilde{\hpx}\,\frac{d\xi}{2\pi}\right)\widetilde{\B},\quad \widetilde{\hpx} = -\frac{1}{2}\Delta_{\mathrm{per}}- \rmi \xi\frac{\partial}{\partial_x}+\frac{\xi^2}{2}+\vp,$$
where $ \Delta_{\mathrm{per}}$ is the Laplacian operator defined on $\lpg$ with periodic boundary conditions. The reason for the introduction of $\widetilde{\B}$ is that under $\widetilde{\B}$, the operator $ \left(\widetilde{\hpx}\right)_{\xi\in\Gamma^*} $ has a constant domain in $\lpg$, which leads to the characterization of analytic family in the sense of Kato (See, for ex \cite[Section XII.2]{ReeSim4}). An application of perturbation theory of the isolated eigenvalues (not necessarily non-degenerate) \cite[Theorem XII.13]{ReeSim4} for the analytic family $\left(\widetilde{\hpx}\right)_{\xi\in\Gamma^*} $ indicates that if $\lambda \in \sigma_{\mathrm{disc}}(\widetilde{H_{\mathrm{per},\xi_0}})$ for some $\xi_0\in \Gamma^*$ with multiplicity $m$, then there are $m$ not necessarily distinct single-valued functions $\lambda^{(1)}(\xi)\cdots \lambda^{(m)}(\xi)$, analytic near $\xi_0$ which are eigenvalues of $\widetilde{\hpx}$. By an argument which is similar to \cite[Lemma 2 of Theorem XIII.99]{ReeSim4}, $\lambda^{i}(\xi)$ cannot be constant({\color{red}{cannot see the impact of compact resolvent condition}}). By the analyticity and non-constancy of the eigenvalue of $\widetilde{\hpx}$, we can see that $\left\{\xi\in\Gamma^*\left|\, \lambda\in \sigma_{\mathrm{disc}}(\widetilde{\hpx})\right.\right\}$ has Lebesgue measure $0$. Therefore $\sigma_{\mathrm{disc}}(\hp)\equiv \emptyset$. Moreover, it is easy to see that if $\lambda$ is an eigenvalue associated with an eigenfunction $\phi$ of $\hpx$ then $\phi(x,\bm r) \rme^{-\rmi x\xi }$ is an eigenfunction of $\widetilde{\hpx}$ with eigenvalue $\lambda$, therefore $\cup_{\xi\in\Gamma^*}\sigma_{\mathrm{disc}}\left(\hpx\right) = \cup_{\xi\in\Gamma^*}\sigma_{\mathrm{disc}}\left(\widetilde{\hpx}\right)$, and
\end{detail}
Once the unique minimizer is shown to be a spectral projector, we can use the exponential decay property of the eigenfunctions of $\hpx$ in the $\bm r$-direction via the Combe--Thomas estimate~\cite[Theorem 1]{combes1973}: for almost all $\xi\in\Gamma^*$, there exist positive constant $C(\xi)$ and $\alpha(\xi)$ such that $$\forall (x,\bm r)\in \Gamma, \,\,\, \forall \, 1\leq n\leq N_H,\qquad \left|e_{n,\xi}(x,\bm r)\right| \leq C(\xi)\rme^{-\alpha(\xi)|\bm r|}.$$
On the other hand, the fact that $\int_{\Gamma}\mups  = F(\epsilon_F) <+\infty$ implies that there exist only finitely many states of $\hpx$ below $\epsilon_F$ for all $\xi\in\Gamma^*$. Therefore there exist positive constants $C_{\epsilon_F}$ and $\alpha_{\epsilon_F}$ such that 
\begin{align*}
0\leq \rho_{\rp} (x,\bm r)  \leq  \frac{1}{2\pi}\int_{\Gamma^*}\sum_{ n=1}^{N_H}\mathds 1(\lambda_n(\xi)\leq \epsilon_F)  C^2(\xi)\rme^{-2\alpha(\xi)|\bm r|} \,d\xi \leq  C_{\epsilon_F}\rme^{-\alpha_{\epsilon_F} |\bm r|}.
\end{align*}
Remark that this exponential decay property coincides with Assumption~\ref{assumption_density} that $\int_{\Gamma}|\bm r|\rho_{\rp}(x,\bm r)\,dx\,d\bm r <+\infty$.

\subsection{Proof of Lemma~\ref{symmetryPotentialLemma}}
\label{symmetryPotentialLemmaSec}
Assume that~\eqref{absenceDipole} holds, that is $\mup(x, \bm r)\equiv \mup(x, |\bm r|)$ has radial symmetry in the $\bm r$-direction. It is clear that the results of Theorem~\ref{thmPeriodicExistence2} hold. We employ the same notations as in Theorem~\ref{thmPeriodicExistence2} in the sequel. By the uniqueness of density, $\rho_{\rp}$ enjoys the same radial symmetry in the $\bm r$-direction. Recall that $\ep =\rho_{\rp} -\mups $. Together with the facts that $\int_{\Gamma}|\bm r|\rho_{\rp}(x,\bm r)\,dx\,d\bm r <+\infty$ and that $\mup$ has compact support in the $\bm r$-direction. The radial symmetry in the $\bm r$-direction implies that
\begin{equation}
\label{symmetry_rp}
\int_{\Gamma}\bm r\cdot \ep(x,\bm r)\,dx\,d\bm r \equiv 0.
\end{equation}
Remark also that the exponential decay of density implies that $\int_{\Gamma}|\bm r|^2 |\ep(x,\bm r)|\,dx\,d\bm r <+\infty.$ Following calculations similar to the those in~\eqref{majoration_partiale}, ~\eqref{eq:FourierofV} and~\eqref{vpEstimateLP}, it is easy to deduce that $\partial_{\bm k}\F\ep(0,\bm k)\equiv 0$ and $\partial_{\bm k}^2\F\ep(0,\bm k)$ is continuous and bounded, so that $\F\vp (0,\cdot)$ belongs to $L^2(\R^2)$, and $\vp$ also belongs to $\lpg$. Let us prove that $\vp \in L_{\mathrm{per},x}^p(\Gamma)$ for $1<p < 2$ (for which we can conclude that~$\vp $ belongs to $ L_{\mathrm{per},x}^p(\Gamma)$ for $1< p\leq +\infty$). Let us rewrite~$\vp$ as
\begin{equation}
\begin{aligned}
\vp(x, \bm r)  &= \left(\ep  \star_{\Gamma}G\right)(x,\bm r)
= \left( \ep  \star_{\Gamma}\widetilde{G}\right)(x,\bm r) +T(\bm r),
\end{aligned}
\label{VpDecomposed}
\end{equation}
where $$T(\bm r)= -2 \int_{\bb R^2}\epx(\bm r') \log\left(|\bm r-\bm r'|\right)\, d\bm r',\quad  \epx(\bm r) := \int_{-1/2}^{1/2} \ep (x,\bm r) \,dx.$$
Recall that $\mups$ has compact support in the $\bm r$-direction, hence there exist positive constants $C_{q}, \alpha_{q}$ such that
$$
\forall (x,\bm r)\in \R^3,\quad |\ep(\cdot,\bm r)| \leq C_{q}\rme^{-\alpha_{q} |\bm r|},\quad  |\epx(\bm r) |\leq C_{q}\rme^{-\alpha_{q} |\bm r|}.
$$
As $\widetilde{G}$ belongs to $ L_{\mathrm{per},x}^p(\Gamma) $ for $1\leq p<2$, by Young's convolution inequality we deduce that $ \ep \star_{\Gamma}\widetilde{G}$ belongs to $ L_{\mathrm{per},x}^t(\Gamma) $ for $1\leq t \leq +\infty$. 

It remains to prove that $T(\bm r) $ belongs to $ L^p(\bb R^2) $ for $ 1<p<2$. Let us use the partition $\bb R^2 = \left\{|\bm r| \leq 2R \right\}\cup \left\{|\bm r| >  2R
\right\}$ for the integration domain of $T(\bm r)$. Note first that $ \log\left(\left|\bm r \right|\right)$ is $L_{\mathrm{loc}}^t(\bb R^2)$ for $1\leq t < +\infty$. Therefore, by a Cauchy-Schwarz inequality, there exists a positive constant $C_{R,1} $ such that for $p' = p/(p-1) \in (2,+\infty)$: 
\begin{equation}
\begin{aligned}
\label{CR1}
\left(\int_{|\bm r| \leq 2R}\left|T(\bm r)\right|^p d\bm r \right)^{1/p}&=2\left(\int_{|\bm r| \leq 2 R}\left| \int_{\bb R^2}\epx(\bm r') \log\left(\left|\bm r- \bm r'\right|\right)  d\bm r'\right|^p d\bm r\right)^{1/p} \\
&\leq 2\left( \int_{|\bm r| \leq 2R}\left|\int_{|\bm r'| \leq 3R}\epx(\bm r') \log\left(\left|\bm r- \bm r'\right|\right)  d\bm r' \right|^p d\bm r\right)^{1/p}\\
&\qquad+2C_{q} \left(\int_{|\bm r| \leq 2R}\left|\int_{|\bm r'| > 3R}\rme^{-\alpha_{q} |\bm r'|}\left|\log\left(\left|\bm r- \bm r'\right|\right)\right| d\bm r' \right|^p d\bm r \right)^{1/p}\\
& \leq 2\left(\int_{|\bm r'| \leq 3R}\left|\epx\right|^p\right)^{1/p}\left( \int_{|\bm r| \leq 2R}\left( \int_{|\bm r'| \leq 3R}\left|\log\left(\left|\bm r- \bm r'\right|\right)\right|^{p'}  d\bm r' \right)^{p/p'} d\bm r\right)^{1/p}\\
&\qquad+2C_{q} \left(\int_{|\bm r| \leq 2R}\left|\int_{|\bm r'| > 3R}\rme^{-\alpha_{q} |\bm r'|} \left|\log(|(R, R)-\bm r'|)\right| d\bm r' \right|^p d\bm r \right)^{1/p}\leq C_{R,1}.
\end{aligned}
\end{equation}
Let us look at the integration domain $\left\{ |\bm r| >2R\right\}$. Remark that by the charge neutrality condition and the radial symmetry condition~\eqref{symmetry_rp}, it holds, for any $\bm r\neq 0$, $$ \int_{\bb R^2}\epx(\bm r')\log\left(\left|\bm r \right|\right)\,d\bm r' \equiv 0,\quad \int_{\bb R^2}\epx(\bm r')\frac{\bm r'\bm r}{|\bm r|^2}\,d\bm r'\equiv 0 .$$ Denote by 
$$Q(\bm r, \bm r'):= \log\left(\left|\bm r- \bm r'\right|\right) -\log\left(\left|\bm r \right|\right)- \frac{\bm r'\bm r}{|\bm r|^2}= \frac{1}{2}\log\left(1 - \frac{2\bm r\bm r'}{|\bm r|^2} + \frac{|\bm r'|^2}{|\bm r|^2}\right) - \frac{\bm r'\bm r}{|\bm r|^2}.$$
Then $T(\bm r)= -2 \int_{\bb R^2}\epx(\bm r') Q(\bm r, \bm r')\, d\bm r'. $ Remark that when $|\bm r| >2R$ and $\frac{|\bm r'|}{|\bm r|} \leq \varepsilon_R$ for $\varepsilon_R >0$ fixed. A Taylor expansion shows that there exists a positive constant $C$ such that
$\left|Q(\bm r, \bm r')\right| \leq  C \frac{|\bm r'|^2}{|\bm r|^2}$. This motivates the following partition of $\bb R^2$ given $|\bm r| >2R$: $$\bb R^2=\bb B_{\varepsilon_R} \cup \bb B_{\varepsilon_R}^{\complement}, \quad\bb B_{\varepsilon_R}:= \left\{\bm r' \in \R^2\left|\, \frac{|\bm r'|}{|\bm r|} \leq \varepsilon_R \right. \right\}.$$
Hence $$T(\bm r) =T_{\m{int}}(\bm r) + T_{\m{ext}}(\bm r),\quad  T_{\m{int}}(\bm r):= 
\int_{ \bb B_{\varepsilon_R}}\epx(\bm r') Q(\bm r, \bm r')d\bm r', \qquad T_\m{ext}(\bm r):= \int_{ \bb B_{\varepsilon_R}^{\complement}}\epx(\bm r') Q(\bm r, \bm r') d\bm r'.
$$
Therefore, for $1<p<2$,
\begin{equation}
\label{CR2}
\begin{aligned}
\int_{|\bm r| > 2R}\left|T_\m{int}(\bm r)\right|^p d\bm r 
&\leq 2C^p\int_{|\bm r| >2 R}\left| \int_{\bb B_{\varepsilon_R}}\epx(\bm r') \frac{|\bm r'|^2}{|\bm r|^2}  d\bm r'\right|^p d\bm r\\
& \leq 2C' \int_{|\bm r| >2 R}\left| \int_{ |\bm r'|\leq \varepsilon_R |\bm r|} \rme^{ -\alpha_{q}  |\bm r'| } |\bm r'|^2 d\bm r'\right|^p |\bm r|^{-2p} d\bm r <+\infty.
\end{aligned}
\end{equation}
Similarly,
\begin{equation}
\label{CR3}
\begin{aligned}
\int_{|\bm r| > 2R}\left|T_\m{ext}(\bm r)\right|^p d\bm r &\leq C_1 \int_{|\bm r| >2 R}\left| \int_{|\bm r'| > \varepsilon_R |\bm r|}\rme^{-\alpha_{q} |\bm r'|} \left|Q(\bm r, \bm r')\right|  d\bm r'\right|^p d\bm r \\
&\leq C_1 \int_{|\bm r| >2 R}\left| \int_{|\bm r'| > \varepsilon_R |\bm r|}\rme^{-\alpha_{q}\varepsilon_R |\bm r|/2} \rme^{-\alpha_{q} |\bm r'|/2} \left|Q(\bm r, \bm r')\right|  d\bm r'\right|^p d\bm r <+\infty.
\end{aligned}
\end{equation}
In view of~\eqref{CR1}, ~\eqref{CR2} and~\eqref{CR3} we conclude that $T(\bm r)$ belongs to $L^p(\bb R^2)$ for $1<p<2$. This leads to the conclusion that~$\vp$ belongs to $L_{\mathrm{per},x}^p(\Gamma)$ for $1<p\leq +\infty$.

\subsection{Proof of Proposition \ref{spectrapPpHchi}}
\label{spectrapPpHchiSec}
Let us emphasize that the function $\chi$ being translation-invariant in the $\bm r$-direction makes it difficult to control the compactness in the $\bm r$-direction across the junction surface. Our geometry is very different from the cylindrical geometry considered in~\cite{Hempel2014} for instance which automatically provides compactness in the $\bm r$-direction.

The proof of $\sigma_{\mathrm{ess}}\left(\hl\right)\cup  \sigma_{\mathrm{ess}}\left(\hr\right)\subseteq \sigma_{\mathrm{ess}}(\hc)$ is relatively easier than the converse inclusion. Intuitively, the $a_L$-periodicity (resp. $a_R$-periodicity) implies that the Weyl sequences of $\hl$ (resp. $\hr$) after a translation by $na_L$ (resp. $na_R$-translation) are still Weyl sequences of $\hl$ (resp. $\hr$) for any $n\in \bb Z$. This suggests that one can construct Weyl sequences of $\hc$ by properly translating Weyl sequences of $\hl $ or $\hr$, as $\hc$ is a linear interpolation of $\hl$ and $\hr$ hence behaves like $\hl$ or $\hr$ away from the junction surface. The construction of Weyl sequences of either $\hl$ or $\hr$ from Weyl sequences of $\hc$ is much more difficult, as the support of Weyl sequences is essentially away from any compact set, making it difficult to control their behaviors. One naive approach should be to cut-off Weyl sequences of $\hc$ by the function $\chi$ in order to construct Weyl sequences of $\hl$. However, one quickly remarks that the commutator $[-\Delta, \chi]$ is not $-\Delta$-compact, hence it is difficult to ensure that the cut-off sequence is a Weyl sequence of $\hl$. Another naive approach is to use a cut-off function $\chi_c$ which has compact support in the $\bm r$-direction. However as mentioned before, it is also difficult to ensure that the Weyl sequences leave any mass in the support of $\chi_c$ as Weyl sequences essentially have supports away from any compact. We construct a special Weyl sequence of $\hc$ from $\hl$ (or $\hr$) by a suitable cut-off function which introduces $-\Delta$-compact perturbations to solve this difficulty.

The proof is organized as follows: we first prove that $[0,+\infty) \subset \sigma_{\mathrm{ess}}(\hc)$ by an explicit construction of Weyl sequences with positive energies, as $[0,+\infty)$ belongs to the essential spectrum of $\hl$ and $\hr$. We next prove for any $\lambda <0$ such that~$\lambda\in \sigma_{\mathrm{ess}}\left(\hl\right)\cup  \sigma_{\mathrm{ess}}\left(\hr\right)$, $\lambda$ also belongs to $ \sigma_{\mathrm{ess}}(\hc)$. Finally we prove that $\sigma_{\mathrm{ess}}(\hc)$ in included in~$\sigma_{\mathrm{ess}}\left(\hl\right)\cup  \sigma_{\mathrm{ess}}\left(\hr\right)$ by a rather technical construction of a Weyl sequence. 
\paragraph{The essential spectrum $\sigma_{\mathrm{ess}}(\hc)$ contains $[0,+\infty)$. }
Recall that a sequence $\{\psi_n\}_{n\in \bb N^*}$ is a Weyl sequence of an operator $O$ on $L^2(\bb R^3)$ associated with $\lambda\in \R$ if the following properties are satisfied:
\begin{itemize}
\item for all $n$, $\psi_n$ is in the domain of the operator $O$ and $\norm{\psi_n}_{L^2(\bb R^3)}= 1$;
\item the sequence $\psi_n \rightharpoonup 0$ weakly in $L^2(\bb R^3)$;
\item $\displaystyle\lim_{n\to+\infty} \norm{\left(O-\lambda\right)\psi_n}_{L^2(\bb R^3)}= 0$.
\end{itemize}
Consider a $C_c^{\infty}(\bb R)$ function $f(z)$ supported on $[0,1]$ with $\int_{\bb R}|f|^2 =1$, and $g\in C_c (\bb R^2)$ supported on the unit disk centered at $0$ and such that $\int_{\bb R^2}|g|^2 =1$.
For any $\lambda>0$, consider a sequence of functions $\{\psi_{n}\}_{n\in \bb N^*}$ defined as follows: $$\psi_n(x,y,z):= n^{-3/2}\rme^{\rmi \sqrt{\lambda}z}f((z-2n)/n)g(x/n,y/n).$$
It is easy to see that $\int_{\bb R^3}|\psi_n|^2 = 1$ for all $n\in \bb N^*$, and that $\psi_n$ tends weakly to $0$ in $L^2(\bb R^3)$. 
On the other hand
\begin{align*}
\left(\hc - \lambda\right)\psi_n(x,y,z) =& \left(-n^{-2}f''((z-2n)/n)-2\rmi \sqrt{\lambda}n^{-1}f'((z-2n)/n)\right) n^{-3/2}\rme^{\rmi\sqrt{\lambda}z}g(x/n,y/n)\\
&-n^{-7/2}\rme^{\rmi \sqrt{\lambda}z}f((z-2n)/n)\left(\Delta g\right)(x/n,y/n) + V_{\chi}\psi_n(x,y,z) .
\end{align*}
Recall also that by the results of Theorem~\ref{thmPeriodicExistence2}, there exists for any $\epsilon >0$ an integer $N_{\epsilon}$ such that $|V_\chi (x,y, 2n)| \leq \epsilon$ when $n\geq N_{\epsilon}$. Therefore, for $n\geq N_{\epsilon}$, there exists a positive constant $C$ such that
\begin{align*}
\norm{\left(\hc - \lambda\right)\psi_n}_{L^2(\bb R^3)} &\leq \frac{1}{n^2}\norm{f''}_{L^2(\bb R)}\norm{g}_{L^2(\bb R^2)} + \frac{2\sqrt{\lambda}}{n}\norm{f'}_{L^2(\bb R)}\norm{g}_{L^2(\bb R^2)} + \frac{1}{n^2}\norm{f}_{L^2(\bb R)}\norm{\Delta g}_{L^2(\bb R^2)}\\
& \quad+\left(n^{-3} \int_{\bb R^2}\left(\int_{2n}^{3n}\left|V_{\chi}(x,y,z) f((z-2n)/n)\right|^2\,dz\right)\left| g(x/n,y/n)\right|^2\,dx\,dy\right)^{1/2}\\
&\leq \frac{C}{n} +\left( \int_{\bb R^2}\int_0^1\left|V_{\chi}(nx,ny,2n +nz) f(z)\right|^2\left| g(x,y)\right|^2\,dz \,dx\,dy\right)^{1/2}\\
&\leq \frac{C}{n}+ \epsilon \norm{f}_{L^2(\bb R)}\norm{g}_{L^2(\bb R^2)} = \frac{C}{n} + \epsilon. 
\end{align*}
This shows that $\{\psi_{n}\}_{n\in \bb N^*}$ is a Weyl sequence of $\hc$ associated with $\lambda >0$. Since that $0$ is an accumulation point of $\sigma_{\mathrm{ess}}(\hc)$, it holds $[0,+\infty) \subset \sigma_{\mathrm{ess}}(\hc)$.
\paragraph{The union of $\sigma_{\mathrm{ess}}\left(\hl\right)\cup  \sigma_{\mathrm{ess}}\left(\hr\right)$ is included in $\sigma_{\mathrm{ess}}(\hc)$.}
Without loss of generality we prove that a negative $\lambda_{L} $ belonging to $\sigma_{\mathrm{ess}}(\hl)$ also belongs to $ \sigma_{\mathrm{ess}}(\hc)$. Consider a Weyl sequence $\left\{w_{n}\right\}_{n\in\bb N^*}$ for $\hl$ associated with $\lambda_{L}$. Let us construct a Weyl sequence for $H_{\chi}$ from $\left\{w_{n}\right\}_{n\in\bb N^*}$. Fix $n\in \bb N^*$, there exists a sequence $\left\{v_{k,n}\right\}_{k\in \bb N^*}$ belonging to $C_c^{\infty}(\bb R^3)$ such that for all $\varepsilon>0$, there exists a $K_n \in \bb N^*$ such that for any $k \geq K_n$,
\begin{equation}
\norm{v_{k,n} - w_n}_{H^2(\bb R^3)} \leq \varepsilon.
\label{eq1_Weyl}
\end{equation}
It is easy to see that $v_{K_n,n}$ tends weakly to $0$ in $L^2(\bb R^3)$ as $n \to \infty$ since $w_n$ converges weakly to $0$. As $v_{K_n,n}$ has compact support, for any fixed $n\in \bb N^*$ and for $m \in \bb N^*$ large enough,
\begin{equation}
\supp\left(\tau_{a_{L}m}^xv_{K_n,n}\right)\bigcap \left(\left(\left[-a_{L}/2,+\infty\right)\times \bb R^2\right)\bigcup \mathfrak{B}_{n}\right) = \emptyset,
\label{eq2_Weyl}
\end{equation}
where $\mathfrak{B}_{n}$ denotes the ball of radius $n$ centered at $0$ in $\bb R^3$. Remark that the above equality also ensures that $\tau_{a_{L}m}^xv_{K_n,n}$ tends weakly to $0$ in $L^2(\bb R^3)$ when $m\to+\infty$ for $n$ fixed. In view of~\eqref{eq1_Weyl} and~\eqref{eq2_Weyl}, we introduce $\widetilde{w_n} := \tau_{a_{L}m_n}^xv_{K_n, n}$ for $n \in \bb N^*$ so that~\eqref{eq2_Weyl} is satisfied. This implies that $\widetilde{w_n}$ tends weakly to $0$ in $L^2(\bb R^3)$ when $n \to +\infty$. Moreover, in view of~\eqref{eq1_Weyl} and by the definition of the Weyl sequence 
\begin{align*}
\bnorm{(\hc - \lambda_L)\widetilde{w_n}}_{L^2} &= \bnorm{(\hc - \lambda_L)\tau_{a_{L}m_n}^xv_{K_n,n}}_{L^2}  = \bnorm{(\hl- \lambda_L)\tau_{a_{L}m_n}^xv_{K_n,n}}_{L^2} \\
& \leq \bnorm{\tau_{a_{L}m_n}^x(\hl- \lambda_L)\left(v_{K_n,n} - w_{n}\right)}_{L^2} + \bnorm{\tau_{a_{L}m_n}^x(\hl- \lambda_L) w_{n}}_{L^2} \\
&\leq  \left(1+\norm{\vl}_{L^{\infty}} +| \lambda_L| \right)\bnorm{v_{K_n,n} - w_{n}}_{H^2} + \bnorm{(\hl- \lambda_L) w_{n}}_{L^2}\xrightarrow[n \to +\infty]{} 0.
\end{align*}
Therefore the sequence $\widetilde{w_n}/\norm{\widetilde{w_n}}_{L^2}$ is a Weyl sequence of $\hc$ associated with $\lambda_L$. This leads to the conclusion that
\begin{equation}
\sigma_{\mathrm{ess}}\left(\hl\right)\cup  \sigma_{\mathrm{ess}}\left(\hr\right) \subseteq \sigma_{\mathrm{ess}}(\hc).
\label{spHCHI2}
\end{equation}
\paragraph{The essential spectrum $\sigma_{\mathrm{ess}}(\hc)$ in included in $\sigma_{\mathrm{ess}}\left(\hl\right)\cup  \sigma_{\mathrm{ess}}\left(\hr\right)$. }
We prove that for strictly negative $\lambda\in \sigma_{\mathrm{ess}}(\hc)$, it holds that $\lambda\in\sigma_{\mathrm{ess}}\left(\hl\right)\cup  \sigma_{\mathrm{ess}}\left(\hr\right)$. The main technique is to use spreading sequences (Zhislin sequences)~\cite[Definition 5.12]{stephen2003mathematical}, which are special Weyl sequences for which the supports of the functions move off to infinity. More precisely, a sequence $\{\psi_n\}_{n\in \bb N^*}$ is a spreading sequence of an operator $O$ on $L^2(\bb R^3)$ associated with $\lambda$ if the following properties are satisfied:
\begin{itemize}
\item for all $n$, $\psi_n$ is in the domain of the operator $O$ and $\norm{\psi_n}_{L^2(\bb R^3)}= 1$;
\item for any bounded set $\mathcal{G}\subset \bb R^3$, $ \supp(\psi_n)\cap \mathcal{G} = \emptyset$ for $n$ sufficiently large. As a consequence, $\psi_n \rightharpoonup 0$ weakly in $L^2(\bb R^3)$;
\item $\displaystyle \lim_{n\to+\infty} \norm{\left(O-\lambda\right)\psi_n}_{L^2(\bb R^3)}= 0$.
\end{itemize}
As $\vl$ and $\vr$ are continuous and bounded by Theorem~\ref{thmPeriodicExistence2}, it holds by~\cite[Theorem 5.14]{stephen2003mathematical} that $$\sigma_{\mathrm{ess}}\left(H_{\theta}\right) = \left\{\lambda \in \bb C \left|\, \text{there is a spreading sequence for $H_{\theta}$ and $\lambda$}\right.\right\}$$
with $H_{\theta}$ being $\hc$, $\hl$ or $\hr$. The results of Theorem~\ref{thmPeriodicExistence2} also imply that, for any $\varepsilon>0$, there exists a constant $\mathcal{R}_{\varepsilon}$ such that 
\begin{equation}
\max\left(\norm{\vl \mathds 1_{|\bm r|>\mathcal{R}_{\varepsilon}}}_{L^{\infty}(\bb R^3)} , \norm{\vr  \mathds 1_{|\bm r|>\mathcal{R}_{\varepsilon}}}_{L^{\infty}(\bb R^3)} \right) <\varepsilon.
\label{maxVdecay}
\end{equation}
Consider a spreading sequence $\left\{\phi_n\right\}_{n\in\bb N^*}$ for $\hc$ associated with $\lambda<0$. For all $n\in \bb N^*$, either $\norm{\phi_{n}}_{L^2\left((0,+\infty)\times \bb R^2\right)}\geq 1/2$ or $\norm{\phi_{n}}_{L^2\left((-\infty, 0]\times \bb R^2\right)}\geq 1/2$. Without loss of generality, we assume in the sequel that there exists a sub-sequence $\left\{\phi_{n}\right\}_{n\in\bb N^*}$ such that, for $n$ sufficiently large, $\norm{\phi_{n}}_{L^2\left((0,+\infty)\times \bb R^2\right)}\geq 1/2$. We next construct a Weyl sequence of $\hr$ from $\left\{\phi_{n}\right\}_{n\in\bb N^*}$ by constructing a special cut-off function $\rho$ which has a non-trivial mass on $(0,+\infty)\times \bb R^2$, and whose derivatives decay rapidly: $$\rho(x):= \frac{\int_{-\infty}^x\eta(y)\,dy}{\norm{\eta}_{L^1(\bb R)}} ,$$ where $\eta$ satisfies the following one-dimensional Yukawa equation
\begin{equation}
-\eta''-\lambda\eta = \rme^{-2\sqrt{-\lambda}| x|}.
\label{eq:Yuka1}
\end{equation}
The following lemma summarizes some properties of the cut-off function $\rho$.
\begin{lemma}
\label{lemma_approx}
It holds that $0\leq \rho\leq 1$ and $\lim_{n\to\infty}\norm{\rho \phi_{n}}_{L^2}\neq 0 $. Moreover, 
\begin{equation}
\label{eq:approxrho}
\norm{\rho''\phi_{n} +2\rho'\partial_x\phi_{n}}_{L^2} \xrightarrow[n\to \infty]{} 0,\qquad \rho \chi^2\left(\vl-\vr\right) \in L^2(\bb R^3).
\end{equation}
\end{lemma}
We postpone the proof of this lemma to Section~\ref{lemma_approxSec}. Define $w_n:= \rho \phi_{n}/\norm{\rho \phi_{n}}_{L^2}$. Let us show that $\{w_n\}_{n\in \bb N^*}$ is a spreading sequence of $\hr$ associated with $\lambda$. First of all, the sequence $\{w_n\}_{n\in \bb N^*}$ is well defined at least for large $n$ as $\lim_{n\to\infty}\norm{\rho \phi_{n}}_{L^2}\neq 0 $. It is also easy to see that $\norm{w_n}_{L^2} = 1$ for all $n\in \bb N^*$. For any bounded set $\mathcal{G}\in \bb R^3$, it holds $ \supp(w_n)\cap \mathcal{G} = \emptyset$ for $n$ sufficiently large as $\left\{\phi_{n}\right\}_{n\in\bb N^*}$ is a spreading sequence. Note that 
\begin{equation}
\label{decompose_op}
\left(\hr-\lambda\right)w_n = \frac{1}{\norm{\rho \phi_{n}}_{L^2}}\left(\rho\left(\hc-\lambda \right)\phi_{n} +A_\chi\phi_{n}\right),
\end{equation}
where $A_\chi:=  -\frac{1}{2}\rho''  -\rho'\partial_x +\rho \chi^2\left(\vl-\vr\right).$ As $\lim_{n\to \infty}\norm{\left(\hc-\lambda \right)\phi_{n}}_{L^2} =0$ by definition of the spreading sequence, it follows that $\lim_{n\to \infty}\norm{\rho\left(\hc-\lambda \right)\phi_{n}}_{L^2} \leq \lim_{n\to \infty}\norm{\left(\hc-\lambda \right)\phi_{n}}_{L^2} =0.$ It therefore suffices to prove that (possibly up to extraction) $\lim_{n\to \infty}\norm{A_{\chi}\phi_{n}}_{L^2} = 0$.  By the Kato--Seiler--Simon inequality (\ref{KSS}) we obtain that
\[
\Bnorm{ (1-\Delta)^{-1}\rho \chi^2\left(\vl-\vr\right)}_{\mathfrak{S}_2}\leq \frac{1}{2\sqrt{\pi}}\bnorm{\rho \chi^2\left(\vl-\vr\right)}_{L^2(\bb R^3)}.
\]
In particular $\rho \chi^2\left(\vl-\vr\right)$ is $-\Delta$--compact, hence $\hr$--compact by the boundedness of $\vr$. As the sequence $\hr w_n$ is bounded in $L^2(\R^3)$, the $\hr$--compactness of $\rho \chi^2\left(\vl-\vr\right)$ implies that $\rho \chi^2\left(\vl-\vr\right)w_n $ converges strongly to some function $v\in L^2(\bb R^3)$. On the other hand, for any $f\in L^2(\bb R^3)$, 
$$\left(v,f\right)_{L^2} = \lim_{n\to+\infty}\left(\rho \chi^2\left(\vl-\vr\right)w_n ,f\right)_{L^2} = \lim_{n\to+\infty}\left(w_n ,\rho \chi^2\left(\vl-\vr\right)f\right)_{L^2} = 0, $$
where we have used the fact that $w_n\rightharpoonup 0$ and $\rho \chi^2\left(\vl-\vr\right) \in L^{\infty}(\bb R^3)$ so $\rho \chi^2\left(\vl-\vr\right)f \in L^2(\bb R^3) $. Therefore by the uniqueness of weak limit, $\rho \chi^2\left(\vl-\vr\right)w_n $ converges strongly to $ v\equiv 0$. Together with~\eqref{eq:approxrho} we conclude that $\lim_{n\to \infty}\norm{A\phi_{n}}_{L^2} = 0$. Therefore in view of~\eqref{decompose_op}, it holds $$\lim_{n\to \infty}\norm{\left(\hr -\lambda\right)w_{n}}_{L^2} =0.$$
Hence $\left\{w_{n}\right\}_{n\in \bb N}$ is a spreading sequence of $\hr$ associated with $\lambda<0$. This implies that for any spreading sequence of $\hc$ associated with $\lambda<0$, we can construct a spreading sequence for either $\hl$ or $\hr$ associated with $\lambda <0$, depending on whether (up to extraction) the spreading sequence of $\hc$ has non-trivial mass on $\left( (-\infty, 0]\times \bb R^2\right)$ or $\left((0,+\infty)\times \bb R^2\right)$. This allows us to conclude that
\begin{equation}
\sigma_{\mathrm{ess}}(\hc)\subseteq \sigma_{\mathrm{ess}}\left(\hl\right)\cup  \sigma_{\mathrm{ess}}\left(\hr\right).
\label{spHCHI1}
\end{equation}
By gathering~\eqref{spHCHI2} and~\eqref{spHCHI1} we conclude that $\sigma_{\mathrm{ess}}(\hc) \equiv \sigma_{\mathrm{ess}}\left(\hl\right)\cup  \sigma_{\mathrm{ess}}\left(\hr\right) $. In particular, $\sigma_{\mathrm{ess}}(\hc)$ is independent of the function $\chi\in\mathcal{X}$.
\subsection{Proof of Lemma \ref{lemma_approx}}
\label{lemma_approxSec}
The solution of the one-dimensional Yukawa equation~\eqref{eq:Yuka1} is
\begin{equation} 
 0<\eta(x)=\int_{\bb R}\frac{\rme^{-\sqrt{-\lambda}|x-y|}}{\sqrt{-\lambda}} \rme^{-2\sqrt{-\lambda}| y|}\,dy \leq \int_{\bb R}\frac{\rme^{-\sqrt{-\lambda}\left(|x|+|y|\right)}}{\sqrt{-\lambda}}\,dy\leq \frac{2}{-\lambda}\rme^{-\sqrt{-\lambda}| x|}.
 \label{etaprop}
 \end{equation}
This implies that $\eta $ belongs to $ L^p(\bb R)$ for $1\leq p\leq +\infty$. As $\eta >0$ is integrable, the cut-off function $\rho(x)=\frac{\int_{-\infty}^x\eta(y)\,dy}{\norm{\eta}_{L^1(\bb R)}} $ is well defined. It is easy to see that $0\leq \rho\leq 1$. Since $\norm{\phi_{n}}_{L^2\left((0,+\infty)\times \bb R^2\right)}\geq 1/2$ for $n$ large enough and $\rho(x)\geq \frac12$ for $x\geq 0$, hence $\lim_{n\to\infty}\norm{\rho \phi_{n}}_{L^2(\bb R^3)}\neq 0 $. 

Let us next prove that $\norm{\rho''\phi_{n} +2\rho'\partial_x\phi_{n}}_{L^2} \xrightarrow[n\to \infty]{} 0$. 
In view of~\eqref{etaprop}, 
\begin{align*}
\left|\eta'(x) \right| &= \left|\int_{-\infty}^{\infty}\rme^{-\sqrt{-\lambda}|x-y|-2\sqrt{-\lambda}| y|} \left(\mathds 1_{x\leq y} - \mathds 1_{x>y}\right)\,dy\right|\\
&\leq 2\int_{-\infty}^{+\infty}\rme^{-\sqrt{-\lambda}|x-y|-2\sqrt{-\lambda}| y|} \,dy = 2\sqrt{-\lambda}\eta(x) \leq \frac{4}{\sqrt{-\lambda}}\rme^{-\sqrt{-\lambda}| x|}.
\end{align*}
Combining the above inequality with~\eqref{etaprop} we obtain that
\begin{equation}
\label{rhoderivative_EQ}
\begin{aligned}
&\norm{\rho''\phi_{n} +2\rho'\partial_x\phi_{n}}_{L^2(\bb R^3)}^2  = \frac{\norm{\eta'\phi_{n} +2\eta\partial_x\phi_{n}}_{L^2(\bb R^3)}^2}{\norm{\eta}_{L^1(\bb R)}^2} \leq \frac{2}{\norm{\eta}_{L^1(\bb R)}^2}\int_{\bb R^3} \left(\left(\eta'\right)^2 |\phi_{n}|^2 + 4\eta^2 |\partial_x\phi_{n}|^2\right)\\
&\leq \frac{32}{|\lambda|\norm{\eta}_{L^1(\bb R)}^2}\int_{\bb R^3}\rme^{-2\sqrt{-\lambda}| \cdot|}|\phi_{n}|^2  + \frac{ 8\norm{\eta}_{L^{\infty}}}{\norm{\eta}_{L^1(\bb R)}^2}  \int_{\bb R^3}\eta \left|\partial_x\phi_{n}\right|^2.
\end{aligned}
\end{equation}
It suffices to prove that each of the previous integrals tends to $0$ when $n\to \infty$. Remark that the integrands of these terms appear when one does an explicit calculation of $\left(\eta \left(-\Delta -\lambda\right)\phi_n,  \left(-\Delta -\lambda\right)\phi_n\right)_{L^2(\bb R^3)}$. Let us therefore first prove that $$\left(\eta \left(-\Delta -\lambda\right)\phi_n,  \left(-\Delta -\lambda\right)\phi_n\right)_{L^2(\bb R^3)}\xrightarrow[n\to \infty]{} 0.$$ For this purpose, we prove that $\norm{\eta \left(-\Delta -\lambda\right)\phi_n}_{L^2(\bb R^3)} \xrightarrow[n\to \infty]{} 0$ strongly and $(-\Delta-\lambda)\phi_{n} \xrightharpoonup[n\to \infty]{} 0$ weakly in $L^2(\bb R^3)$. In view of~\eqref{etaprop}, for any $\varepsilon >0$ there exists $x_\varepsilon>0$ such that when $|x|>x_\varepsilon$, $0\leq \eta(x)\leq \varepsilon$. Together with~\eqref{maxVdecay} and the fact that $\left\{\phi_n\right\}_{n\in \bb N^*}$ is a spreading sequence, it holds
\[
\begin{aligned}
 \norm{\eta V_{\chi} \phi_n}_{L^2(\bb R^3)}^2 &=\int_{\bb R\times \left\{ |\bm r| > R_{\varepsilon}\right\}}\left|\eta V_{\chi} \phi_n\right|^2+ \int_{[-x_0,x_0]\times \left\{ |\bm r| \leq  R_{\varepsilon}\right\}} \left|\eta V_{\chi} \phi_n\right|^2+ \int_{[-x_0,x_0]^{\complement}\times \left\{ |\bm r| \leq R_{\varepsilon}\right\}} \left|\eta V_{\chi} \phi_n\right|^2\\
   &\leq \varepsilon \norm{\eta}_{L^{\infty}(\bb R)}^2 + \norm{\eta V_{\chi}}_{L^{\infty}(\bb R)}^2 \bnorm{ \phi_n\mathds 1_{[-x_0,x_0]\times\left\{ |\bm r| \leq  R_{\varepsilon}\right\}}}_{L^2(\bb R)}^2+\varepsilon\norm{V_\chi}_{L^{\infty}(\bb R^3)}^2 \xrightarrow[n\to \infty]{} 0.
\end{aligned}
\]
The above convergence allows us to conclude that 
\begin{equation}
 \norm{\eta \left(-\Delta -\lambda\right)\phi_n}_{L^2(\bb R^3)}\leq \norm{\eta  (\hc-\lambda) \phi_n}_{L^2(\bb R^3)} + \norm{\eta V_{\chi} \phi_n}_{L^2(\bb R^3)} \xrightarrow[n\to \infty]{} 0.
 \label{approximate1}
\end{equation}
Moreover, $\phi_n\in H^2(\bb R^3) \hookrightarrow L^{\infty}(\bb R^3)$ with continuous embedding for all $n\in \bb N^*$. Approximating $\phi_n$ by $C_c^{\infty}(\bb R^3)$ functions we deduce that $\lim_{|\bm x|\to \infty}|\phi_n(\bm x)| = 0$. Furthermore, remark that for any $\psi \in C_c^{\infty}(\bb R^3)$ it holds $$\left((-\Delta-\lambda)\phi_n,\psi\right)_{L^2(\bb R^3)}=\left(\phi_n,(-\Delta-\lambda)\psi\right)_{L^2(\bb R^3)}  \xrightarrow[n\to \infty]{} 0, $$
which implies the weak convergence $(-\Delta-\lambda)\phi_n \xrightharpoonup[n\to \infty]{}  0$ by the density of $C_c^{\infty}(\bb R^3)$ in $L^2(\bb R^3)$. Together with the strong convergence~\eqref{approximate1} and an integration by parts, this leads to,
\[
\begin{aligned}
& \left(\eta \left(-\Delta -\lambda\right)\phi_n,  \left(-\Delta -\lambda\right)\phi_n\right)_{L^2(\bb R^3)}= \int_{\bb R^3} \eta\left| \Delta \phi_n\right|^2+ \eta\lambda \overline{\Delta\phi_n} \phi_n + \lambda\eta \Delta\phi_n \overline{\phi_n} +  \eta\lambda^2|\phi_n|^2\\
 &\qquad\qquad= \int_{\bb R^3} \eta \left| \Delta \phi_n\right|^2- \lambda \eta'\partial_x\left(\left|\phi_n\right|^2\right)- 2\eta\lambda \left|\nabla\phi_n\right|^2  +  \lambda^2\eta|\phi_n|^2\\
 & \qquad\qquad=\int_{\bb R^3} \eta \left|\Delta \phi_n\right|^2+\left(\lambda \eta''+  \lambda^2\eta\right)|\phi_n|^2+2\eta|\lambda| \left|\nabla\phi_n\right|^2 \xrightarrow[n\to \infty]{} 0.
\end{aligned}
\]
In view of~\eqref{eq:Yuka1} we have $\left(\lambda \eta''+  \lambda^2\eta\right) =-\lambda \rme^{-2\sqrt{-\lambda}| \cdot|} > 0$, hence the integrand of the above integral is positive. This implies that 
\begin{equation}
\int_{\bb R^3}\rme^{-2\sqrt{-\lambda}| \cdot|}|\phi_n|^2 \xrightarrow[n\to \infty]{} 0,\qquad \int_{\bb R^3}\eta \left|\nabla\phi_n\right|^2 \xrightarrow[n\to \infty]{} 0.
\label{eq:cvgeta}
\end{equation}
 In view of~\eqref{rhoderivative_EQ} and~\eqref{eq:cvgeta}, we deduce that $ \norm{\rho''\phi_{n} +2\rho'\partial_x\phi_{n}}_{L^2(\bb R^3)}^2 \xrightarrow[n\to \infty]{} 0.$ Let us finally prove that $$ \rho\chi^2\left(\vl-\vr\right)\in L^2(\bb R^3).$$ 
By definition of $\chi$, the function $\rho\chi^2$ has support in $(-\infty, a_{R}/2]$ and is equal to $\rho(x)$ when $x\in (-\infty, -a_{L}/2)$. Remark also that~\eqref{etaprop} implies that 
\[\forall x<0, \qquad \rho(x) \leq \frac{1}{\norm{\eta}_{L^1}} \int_{-\infty}^x  \frac{2}{|\lambda|}\rme^{-\sqrt{-\lambda}| y|} \,dy \leq \frac{2}{|\lambda|^{3/2}\norm{\eta}_{L^1}} \rme^{\sqrt{-\lambda}x}.
\]  
Hence, as $\vl \in L_{\mathrm{per},x}^s\left(\Gamma_{L}\right) $ for $1< s \leq \infty $ by Theorem~\ref{thmPeriodicExistence2},
\begin{align*}
\int_{\bb R^3} (\rho\chi^2\vl)^2 &= \int_{(-\infty,-a_{L}/2)\times \bb R^2} (\rho\vl)^2 + \int_{[-a_{L}/2,a_{R}/2]\times \bb R^2}(\rho\chi^2\vl)^2 \\
&\leq \frac{4\norm{\vl}_{L^2(\Gamma_{a_{L}})}^2}{|\lambda|^{3}\norm{\eta}_{L^1(\bb R)}^2}\sum_{n=0}^{+\infty}\rme^{-2\sqrt{-\lambda} a_{L} n}+ \int_{[-a_{L}/2,a_{R}/2]\times \bb R^2}\vl^2 < +\infty.
\end{align*}
This implies that $\rho\chi^2\vl\in L^2(\bb R^3)$. Similar arguments show that $\rho\chi^2\vr\in L^2(\bb R^3)$, which concludes the proof of the lemma.
\subsection{Proof of Proposition \ref{propExist}} 
\label{propExistsec}
In view of Proposition \ref{spectrapPpHchi}, consider a contour $\mathfrak{C}$ in the complex plan enclosing the spectrum of $\hc$ below the Fermi level $\epsilon_F$ without intersecting it, crossing the real axis at $c<\inf\left\{-\norm{\vl}_{L^{\infty}},-\norm{\vr}_{L^{\infty}}\right\}$ (See Fig. \ref{spectrumFig}). This is possible even if $\epsilon_F$ is an eigenvalue: one can always slightly move the curve $\mathfrak{C}$ below $\epsilon_F$ in order bypass $\epsilon_F$ but still enclose all the spectrum of $\hc$ below $\epsilon_F$. 
\begin{figure} 
\begin{center}
\begin{tikzpicture}[scale=0.9]
\draw (-0.5, 0) -- (-0.5, -0.7) ;
\draw (-3.5, 0) -- (-3.5, -0.7) ;
\draw (-0.5, -0.7) node [below] {$\epsilon_F$};
\draw (-3.5, -0.7) node [below] {$c$};
\draw  [blue, very thick] (-2,0) ellipse (1.5 and 0.9);
\draw [->, ->= latex, very thick] (-4,0) -- (4,0);
\draw [green, very thick, fill = green,opacity=0.5](-3.0,-0.1) rectangle (-2.6, 0.1);
\draw [green, very thick, fill = green,opacity=0.5](-2.0,-0.1) rectangle (-1.5,0.1);
\draw (- 1.5 , 0) -- (-1.5, -0.3);
\draw (-1.18, -0.17) node [below] {$\Sigma_a$};
\draw (0.5 , 0) -- (0.5, -0.3);
\draw (0.5, -0.2) node [below right] {$\Sigma_b$};
\draw [green, very thick, fill = white,opacity=0.5](0.5,-0.1) rectangle (1.2,0.1);
\draw [green, very thick, fill = white,opacity=0.5](2.0,-0.1) rectangle (2.5,0.1);
\draw [green, very thick, fill = white](2.0, 0.6) rectangle (2.2,0.7);
\draw (2.2,0.6) node [right] {$-\sigma_{\mathrm{ess}}(\hc)=\sigma_{\mathrm{ess}}\left(H_{\mathrm{per,L}}\right)\cup \sigma_{\mathrm{ess}}\left(H_{\mathrm{per,R}}\right)$};
\draw[->,>=latex] (-0.7, 0.9) to[bend left] (-0.6, 0.27) ;
\draw (-0.7, 0.9) node [above] {$\mathfrak{C}$};
\end{tikzpicture}
\end{center}
\caption{The essential spectrum of $\hc$, $\hl$ and $\hr$, and the contour $\mathfrak{C}$.}
\label{spectrumFig}
\end{figure}
Let us introduce the following estimates, which are useful to characterize the decay property of  densities.
\begin{lemma}[Combes-Thomas estimate \cite{Klopp95anasymptotic, combes1973, GERMINET}]
Consider $H:= -\frac{1}{2}\Delta + V$ with $V\in L^{\infty}(\mathbb{R}^3)$. Let $p, q$ be positive integers such that $p q >3/2$. Then there exists $\varepsilon >0$ and a positive constant $C(p,q)$ such that for any $\zeta \notin \sigma(H)$, and any $(\alpha,\beta) \in \mathbb{Z}^3\times \mathbb{Z}^3$,
\begin{equation}
\norm[\bigg]{w_{\alpha} (\zeta-H)^{-p}w_{\beta}}_{\mathfrak{S}_q}\leq C(p,q)\left(1+ \frac{1}{\theta(\zeta,V)} \right)^{4p}\rme^{-\varepsilon \theta(\zeta,V)|\alpha-\beta|},
\label{CTSnorm}
\end{equation} where $\displaystyle \theta(\zeta,V)= \frac{ \mathrm{dist}(\zeta, \sigma(H))}{|\zeta|+\norm V_{L^{\infty}}+1}$.
\label{CombesThomas}
\end{lemma}
Since $V_\chi$ belongs to $L^{\infty}(\bb R^3)$, the following lemma is a direct adaption of \cite[Lemma 1]{Cances2008}:
\begin{lemma}
Under Assumption \ref{as:1}, there exist two positive constants $c_1, c_2$ such that 
$$\forall \zeta \in \mathfrak{C}, \quad c_1(1-\Delta) \leq |\hc-\zeta| \leq c_2(1-\Delta) $$ as operators on $L^2(\mathbb{R}^3)$. In particular $$\norm[\bigg]{|\hc-\zeta|^{1/2}(1-\Delta)^{-1/2}}\leq \sqrt{c_2}, ~~\norm[\bigg]{|\hc-\zeta|^{-1/2}(1-\Delta)^{1/2}}\leq \frac{1}{\sqrt{c_1}}.$$ Moreover, $(\hc-\zeta)(1-\Delta)^{-1}$ and its inverse are bounded operators.
\label{lemmaTech}
\end{lemma}
Let us turn to the proof of Proposition \ref{propExist}. First of all let us show that $\gamma_\chi$ is locally trace class. Consider $\varrho \in C_c^{\infty}(\bb R^3)$. Remark that $\gamma_\chi$ is a spectral projector. In view of Lemma~\ref{lemmaTech}, by Cauchy's resolvent formula and the Kato--Seiler--Simon inequality~\eqref{KSS}, there exists a positive constant $C_\chi$ such that 
\begin{align*}
 \norm{\varrho \gamma_\chi \varrho }_{\mathfrak{S}_1} = \norm{\varrho \gamma_\chi  \gamma_\chi\varrho }_{\mathfrak{S}_1} = \norm{\varrho \gamma_\chi}_{\mathfrak{S}_2}^2 = \bnorm{\varrho\oint_{\mathfrak{C}}\frac{1}{2\rmi \pi}\frac{1}{\zeta - \hc}\,d\zeta }_{\mathfrak{S}_2}^2\leq C_\chi\bnorm{\varrho \frac{1}{1-\Delta}}_{\mathfrak{S}_2}^2 \leq  \frac{ C_\chi}{4\pi}\norm{\varrho}_{L^2(\bb R^3)}^2.
\end{align*}
This implies that $\gamma_\chi$ is locally trace class so that its density $\rho_{\chi} $ is well defined in $ L_{\mathrm{loc}}^1(\bb R^3)$. Let us prove that~$\chi^2\rol+(1-\chi^2)\ror-\rho_{\chi} $ belongs to $  L^p(\mathbb{R}^3)$ for $1<p\leq 2$. It is difficult to directly compare the difference of $\chi^2\rol+(1-\chi^2)\ror$ and $\rho_{\chi}$. We construct to this end a density operator $\gamma_d$ whose density $\rho_{d}$ is equal to $\chi^2\rol+(1-\chi^2)\ror-\rho_{\chi}$:
\begin{equation}
\gamma_d := \gamma_{d,1} +\gamma_{d,2}, \quad \gamma_{d,1}:= \chi\left(\gl- \gamma_{\chi}\right)\chi ,\quad\gamma_{d,2} :=\sqrt{1-\chi^2}\left(\gamma_{\mathrm{per},R}-\gamma_{\chi} \right)\sqrt{1-\chi^2}.
\label{gammad12}
\end{equation}
Remark that if $\gamma_d \in \mathfrak{S}_1$, then $\mathrm{Tr}_{L^2(\bb R^3)}(\gamma_d) = \chi^2\rol+(1-\chi^2)\ror - \rho_{\chi}$.
\paragraph{The density $\rho_{d}$ is in $ L^p(\mathbb{R}^3)$ for $1<p\leq 2$. }
The proof that $\rho_{d}\in L^p(\mathbb{R}^3)$ relies on duality arguments: denoting by $q  = \frac{p}{p-1} \in [2, +\infty)$, we prove that for any $W \in  L^{q}(\mathbb{R}^3)$ there exists some $K_q>0$ such that $|\mathrm{Tr}_{L^2(\bb R^3)}(\gamma_d W)|  \leq  K_q\norm{W}_{L^{q}}$.
By Cauchy's formula we have
\begin{equation}
\begin{aligned}
\gamma_{d,1} &= \frac{1}{2\pi \rmi} \oint_{\mathfrak{C}}\chi \left(\frac{1}{z- \hl} -  \frac{1}{\zeta-\hc}\right)\chi \, d\zeta ,\\
\gamma_{d,2}&=\frac{1}{2\pi \rmi} \oint_{\mathfrak{C}} \sqrt{1-\chi^2}\left(\frac{1}{\zeta-\hr} -  \frac{1}{\zeta-\hc}\right)\sqrt{1-\chi^2} \, d\zeta.
\end{aligned}
\end{equation}
Let us prove that there exists $ K_q^1 > 0$ such that $$\left|\m{Tr}_{L^2(\bb R^3)}(\gamma_{d,1}W)\right|\leq K_q^1\norm{W}_{L^q}.$$ It is easily shown that a similar inequality holds for $\gamma_{d,2}$. 
Denote by $V_d:=(1-\chi^2)(\vl-\vr) \in L^{\infty}(\bb R^3)$. Remark that the function $V_d\chi $ has compact support in the $x$-direction and that $V_d\chi $ belongs to $L^r(\bb R^3)$ for $1<r\leq +\infty$ by Theorem~\ref{thmPeriodicExistence2}. For any $\zeta \in \mathfrak{C}$, the integrand of $\gamma_{d,1}$ writes 
\begin{align*}
D(\zeta):= \chi\left(\frac{1}{\zeta-\hl} -  \frac{1}{\zeta-\hc}\right)\chi  =\chi\frac{1}{\zeta-\hl}V_d\frac{1}{\zeta-\hc}\chi.
\end{align*}
Remark that $\chi$ being translation-invariant in the $\bm r$-direction, it is not in any $L^p$ space in $\bb R^3$, which prevents us from using the standard techniques such as calculating the commutator $[-\Delta, \chi]$ to give Schatten class estimates on $\gamma_{d,1}$. By writing $1 =\gl+ \gl^{\perp} $ and $1 =\gamma_{\chi}+ \gamma_{\chi}^{\perp} $, the following decomposition holds
\begin{equation}
\begin{aligned}
D(\zeta) = \chi \frac{\gl}{\zeta-\hl}V_d\frac{1}{\zeta-\hc}\chi + \chi \frac{\gl^{\perp}}{\zeta-\hl}V_d\frac{\gamma_{\chi}}{\zeta-\hc}\chi+\chi \frac{\gl^{\perp}}{\zeta-\hl}V_d\frac{\gamma_{\chi}^{\perp}}{\zeta-\hc}\chi.
\end{aligned}
\label{decompoD12}
\end{equation}
By the residue theorem,
\begin{equation}
\int_{\mathfrak{C}} \chi \frac{\gl^{\perp}}{\zeta-\hl}V_d\frac{\gamma_{\chi}^{\perp}}{\zeta-\hc}\chi\,d\zeta \equiv 0.
\label{D121EQ1}
\end{equation}
To estimate other terms in~\eqref{decompoD12} we rely on the following Lemmas~\ref{lemma45} and~\ref{lemmaD121}.
\begin{lemma}
Consider a self-adjoint operator $H = -\Delta + V$ defined on $L^2(\bb R^3)$ with domain $H^2(\bb R^2)$ and $V \in L^{\infty}(\bb R^3)$. For $E\in \bb R\backslash\sigma(H)$ denote by $\gamma = \mathds 1_{(-\infty,E]}(H)$. Then for any $a, b\in \bb R$, the operator $(1-\Delta)^a\gamma(1-\Delta)^{b}$ is bounded. Moreover, if $\gamma\in \mathfrak{S}_k$ for some $k\geq 1$, then $(1-\Delta)^a\gamma(1-\Delta)^{b}\in \mathfrak{S}_k$.
\label{lemma45}
\end{lemma}
\begin{proof}
Similarly as in Lemma~\ref{lemmaTech} it can be shown that for any $\zeta \in \mathbb{R}\backslash \sigma(H)$ the operator $(\zeta-H)^{-a}(1-\Delta)^a$ and its inverse are bounded. Fix $\delta>0$ and define $\lambda_0:= -\norm{V}_{L^{\infty}}-\delta $. Then $\lambda_0\notin \sigma(H)$. By writing $\gamma = \gamma^2$, there exists a positive constant $C$ such that 
\begin{align*}
&\norm[\bigg]{(1-\Delta)^a\gamma(1-\Delta)^{b} }_{\mathfrak{S}_k}\leq 
C \norm[\bigg]{(\lambda_0-H)^a\gamma (\lambda_0-H)^{b}}_{\mathfrak{S}_k} =C \norm[\bigg]{\gamma^2 (\lambda_0-H)^{a+b}}_{\mathfrak{S}_k}< +\infty,
\end{align*}
as $\gamma \in\mathfrak{S}_k$ and $\gamma (\lambda_0-H)^{a+b}$ is a bounded operator. The proof of the boundedness in operator norm follows the same lines.
\end{proof}
\begin{lemma}
For any $1<p\leq 2$, there exist positive constants $d_{p,1}$ and $d_{p,2}$, such that
\begin{equation}
\label{projectionnResolv}
\begin{aligned}
\forall \, \zeta \in \mathfrak{C},\quad \Bnorm{\chi \frac{ \gl}{\zeta-\hl}V_d}_{\mathfrak{S}_p}&\leq d_{p,1}\norm{V_d\chi}_{L^p(\bb R^3)},\quad \Bnorm{V_d\frac{ \gamma_{\chi}}{\zeta-\hc}\chi}_{\mathfrak{S}_p}&\leq d_{p,2}\norm{V_d\chi}_{L^p(\bb R^3)} .
\end{aligned}
\end{equation}
\label{lemmaD121}
\end{lemma}
\begin{proof}
Let us prove the statement for $\chi \gl (\zeta-\hl)^{-1}V_d$, the proof of the bound of $V_d\gamma_{\chi}(\zeta-\hc)^{-1}\chi $ follows similar arguments. Fix $R>0$. Recall that $\mathfrak{B}_R$ is the ball in $\bb R^3$ centered at $0$ with radius $R$. Denote by $\varphi_R $ the characteristic function  of $\mathfrak{B}_R$. For any $R>0$, by the Kato--Seiler--Simon inequality~\eqref{KSS} and the boundedness of $(1-\Delta)(\zeta-\hl)^{-1}$ it is easy to see that $\varphi_R \chi (\zeta-\hl)^{-1}$ and $(\zeta-\hl)^{-1}V_d\varphi_R$ belong to $\mathfrak{S}_2$. The operator $\gl(\zeta-\hl)^m $ is a bounded for any $m\in \bb R$ in view of Lemma~\ref{lemmaD121}. Therefore $$\varphi_R\left(\chi \frac{ \gl}{\zeta-\hl}V_d \right)\varphi_R=\left(\varphi_R\chi \frac{ 1}{\zeta-\hl}\right)\gl (\zeta-\hl) \left(\frac{ 1}{\zeta-\hl}V_d \varphi_R\right) \in \mathfrak{S}_1.$$ 
Let us first prove that for any $ 1\leq p\leq 2$, there exists a positive constant $d_{p,1}$ only depending on $p$ such that for any $R>0$,
\begin{equation}
\label{proof:regularizedInequality}
\Bnorm{\varphi_R\chi \frac{ \gl}{\zeta-\hl}V_d\varphi_R}_{\mathfrak{S}_p} \leq d_{p,1}\norm{V_d\varphi_R^2\chi}_{L^p(\bb R^3)}.
\end{equation}
We first prove~\eqref{proof:regularizedInequality} for $p=1$ and $p=2$, and conclude by an interpolation argument for $1\leq p\leq 2$. Consider $p=1$. By the cyclicity of the trace and the Kato--Seiler--Simon inequality~\eqref{KSS},
\begin{align*}
\bnorm{\varphi_R\chi \frac{ \gl}{\zeta-\hl}V_d\varphi_R }_{\mathfrak{S}_1} &=  \bnorm{\varphi_R\chi \frac{ 1 }{\zeta-\hl}\gl \left(\zeta - \hl\right) \frac{ 1 }{\zeta-\hl} V_d  \varphi_R}_{\mathfrak{S}_1}\\
& =  \bnorm{\gl \left(\zeta - \hl\right) \frac{ 1 }{\zeta-\hl}V_d\varphi_R^2\chi \frac{ 1 }{\zeta-\hl} }_{\mathfrak{S}_1}\\
&\leq  c\bnorm{\frac{ 1 }{\left|1-\Delta\right|} \left|V_d\varphi_R^2\chi\right|\frac{ 1 }{\left|1-\Delta\right|} }_{\mathfrak{S}_1} = c\bnorm{\frac{ 1 }{\left|1-\Delta\right|} \left|V_d\varphi_R^2\chi\right|^{1/2}}_{\mathfrak{S}_2}^2 \leq d_{1,1}\norm{V_d\varphi_R^2\chi}_{L^1}.
\end{align*}
Let us next prove~\eqref{proof:regularizedInequality} for $p =2$. Use again the cyclicity of the trace and the Kato--Seiler--Simon inequality~\eqref{KSS},
\begin{align*}
\bnorm{\varphi_R\chi \frac{ \gl}{\zeta-\hl}V_d\varphi_R }_{\mathfrak{S}_2}^2 &= \bnorm{V_d\varphi_R\frac{ \gl}{\overline{\zeta}-\hl}\varphi_R^2\chi^2 \frac{ \gl}{\zeta-\hl}V_d\varphi_R }_{\mathfrak{S}_1}\\
&= \bnorm{\frac{ \gl}{\overline{\zeta}-\hl}\varphi_R^2\chi^2 \frac{ \gl}{\zeta-\hl}V_d^2\varphi_R^2 }_{\mathfrak{S}_1}\\
&\leq c'\bnorm{\varphi_R^2\chi^2 \frac{ 1 }{\zeta-\hl}\gl \left(\zeta - \hl\right) \frac{ 1 }{\zeta-\hl} V_d^2\varphi_R^2 }_{\mathfrak{S}_1}\\
&= c'\bnorm{\gl \left(\zeta - \hl\right) \frac{ 1 }{\zeta-\hl} V_d^2\varphi_R^4\chi^2 \frac{ 1 }{\zeta-\hl} }_{\mathfrak{S}_1}\\
&\leq  c''\bnorm{\frac{ 1 }{\zeta-\hl} V_d^2\varphi_R^4\chi^2 \frac{ 1 }{\zeta-\hl} }_{\mathfrak{S}_1}\\
&=c''\bnorm{\left| V_d\varphi_R^2\chi\right|\frac{ 1 }{\zeta-\hl}}_{\mathfrak{S}_2}^2 \leq  d_{2,1}^2\norm{V_d\varphi_R^2\chi}_{L^2}^2.
\end{align*}
By the interpolation arguments we can conclude~\eqref{proof:regularizedInequality} for $1\leq p\leq 2$. Remark that for $1<p\leq 2$ the following uniform bound holds:
$$\Bnorm{\varphi_R\chi \frac{ \gl}{\zeta-\hl}V_d\varphi_R}_{\mathfrak{S}_p} \leq d_{p,1}\left\lVert V_d\varphi_R^2\chi\right\rVert_{L^p(\bb R^3)}\leq  d_{p,1}\norm{V_d\chi}_{L^p(\bb R^3)}.$$
By passing the limit $R \to +\infty$ we can conclude the proof.
\end{proof}
Consider $W\in L^q(\bb R^3)$ for $q = \frac{p}{p-1} \in [2,+\infty)$. In view of~\eqref{decompoD12},~\eqref{D121EQ1} and~\eqref{projectionnResolv}, by manipulations similar to the ones used in the proof of Lemma~\ref{lemmaD121}, and the H\"{o}lder's inequality for Schatten class operators (see for example~\cite[Proposition 5]{ReeSim2}), we obtain that 
\begin{equation}
\begin{aligned}
\norm{\gamma_{d,1}W}_{\mathfrak{S}_1} =\frac{1}{2\pi}\Bnorm{\oint_{\mathfrak{C}}D(\zeta) \,d\zeta W}_{\mathfrak{S}_1}&=\Bnorm{\oint_{\mathfrak{C}}\chi \frac{\gl}{\zeta-\hl}V_d\frac{1}{\zeta-\hc}\chi W+ \chi \frac{\gl^{\perp}}{\zeta-\hl}V_d\frac{\gamma_{\chi}}{\zeta-\hc}\chi W\,d\zeta}_{\mathfrak{S}_1}\\
&\leq \Bnorm{\oint_{\mathfrak{C}}\left(\chi \frac{\gl}{\zeta-\hl}V_d\right)\frac{1}{\zeta-\hc}(1-\Delta)\left(\frac{1}{1-\Delta}\chi W\right)\,d\zeta}_{\mathfrak{S}_1}\\
&\quad+ \Bnorm{\oint_{\mathfrak{C}}\left(W\chi \frac{1}{1-\Delta}\right)(1-\Delta)\frac{\gl^{\perp}}{\zeta-\hl}\left(V_d\frac{\gamma_{\chi}}{\zeta-\hc}\chi \right)\,d\zeta}_{\mathfrak{S}_1}\\
&\leq C\norm{V_d\chi}_{L^p(\bb R^3)}\Bnorm{ \frac{1}{1- \Delta}\chi W}_{\mathfrak{S}_q} \leq K_{q}^1\norm{W}_{L^q(\bb R^3)},
\end{aligned}
\label{D1V}
\end{equation}
where we have used the Kato--Seiler--Simon inequality~\eqref{KSS} as well as the fact that $\norm{\chi}_{L^{\infty}} = 1$. Similar estimates hold for $\gamma_{d,2}$. We therefore can conclude that $
\rho_d  = \chi^2\rol+(1-\chi^2)\ror - \rho_{\chi} $ belongs to $L^p(\bb R^3)$ for $1<p\leq 2$.
\paragraph{Decay rate in the $x$-direction.}
Let us show that the density difference $ \chi^2\rol+(1-\chi^2)\ror - \rho_{\chi} $ decays exponentially fast in the $x$-direction. Note that there exists $N_{L} \in \bb N$ such that $ N_{L} -1< a_{L}/2 \leq N_{L} $. Denote by $\bb D_{a_L} := [-a_{L}/2, +\infty)\times \bb R^2$, we prove the exponential decay when $\supp\left(w_{\alpha}\right) \subset \bb R^3\backslash \bb D_{a_L}$. Denote by $$\alpha:= (\alpha_x,0,0)\in (\bb R, 0, 0),\quad \beta = (\beta_x, \beta_y,\beta_z)\in \bb Z^3, \quad \beta_x \geq -N_L.$$
We have $$
 \mathds 1_{\bb D_{a_L}} \left(\sum_{\beta_x \geq -N_{L}}^{+\infty}\sum_{\beta_y, \beta_z\in \bb Z}w_{\beta}  \right)= \mathds 1_{\bb D_{a_L}}, \quad \mathds 1_{\bb D_{a_L} }V_d = V_d, \quad \alpha_x < -\frac{a_{L}}{2}<-N_{L} + 1\leq \beta_x+1.
$$
The above relations imply, together with~\eqref{projectionnResolv}, the Combes--Thomas estimate (\ref{CTSnorm}) and arguments similar to ones used in~\eqref{D1V}, that there exist a positive constants $C_1$ and $t_1$ such that, for $1<p\leq 2$ and $ q = \frac{p}{p-1}\geq 2 $,
 \begin{align*}
\norm{w_{\alpha}\gamma_dw_{\alpha}}_{\mathfrak{S}_1}  &=\norm{w_{\alpha}\gamma_{d,1}w_{\alpha}}_{\mathfrak{S}_1}\\
&=\Bnorm{\frac{1}{2\pi \rmi} \oint_{\mathfrak{C}}\left(w_{\alpha}\chi \frac{\gl}{\zeta-\hl}V_d\frac{1}{\zeta-\hc}\chi  \,w_{\alpha}+ w_{\alpha}\chi \frac{\gl^{\perp}}{\zeta-\hl}V_d\frac{\gamma_{\chi}}{\zeta-\hc}\chi \,w_{\alpha}\right) d\zeta }_{\mathfrak{S}_1}\\
&\leq \Bnorm{\frac{1}{2\pi \rmi} \oint_{\mathfrak{C}}\left(w_{\alpha}\chi \frac{\gl}{\zeta-\hl}V_d\right)\left(\mathds 1_{\bb D_{a_L} }\frac{1}{\zeta-\hc}\chi  \,w_{\alpha}\right) d\zeta }_{\mathfrak{S}_1}\\
&\quad+ \Bnorm{\frac{1}{2\pi \rmi} \oint_{\mathfrak{C}}\left(w_{\alpha}\chi \frac{\gl^{\perp}}{\zeta-\hl}\mathds 1_{\bb D_{a_L} }\right)\left(V_d\frac{\gamma_{\chi}}{\zeta-\hc}\chi \,w_{\alpha}\right) d\zeta }_{\mathfrak{S}_1}\\
& \leq K\sum_{\beta_x \geq- N_{L}}^{+\infty}\sum_{\beta_y, \beta_z\in \bb Z}^{+\infty}  \rme^{- t_1(\beta_x-\alpha_x)} \rme^{-t_1|\beta_y|}\rme^{-t_1|\beta_z|} \leq C_1 \rme^{-t_1 |\alpha_x|}.
\end{align*}
The last step relies on the uniform distance of $\zeta \in \mathfrak{C}$ to $\sigma(\hc)$ and $\sigma( \hl)$. Similar estimates hold when the support of $w_{\alpha}$ is in $[a_{R}/2,+\infty)\times \bb R^2$. There exist therefore positive constants $C$ and $t$ such that $\norm{w_{\alpha}\gamma_dw_{\alpha}}_{\mathfrak{S}_1} = \int_{\bb R^3} \left|w_{\alpha }\rho_dw_{\alpha}\right|\leq C\rme^{-t|\alpha|}$, which concludes the proof.

\subsection{Proof of Lemma~\ref{eta_chi_Lemma}}
\label{eta_chi_LemmaSec}
From the last item of the Theorem~\ref{thmPeriodicExistence2} we know that $\vl\in L_{\mathrm{per},x}^p(\Gamma_L) $ (resp.  $\vr\in L_{\mathrm{per},x}^p(\Gamma_R) $) for $1< p \leq +\infty$. Remark also that $\partial_x^2(\chi^2), \partial_x(\chi^2)$ are uniformly bounded and have support in $\left[-a_{L}/2, a_{R}/2\right]\times \mathbb{R}^2$. It therefore suffices to obtain the $L^p$-estimates on $\partial_x \vl$ and $\partial_x \vr$. We treat $\partial_x\vl$, the $L^p$-estimates of $\partial_x \vr$ following similar arguments. First of all in view of the form of the minimizer~\eqref{FormMinimzer}, by the Cauchy--Schwarz inequality
\begin{align*}
\partial_x \rol &= \partial_x \left(\frac{1}{2\pi} \int_{\Gamma_L^*}\sum_{n\geq 1}\mathds 1(\lambda_n(\xi)\leq \epsilon_L)\left|e_n(\xi, \cdot) \right|^2\,d\xi\right)\\
&\leq  \frac{1}{\pi} \int_{\Gamma_L^*}\left(\sum_{n\geq 1}\mathds 1(\lambda_n(\xi)\leq \epsilon_L)\left|\partial_x |e_n|(\xi, \cdot) \right|^2\right)^{1/2}\left(\sum_{n\geq 1}\mathds 1(\lambda_n(\xi)\leq \epsilon_L)\left| e_n(\xi, \cdot) \right|^2\right)^{1/2}\,d\xi\\
&\leq \frac{1}{\pi} \sqrt{K_{\xi,L}}\sqrt{\rol},
\end{align*}
where $K_{\xi,L}(\bm x):= \int_{\Gamma_L^*}\sum_{n\geq 1}\mathds 1(\lambda_n(\xi)\leq \epsilon_L)\left|\partial_x  e_n(\xi, \bm x) \,d\xi\right|^2$. We also have used the fact that $|\nabla |f  | | \leq  |\nabla f | $ for any complex-valued function $f$. In view of the potential decomposition~\eqref{VpDecomposed}, the term $T(\bm r)$ does not contribute to the $x$-directional derivative, hence
\begin{align*}
&\left|\partial_x \vl \right|=\left| \left(\partial_x \left(\rol-\mu_{\mathrm{per},L}\right)\right) \star \widetilde{G_{a_L}} \right| \leq  \left(\frac{1}{2\pi} \sqrt{K_{\xi,L}}\sqrt{\rol}+\left|\partial_x  \mu_{\mathrm{per},L}\right|\right) \star \left| \widetilde{G_{a_L}}\right|.
\end{align*}
On the other hand, finite kinetic energy condition~\eqref{kineticEnergyset} implies that $ K_{\xi,L} \in L_{\mathrm{per},x}^1(\Gamma_L)$. Moreover, $\sqrt{\rol}$ belongs to $ H_{\mathrm{per},x}^1(\Gamma_L)$ hence to $L_{\mathrm{per},x}^s(\Gamma_L)$ for $2\leq s\leq 6$. Therefore, by H\"{o}lder's inequality, for $p, m \geq 1$:
$$\int_{\Gamma_L}\left(K_{\xi,L}\rol\right)^{p/2}\leq \left(\int_{\Gamma_L}K_{\xi,L}^{pm/2}\right)^{1/m}\left(\int_{\Gamma_L}\rol^{pm/\left(2(m-1)\right)}\right)^{(m-1)/m}, $$
with the conditions $pm =2$ and $ 2 \leq pm/(m-1) \leq 6$. This is the case for $ 4/3\leq  m \leq 2 $ and $1 \leq  p  \leq 3/2  $ so that~$\left(K_{\xi,L}\rol\right)^{1/2}$ belongs to $L_{\mathrm{per},x}^{p}(\Gamma_L)$ for $1\leq p\leq 3/2$. As~$\partial_x\mu_{\mathrm{per},L}$ is in $L_{\mathrm{per},x}^{p}(\Gamma_L)$ for any $1\leq p\leq +\infty$ and~$\widetilde{G_{a_L}}\in L_{\mathrm{per},x}^q(\Gamma_L)$ for $1\leq q <2 $ by Lemma~\ref{lemma1}, we obtain by Young's convolution inequality that~$ \partial_x \vl \in L_{\mathrm{per},x}^{s}(\Gamma_L)$ for $  1\leq s < 6 $. This allows us to conclude the lemma.
\subsection{Proof of Proposition \ref{PropExistMinimizer}}
\label{SecPropExistMinimizer}
The following statements assure that the problem~\eqref{minimizationQ} is well-defined: a duality argument shows that densities of operators in $\mathcal{Q}_{\chi}$ are well defined. The energy functional (\ref{energyFunctional}) is well defined and is bounded from below. The proofs are direct adaptations of \cite[Proposition 1, Lemma 2, Corollary 1 and Corollary 2]{Cances2008}. 
\begin{enumerate}
\item For any $Q_{\chi} \in \mathcal{Q}_{\chi}$, it holds that $Q_{\chi}W\in \mathfrak{S}_1^{\gamma_{\chi}}$ for $W = W_1 + W_2 \in \mathcal{C}' + L^2(\mathbb{R}^3)$. Moreover, there exists a positive constant $C_\chi $ such that:
$$|\mathrm{Tr}_{\gamma_{\chi}}\left(Q_{\chi}W\right)| \leq C_\chi \left\lVert Q_{\chi}\right\rVert_{\mathcal{Q}_{\chi}}(\lVert W_1\rVert_{\mathcal{C}'}+ \lVert W_2\rVert_{L^2(\mathbb{R}^3)}).$$ 
Moreover, there exists a uniquely defined function $\rho_{Q_{\chi}} \in \mathcal{C} \bigcap L^2(\mathbb{R}^3)$ such that 
$$\forall \, W = W_1+W_2 \in \mathcal{C}' + L^2(\mathbb{R}^3), \quad \mathrm{Tr}_{\gamma_{\chi}} \left(Q_{\chi}W\right) =\langle \rho_{Q_{\chi}}, W_1\rangle_{\mathcal{C}',\mathcal{C}} + \int_{\mathbb{R}^3}\rho_{Q_{\chi}}W_2 .$$
The linear map $Q_{\chi}\in \mathcal{Q}_{\chi} \mapsto \rho_{Q_{\chi}} \in \mathcal{C} \bigcap L^2(\mathbb{R}^3)$ is continuous:
$$\lVert \rho_{Q_{\chi}}\rVert_{\mathcal{C}} +\lVert \rho_{Q_{\chi}}\rVert_{L^2(\mathbb{R}^2)} \leq C_\chi \lVert Q_{\chi}\rVert_{\mathcal{Q}_{\chi}}.$$ 
Moreover, if $Q_{\chi}\in \mathfrak{S}_1 \subset \mathfrak{S}_1^{\gamma_{\chi}}$, then $\rho_{Q_{\chi}}(\bm x)=Q_{\chi}(\bm x,\bm x)$ where $Q_{\chi}(\bm x,\bm x)$ the integral kernel of $Q_{\chi}$.
\item For any $\kappa\in (\Sigma_{N,\chi}, \epsilon_F)$ and any state $Q_{\chi}\in \mathcal{K}_{\chi}$, the following inequality holds \begin{align*}
0\leq c_1 &\mathrm{Tr}\left((1-\Delta)^{1/2}\left(Q_{\chi}^{++}-Q_{\chi}^{--}\right)(1-\Delta)^{1/2}\right)\leq\mathrm{Tr}_{\gamma_{\chi}}\left(\hc Q_{\chi}\right)-\kappa\mathrm{Tr}_{\gamma_{\chi}}(Q_{\chi}) \\
&\qquad\leq  c_2\mathrm{Tr}\left((1-\Delta)^{1/2}(Q_{\chi}^{++}-Q_{\chi}^{--})(1-\Delta)^{1/2}\right),
\end{align*}
where $c_1$ and $c_2$ are the same constants as in Lemma~\ref{lemmaTech}.
\item Assume that Assumption~\ref{as:1} holds. There are positive constants $\widetilde{d_1}$, $\widetilde{d_2}$, such that
\begin{align*}
\mathcal{E}_{\chi}(Q_{\chi})-\kappa\mathrm{Tr}_{\gamma_{\chi}}(Q_{\chi})&\geq \widetilde{d_1}\left(\lVert Q_{\chi}^{++}\rVert_{{\mathfrak{S}}_1} +\lVert Q_{\chi}^{--}\rVert_{{\mathfrak{S}}_1} + \lVert | \nabla |Q_{\chi}^{++}| \nabla |\rVert_{{\mathfrak{S}}_1} +\lVert | \nabla |Q_{\chi}^{--}| \nabla |\rVert_{{\mathfrak{S}}_1}\right)\\
&+\widetilde{d_2}\left(\lVert| \nabla |Q_{\chi}\rVert_{{\mathfrak{S}}_2}^2+\lVert Q_{\chi}\rVert_{{\mathfrak{S}}_2}^2\right)  -\frac{1}{2}D(\nu_{\chi}, \nu_{\chi}).
\end{align*}
Hence $\mathcal{E}_{\chi}(\cdot)-\kappa\mathrm{Tr}_{\gamma_{\chi}}(\cdot)$ is bounded from below and coercive on $\mathcal{K}_{\chi}$.
\end{enumerate}
The existence and the form of the minimizers are direct adaptations of~\cite[Theorem 2]{Cances2008}. 
\subsection{Proof of Theorem~\ref{thm2}}
\label{thm2Sec}
We prove this theorem by taking two arbitrary cut-off functions $\chi_1, \chi_2$ belonging to $\mathcal{X}$, and prove that $\rho_{\chi_1} + \rho_{Q_{\chi_1}}= \rho_{\chi_2} + \rho_{Q_{\chi_2}}$. For $i = 1,2$, consider the reference states associated with the Hamiltonian $H_{\chi_i}$. Denote by $\gamma_{\chi_i}$ the spectral projector of $H_{\chi_i}$ below $\epsilon_F$ and by $Q_{\chi_i}$ the solutions of (\ref{selfEquationQbar}) associated with $\chi_i$. Consider a test state
\begin{equation}
 \widetilde{Q} := \gamma_{\chi_1} +Q_{\chi_1}- \gamma_{\chi_2}.
 \label{testState}
\end{equation}
We show that $ \widetilde{Q}$ is a minimizer of the problem~\eqref{minimizationQ} associated with the cut-off function $\chi_2$, so that $\rho_{\widetilde{Q}} \equiv \rho_{Q_{\chi_2}}$ by the uniqueness of the density of the minimizer provided by Proposition~\ref{PropExistMinimizer}. Note that Assumption~\ref{as:1} and Proposition~\ref{spectrapPpHchi} guarantee that there is a common spectral gap for $H_{\chi_i}$ and $\sigma_{\mathrm{ess}}(H_{\chi_1})=\sigma_{\mathrm{ess}}(H_{\chi_2})$. We first show that the test state $\widetilde{Q}$ belongs to the convex set $$\mathcal{K}_{\chi_2}:=\left\{Q\in \mathcal{Q}_{\chi_2} \mid - \gamma_{\chi_2}\leq Q \leq 1-\gamma_{\chi_2}\right\},$$
hence is an admissible state for the minimization problem~\eqref{minimizationQ} associated with $\chi_2$. We next show that $\widetilde{Q}$ is a minimizer. 
\paragraph{The test state $\widetilde{Q} $ belongs to $ \mathcal{K}_{\chi_2}$.}
We begin by proving that $\widetilde{Q}$ is in $\mathcal{Q}_{\chi_2}$. Let us prove that $\widetilde{Q}$ is $\gamma_{\chi_2}$-trace class. The following lemma will be useful.
\begin{lemma}
\label{Lemma_diffProj}
The difference of the spectral projectors $\gamma_{\chi_1}-\gamma_{\chi_2}$ belongs to $\mathfrak{S}_1^{\gamma_{\chi_2}}$. Moreover,
\begin{equation}
|\nabla|\left(\gamma_{\chi_1} - \gamma_{\chi_2}\right) \in \mathfrak{S}_2,\quad \left(\gamma_{\chi_1} - \gamma_{\chi_2}\right) |\nabla| \in \mathfrak{S}_2.
\end{equation}
\end{lemma}
\begin{proof}
By Cauchy's resolvent formula and the Kato--Seiler--Simon inequality~\eqref{KSS},
\begin{equation}
\begin{aligned}
\norm{\gamma_{\chi_1} - \gamma_{\chi_2}}_{\mathfrak{S}_2} &= \Bnorm{\frac{1}{2 
\rmi \pi } \oint_{\mathfrak{C}} (\zeta-H_{\chi_1})^{-1}(\chi_1^2-\chi_2^2)(\vl-\vr)(\zeta-H_{\chi_2})^{-1} \, d\zeta }_{\mathfrak{S}_2}\\
&\leq C\norm[\bigg]{(1-\Delta)^{-1}(\chi_1^2-\chi_2^2)(\vl-\vr) }_{\mathfrak{S}_2} \leq \frac{C}{2\sqrt{\pi}}\norm[\bigg]{(\chi_1^2-\chi_2^2)(\vl-\vr) }_{L^2} < +\infty.
\end{aligned}
\label{diffGamma12}
\end{equation}
The results of Lemma~\ref{lemmaTech} imply that $|\nabla|(\zeta-H_{\chi_i})^{-1}$ is uniformly bounded with respect to $\zeta\in\mathfrak{C}$. By calculations similar to~\eqref{diffGamma12},
\begin{equation}
\bnorm{|\nabla|\left(\gamma_{\chi_1} - \gamma_{\chi_2}\right)}_{\mathfrak{S}_2} \leq c_1\bnorm{(\chi_1^2-\chi_2^2)(\vl-\vr) (1-\Delta)^{-1}}_{\mathfrak{S}_2}<+\infty.
\label{GraddiffGamma}
\end{equation}
Hence $ \left(\gamma_{\chi_1} - \gamma_{\chi_2}\right) |\nabla| $ also belongs to $\mathfrak{S}_2$ since it is the adjoint of $|\nabla|\left(\gamma_{\chi_1} - \gamma_{\chi_2}\right)$. On the other hand, as $\gamma_{\chi_1}$ is a bounded operator, in view of~\eqref{diffGamma12} and by writing $ \gamma_{\chi_1} -  \gamma_{\chi_2}=\gamma_{\chi_2}^{\perp} -\gamma_{\chi_1}^{\perp} $ and using the fact that $\gamma_{\chi_i} +\gamma_{\chi_i}^{\perp} = 1 $,
\begin{equation}
\begin{aligned}
&\gamma_{\chi_2}^{\perp}\left(\gamma_{\chi_1} - \gamma_{\chi_2}\right) \gamma_{\chi_2}^{\perp} = \gamma_{\chi_2}^{\perp}\gamma_{\chi_1}\gamma_{\chi_2}^{\perp} = \left(\gamma_{\chi_1}- \gamma_{\chi_2}\right)\gamma_{\chi_1}\left(\gamma_{\chi_1}- \gamma_{\chi_2}\right) \in \mathfrak{S}_1,\\
 & \gamma_{\chi_2}\left(\gamma_{\chi_1}- \gamma_{\chi_2}\right)\gamma_{\chi_2} = - \gamma_{\chi_2}\gamma_{\chi_1}^{\perp}\gamma_{\chi_2} =-\left(\gamma_{\chi_2}- \gamma_{\chi_1}\right)\gamma_{\chi_1}^{\perp}\left(\gamma_{\chi_2}- \gamma_{\chi_1}\right) \in \mathfrak{S}_1.
\end{aligned}
\label{estimate_gamma12}
\end{equation}
Together with~\eqref{diffGamma12} we conclude that $\gamma_{\chi_1}- \gamma_{\chi_2}$ belongs to $\mathfrak{S}_1^{\gamma_{\chi_2}}$.
\end{proof}
The following lemma is a consequence of~\cite[Lemma 1]{Hainzl2005} and the fact that $\gamma_{\chi_1}-\gamma_{\chi_2}\in\mathfrak{S}_2$.
\begin{lemma}
\label{TracediffLemma}
Any self-adjoint operator $A$ is in $ \mathfrak{S}_1^{\gamma_{\chi_1}}$ if and only if $A$ is in $\mathfrak{S}_1^{\gamma_{\chi_2}}$. Moreover $\m{Tr}_{\gamma_{\chi_1}}(A) =\m{Tr}_{\gamma_{\chi_2}}(A) $.
\end{lemma}
The fact that $Q_{\chi_1}\in \mathfrak{S}_1^{\gamma_{\chi_1}}$ implies that $|\nabla|Q_{\chi_1} \in \mathfrak{S}_2$, and $Q_{\chi_1} \in \mathfrak{S}_1^{\gamma_{\chi_2}}$ by Lemma~\ref{TracediffLemma}. In view of this and~Lemma~\ref{Lemma_diffProj} we know that $\widetilde{Q} = \gamma_{\chi_1}-\gamma_{\chi_2} +Q_{\chi_1}$ belongs to $\mathfrak{S}_1^{\gamma_{\chi_2}}$. The inequality~\eqref{GraddiffGamma} implies that $|\nabla| \widetilde{Q}= |\nabla|Q_{\chi_1} +  |\nabla| (\gamma_{\chi_1}-\gamma_{\chi_2}) \in \mathfrak{S}_2$. It remains to prove that $|\nabla|\gamma_{\chi_2}^{\perp}\widetilde{Q}\gamma_{\chi_2}^{\perp}|\nabla| \in \mathfrak{S}_1$ and $|\nabla| \gamma_{\chi_2}\widetilde{Q}\gamma_{\chi_2}|\nabla| \in \mathfrak{S}_1$. In view of~\eqref{testState} we have
\begin{equation}
\begin{aligned}
|\nabla|\gamma_{\chi_2}^{\perp}\widetilde{Q}\gamma_{\chi_2}^{\perp}|\nabla| &=|\nabla|\gamma_{\chi_2}^{\perp}Q_{\chi_1}\gamma_{\chi_2}^{\perp}|\nabla| + |\nabla|\gamma_{\chi_2}^{\perp}\left(\gamma_{\chi_1}- \gamma_{\chi_2}\right)\gamma_{\chi_2}^{\perp}|\nabla|,\\
|\nabla|\gamma_{\chi_2}\widetilde{Q}\gamma_{\chi_2}|\nabla| &=|\nabla|\gamma_{\chi_2}Q_{\chi_1}\gamma_{\chi_2}|\nabla| + |\nabla|\gamma_{\chi_2}\left(\gamma_{\chi_1}- \gamma_{\chi_2}\right)\gamma_{\chi_2}|\nabla|.
\end{aligned}
\label{chi2Qchi2}
\end{equation}
We estimate~\eqref{chi2Qchi2} term by term. By Lemma \ref{lemma45} we know that $\norm{|\nabla|\gamma_{\chi_i}}  \leq \norm{|\nabla|(1-\Delta)^{-1}}\norm{(1-\Delta)\gamma_{\chi_i}} <\infty$. Moreover, by writing $\gamma_{\chi_2}^{\perp} = 1-\gamma_{\chi_1} +  \gamma_{\chi_1} -  \gamma_{\chi_2}$ and $\gamma_{\chi_2} = \gamma_{\chi_1} +  \gamma_{\chi_2}- \gamma_{\chi_1} $, in view of~\eqref{GraddiffGamma},
\[
\begin{aligned}
&\bnorm{ |\nabla|\gamma_{\chi_2}^{\perp}Q_{\chi_1}\gamma_{\chi_2}^{\perp}|\nabla|}_{\mathfrak{S}_1}=\bnorm{ |\nabla|\gamma_{\chi_1}^{\perp}Q_{\chi_1}\gamma_{\chi_1}^{\perp}|\nabla|  + |\nabla| \gamma_{\chi_1}^{\perp}Q_{\chi_1}(\gamma_{\chi_1}-\gamma_{\chi_2})|\nabla| +|\nabla| (\gamma_{\chi_1}-\gamma_{\chi_2})Q_{\chi_1}\gamma_{\chi_2}^{\perp}|\nabla|  }_{\mathfrak{S}_1}\\
&\leq  \bnorm{ |\nabla|\gamma_{\chi_1}^{\perp}Q_{\chi_1}\gamma_{\chi_1}^{\perp}|\nabla|}_{\mathfrak{S}_1}  + \bnorm{|\nabla| Q_{\chi_1}(\gamma_{\chi_1}-\gamma_{\chi_2})|\nabla|}_{\mathfrak{S}_1}+\bnorm{|\nabla| \gamma_{\chi_1}Q_{\chi_1}(\gamma_{\chi_1}-\gamma_{\chi_2})|\nabla|}_{\mathfrak{S}_1}\\
&\qquad + \bnorm{|\nabla| (\gamma_{\chi_1}-\gamma_{\chi_2})Q_{\chi_1}\gamma_{\chi_2}|\nabla|  }_{\mathfrak{S}_1} +  \bnorm{|\nabla| (\gamma_{\chi_1}-\gamma_{\chi_2})Q_{\chi_1}|\nabla|  }_{\mathfrak{S}_1} \\
&\leq  \bnorm{ |\nabla|\gamma_{\chi_1}^{\perp}Q_{\chi_1}\gamma_{\chi_1}^{\perp}|\nabla|}_{\mathfrak{S}_1}  + \bnorm{|\nabla| Q_{\chi_1}}_{\mathfrak{S}_2}\bnorm{(\gamma_{\chi_1}-\gamma_{\chi_2})|\nabla|}_{\mathfrak{S}_2}+\bnorm{|\nabla| \gamma_{\chi_1}}\bnorm{Q_{\chi_1}}_{\mathfrak{S}_2}\bnorm{(\gamma_{\chi_1}-\gamma_{\chi_2})|\nabla|}_{\mathfrak{S}_2}\\
& \qquad + \bnorm{\gamma_{\chi_2}|\nabla| }\bnorm{Q_{\chi_1}}_{\mathfrak{S}_2}\bnorm{|\nabla|(\gamma_{\chi_1}-\gamma_{\chi_2})}_{\mathfrak{S}_2}+ \bnorm{Q_{\chi_1}|\nabla|}_{\mathfrak{S}_2}\bnorm{|\nabla|(\gamma_{\chi_1}-\gamma_{\chi_2})}_{\mathfrak{S}_2} <\infty,
\end{aligned}
\]
and similarly
\[
\begin{aligned}
&\bnorm{ |\nabla|\gamma_{\chi_2}Q_{\chi_1}\gamma_{\chi_2}|\nabla|}_{\mathfrak{S}_1}=\bnorm{ |\nabla|\gamma_{\chi_1}Q_{\chi_1}\gamma_{\chi_1}|\nabla|  + |\nabla| \gamma_{\chi_1}Q_{\chi_1}(\gamma_{\chi_2}-\gamma_{\chi_1})|\nabla| +|\nabla| (\gamma_{\chi_2}-\gamma_{\chi_1})Q_{\chi_1}\gamma_{\chi_2}|\nabla|  }_{\mathfrak{S}_1}\\
&\leq  \bnorm{ |\nabla|\gamma_{\chi_1}Q_{\chi_1}\gamma_{\chi_1}|\nabla|}_{\mathfrak{S}_1}  +\bnorm{|\nabla| \gamma_{\chi_1}Q_{\chi_1}(\gamma_{\chi_1}-\gamma_{\chi_2})|\nabla|}_{\mathfrak{S}_1}+ \bnorm{|\nabla| (\gamma_{\chi_1}-\gamma_{\chi_2})Q_{\chi_1}\gamma_{\chi_2}|\nabla|  }_{\mathfrak{S}_1} \\
&\leq  \bnorm{ |\nabla|\gamma_{\chi_1}Q_{\chi_1}\gamma_{\chi_1}|\nabla|}_{\mathfrak{S}_1}  +\bnorm{|\nabla| \gamma_{\chi_1}}\bnorm{Q_{\chi_1}}_{\mathfrak{S}_2}\bnorm{(\gamma_{\chi_1}-\gamma_{\chi_2})|\nabla|}_{\mathfrak{S}_2}+\bnorm{\gamma_{\chi_2}|\nabla| }\bnorm{Q_{\chi_1}}_{\mathfrak{S}_2}\bnorm{|\nabla|(\gamma_{\chi_1}-\gamma_{\chi_2})}_{\mathfrak{S}_2} <\infty,
\end{aligned}
\]
From~\eqref{estimate_gamma12} we know that 
\[
\begin{aligned}
&\bnorm{ |\nabla|\gamma_{\chi_2}^{\perp}\left(\gamma_{\chi_1}- \gamma_{\chi_2}\right)\gamma_{\chi_2}^{\perp}|\nabla|}_{\mathfrak{S}_1}= \bnorm{ |\nabla|\left(\gamma_{\chi_1}- \gamma_{\chi_2}\right)\gamma_{\chi_1}\left(\gamma_{\chi_1}- \gamma_{\chi_2}\right)|\nabla|}_{\mathfrak{S}_1} \leq \bnorm{|\nabla|(\gamma_{\chi_1}-\gamma_{\chi_2})}_{\mathfrak{S}_2}^2 <\infty,\\
&\bnorm{ |\nabla|\gamma_{\chi_2}\left(\gamma_{\chi_1}- \gamma_{\chi_2}\right)\gamma_{\chi_2}|\nabla|}_{\mathfrak{S}_1}= \bnorm{ |\nabla|\left(\gamma_{\chi_2}- \gamma_{\chi_1}\right)\gamma_{\chi_1}^{\perp}\left(\gamma_{\chi_2}- \gamma_{\chi_1}\right)|\nabla|}_{\mathfrak{S}_1} \leq \bnorm{|\nabla|(\gamma_{\chi_1}-\gamma_{\chi_2})}_{\mathfrak{S}_2}^2 <\infty.
\end{aligned}
\]
This shows that $|\nabla| \gamma_{\chi_2}^{\perp}\widetilde{Q}\gamma_{\chi_2}^{\perp}|\nabla| \in \mathfrak{S}_1$ and $|\nabla|\gamma_{\chi_2}\widetilde{Q}\gamma_{\chi_2}|\nabla| \in \mathfrak{S}_1$. In view of~\eqref{chi2Qchi2}, this allows us to conclude that $\widetilde{Q}\in\mathcal{Q}_{\chi_2}$. On the other hand, it is easy to see that $ -\gamma_{\chi_2}\leq \widetilde{Q} = \gamma_{\chi_1} +Q_{\chi_1}- \gamma_{\chi_2} \leq 1-\gamma_{\chi_2}$, which shows that $\widetilde{Q}$ belongs to the convex set $\mathcal{K}_{\chi_2}$.
\paragraph{The state $\widetilde{Q}$ is a minimizer.}
We now prove that $\widetilde{Q}$ is a minimizer of the problem~\eqref{minimizationQ} associated with the cut-off function $\chi_2$. As $\widetilde{Q}\in \mathcal{K}_{\chi_2}$, the fact that $Q_{\chi_2}$ is a minimizer implies that
\begin{equation}
\mathcal{E}_{\chi_2}\left(\widetilde{Q}\right)-\kappa \mathrm{Tr}_{\gamma_{\chi_2}}\left(\widetilde{Q}\right) \geq \mathcal{E}_{\chi_2}\left(Q_{\chi_2}\right)-\kappa\mathrm{Tr}_{\gamma_{\chi_2}}\left(Q_{\chi_2}\right).
\label{MinimizationChi2}
\end{equation}
Define $\Theta := \widetilde{Q} -Q_{\chi_2} =  Q_{\chi_1}-Q_{\chi_2} + \gamma_{\chi_1} - \gamma_{\chi_2}$. The inequality~\eqref{MinimizationChi2} can therefore also be written as
\begin{equation}
\mathcal{E}_{\chi_2}\left(\Theta\right) -\kappa \mathrm{Tr}_{\gamma_{\chi_2}}\left(\Theta\right)  +D\left(\rho_{\Theta}, \rho_{Q_{\chi_2}}\right)\geq 0.
\label{eqsup}
\end{equation}
It is easy to see that $-1\leq \Theta \leq 1 $ and $\Theta $ belongs to $\mathcal{Q}_{\chi_2} $ (but not necessarily to the convex set $\mathcal{K}_{\chi_2}$), which also implies that the density $\rho_{\Theta}$ of $\Theta$ is well defined and belongs to the Coulomb space $\mathcal{C}$ (see Section~\ref{SecPropExistMinimizer}). Therefore~\eqref{eqsup} is well defined. Introduce another state by exchanging the indices $1$ and $2$ in the definition of $\widetilde{Q}$: $$\widetilde{\widetilde{Q}}:=\gamma_{\chi_2} + Q_{\chi_2}- \gamma_{\chi_1}. $$ 
Proceeding as before, it can be shown that $\widetilde{\widetilde{Q}}\in \mathcal{K}_{\chi_1}$. By definition $ Q_{\chi_1} =\Theta +\widetilde{\widetilde{Q}}$. Since $Q_{\chi_1}$ minimizes the problem (\ref{minimizationQ}) associated with $\chi_1$ and $\widetilde{\widetilde{Q}}\in \mathcal{K}_{\chi_1}$, 
\[
\mathcal{E}_{\chi_1}\left(\widetilde{\widetilde{Q}}\right)-\kappa\mathrm{Tr}_{\gamma_{\chi_1}}\left(\widetilde{\widetilde{Q}}\right)\geq \mathcal{E}_{\chi_1}\left(\Theta +\widetilde{\widetilde{Q}}\right)-\kappa \mathrm{Tr}_{\gamma_{\chi_1}}\left(\Theta + \widetilde{\widetilde{Q}}\right). 
\]
The above equation can be simplified as
\begin{equation}
\mathcal{E}_{\chi_1}(\Theta) - \kappa \mathrm{Tr}_{\gamma_{\chi_1}}(\Theta)+D\left(\rho_{\Theta}, \rho_{\widetilde{\widetilde{Q}}}\right)\leq 0.
\label{eqinf}
\end{equation}
Let us show that the left hand sides of~\eqref{eqsup} and \eqref{eqinf} are equal. First of all as $\Theta $ belongs to $\mathcal{Q}_{\chi_2} $, we know that $\mathrm{Tr}_{\gamma_{\chi_2}}\left(\Theta\right) = \mathrm{Tr}_{\gamma_{\chi_1}}\left(\Theta\right)$ by Lemma~\ref{TracediffLemma}. Remark also that $\rho_{\widetilde{\widetilde{Q}}}= \rho_{\chi_2} -\rho_{\chi_1} + \rho_{Q_{\chi_2}}$. By Lemma~\ref{TracediffLemma}
\begin{equation}
\begin{aligned}
&\mathcal{E}_{\chi_2}\left(\Theta\right) -\kappa \mathrm{Tr}_{\gamma_{\chi_2}}\left(\Theta\right)  +D\left(\rho_{\Theta}, \rho_{Q_{\chi_2}}\right) - \left(\mathcal{E}_{\chi_1}(\Theta) -\kappa \mathrm{Tr}_{\gamma_{\chi_1}}(\Theta)+D\left(\rho_{\Theta}, \rho_{\widetilde{\widetilde{Q}}}\right)\right)\\
& = \mathrm{Tr}_{\gamma_{\chi_2}}\left((-\Delta + V_{\chi_2}) \Theta\right) - D\left(\rho_{\Theta},\nu_{\chi_2}\right)+  \frac{1}{2}D\left(\rho_\Theta,\rho_\Theta\right) -\mathrm{Tr}_{\gamma_{\chi_1}}\left((-\Delta + V_{\chi_1}) \Theta\right) + D\left(\rho_{\Theta},\nu_{\chi_1}\right)\\
& \qquad -  \frac{1}{2}D\left(\rho_\Theta,\rho_\Theta\right) -\kappa\left( \mathrm{Tr}_{\gamma_{\chi_2}}\left(\Theta\right) - \mathrm{Tr}_{\gamma_{\chi_1}}\left(\Theta\right)\right)  +D\left(\rho_{\Theta}, \rho_{Q_{\chi_2}}\right)  - D\left(\rho_{\Theta}, \rho_{\widetilde{\widetilde{Q}}}\right)  \\
&= \mathrm{Tr}_{\gamma_{\chi_2}}\left((-\Delta + V_{\chi_2}) \Theta\right)  -\mathrm{Tr}_{\gamma_{\chi_1}}\left((-\Delta + V_{\chi_1}) \Theta\right)+D\left(\rho_{\Theta}, \rho_{\chi_1}+\nu_{\chi_1}-\rho_{\chi_2}-\nu_{\chi_2}\right)\\
&= \mathrm{Tr}_{\gamma_{\chi_2}}\left(\left(V_{\chi_2}-V_{\chi_1}\right)\Theta\right) +D\left(\rho_{\Theta}, \rho_{\chi_1}+\nu_{\chi_1}-\rho_{\chi_2}-\nu_{\chi_2}\right).
\end{aligned}
\label{intermE1}
\end{equation}
We show that~\eqref{intermE1} is equal to zero by first showing that $\left(V_{\chi_2}-V_{\chi_1}\right)\Theta\in \mathfrak{S}_1\subset \mathfrak{S}_1^{\gamma_{\chi_2}}$. We start by showing that $(1-\Delta) Q_{\chi_i}\in\mathfrak{S}_2$. By Cauchy's resolvent formula,
$$Q_{\chi_i}  = \frac{1}{2\rmi \pi }\oint_{\mathfrak{C}}\left( \frac{1}{z- H_{\overline{Q}_{\chi_i} }} -\frac{1}{z- H_{\chi_i }}\right) \,d\zeta = Q_{1,i} + Q_{2,i} + Q_{3,i},$$
where 
\begin{align*}
Q_{1,i}&= \frac{1}{2\rmi \pi }\oint_{\mathfrak{C}}\frac{1}{z- H_{\chi_i} } \left(\left(\rho_{\overline{Q}_{\chi_i}}-\nu_{\chi_i}\right)\star \frac{1}{|\cdot|}\right)\frac{1}{z- H_{\chi_i }} \,d\zeta,\\
Q_{2,i}&= \frac{1}{2\rmi \pi }\oint_{\mathfrak{C}} \frac{1}{z- H_{\chi_i} } \left(\left(\rho_{\overline{Q}_{\chi_i}}-\nu_{\chi_i}\right)\star \frac{1}{|\cdot|}\right)\frac{1}{z- H_{\chi_i }}\left(\left(\rho_{\overline{Q}_{\chi_i}}-\nu_{\chi_i}\right)\star \frac{1}{|\cdot|}\right)\frac{1}{z- H_{\chi_i }} \,d\zeta\\
Q_{3,i}&= \frac{1}{2\rmi \pi }\oint_{\mathfrak{C}} \frac{1}{z- H_{\overline{Q}_{\chi_i} } } \left(\left(\rho_{\overline{Q}_{\chi_i}}-\nu_{\chi_i}\right)\star \frac{1}{|\cdot|}\right)\frac{1}{z- H_{\chi_i }}\left(\left(\rho_{\overline{Q}_{\chi_i}}-\nu_{\chi_i}\right)\star \frac{1}{|\cdot|}\right)\frac{1}{z- H_{\chi_i }}\\
&\qquad\qquad\qquad\qquad\cdot\left(\left(\rho_{\overline{Q}_{\chi_i}}-\nu_{\chi_i}\right)\star \frac{1}{|\cdot|}\right)\frac{1}{z- H_{\chi_i }} \,d\zeta.
\end{align*}
Following arguments similar to the ones used in the proof of~\cite[Proposition 2]{Cances2008} we obtain that $\left(\rho_{\overline{Q}_{\chi_i}}-\nu_{\chi_i}\right)\star \frac{1}{|\cdot|}$ belongs to $L^2(\R^3) +\mathcal{C}'$, and we can conclude that $(1-\Delta) Q_{\chi_i}\in\mathfrak{S}_2$. Remark that $ V_{\chi_1}-V_{\chi_2} = (\chi_1^2-\chi_2^2)(\vl-\vr)\in L^{\infty}(\bb R^3) \cap L^2(\bb R^3)$. By definition of $\Theta$, by the Kato--Seiler--Simon inequality and use calculations similar to~\eqref{diffGamma12}
\[
\begin{aligned}
&\bnorm{\left(V_{\chi_1}- V_{\chi_2}\right)\Theta}_{\mathfrak{S}_1}  = \bnorm{ \left(V_{\chi_1}-V_{\chi_2}\right)\left(Q_{\chi_1}-Q_{\chi_2} + \gamma_{\chi_1} - \gamma_{\chi_2}\right)}_{\mathfrak{S}_1} \\
&\leq \bnorm{\left(V_{\chi_1}-V_{\chi_2}\right)(1-\Delta)^{-1}}_{\mathfrak{S}_2}\left(\bnorm{(1-\Delta)\left(Q_{\chi_1}-Q_{\chi_2}\right)}_{\mathfrak{S}_2} +\bnorm{ (1-\Delta)(\gamma_{\chi_1} - \gamma_{\chi_2})}_{\mathfrak{S}_2}\right) \\
&\leq \frac{1}{2\sqrt{\pi}}\bnorm{V_{\chi_1}-V_{\chi_2}}_{L^2}\left(\bnorm{(1-\Delta)\left(Q_{\chi_1}-Q_{\chi_2}\right)}_{\mathfrak{S}_2} +C\norm[\bigg]{(\chi_1^2-\chi_2^2)(\vl-\vr) (1-\Delta)^{-1}}_{\mathfrak{S}_2} \right)<\infty,
\end{aligned}
\]
which proves that $\left(V_{\chi_1}- V_{\chi_2}\right)\Theta$ belongs to $\mathfrak{S}_1$, hence $$\mathrm{Tr}_{\gamma_{\chi_2}}\left(\left(V_{\chi_2}-V_{\chi_1}\right)\Theta\right)  = \mathrm{Tr}\left(\left(V_{\chi_2}-V_{\chi_1}\right)\Theta\right).$$
On the other hand, by the definition of $V_{\chi_i}$ in (\ref{ReferencePotential2}) and $\nu_i$ in (\ref{nuchi}) for $i=1,2$, we deduce that 
\[
\begin{aligned}
\mathrm{Tr}\left(\left(V_{\chi_2}-V_{\chi_1}\right)\Theta\right)&= D\left(\rho_{\Theta}, \left(\rho_{\chi_2}-\mu_{\chi_2}\right)-(\rho_{\chi_1}-\mu_{\chi_1})\right)= D\left(\rho_{\Theta}, \rho_{\chi_2}-\rho_{\chi_1}+\nu_{\chi_2}-\nu_{\chi_1}\right).
\end{aligned}
\]
The above equation implies that the quantity~\eqref{intermE1} equals to $0$. Hence, in view of~\eqref{eqsup} and~\eqref{eqinf},
\[
\mathcal{E}_{\chi_2}\left(\Theta\right) -\kappa \mathrm{Tr}_{\gamma_{\chi_2}}\left(\Theta\right)  +D\left(\rho_{\Theta}, \rho_{Q_{\chi_2}}\right) = \mathcal{E}_{\chi_1}(\Theta) -\kappa \mathrm{Tr}_{\gamma_{\chi_1}}(\Theta)+D\left(\rho_{\Theta}, \rho_{\widetilde{\widetilde{Q}}}\right) \equiv 0.
\]
We conclude with~\eqref{MinimizationChi2} that
\[\mathcal{E}_{\chi_2}(\widetilde{Q})-\kappa \mathrm{Tr}_{\gamma_{\chi_2}}(\widetilde{Q}) \equiv \mathcal{E}_{\chi_2}(Q_{\chi_2})-\kappa \mathrm{Tr}_{\gamma_{\chi_2}}(Q_{\chi_2}).\]
Therefore $\widetilde{Q}$ is a minimizer of the problem~\eqref{minimizationQ} associated with the cut-off function $\chi_2$. From Theorem~\ref{PropExistMinimizer} we know that $\rho_{\widetilde{Q}} \equiv \rho_{Q_{\chi_2}}$, which is equivalent to that $\rho_{Q_{\chi_2}}+\rho_{\chi_2}= \rho_{\chi_1}+\rho_{Q_{\chi_1}}$. By the arbitrariness of the choice of $\chi_1,\chi_2$ we deduce that $\rho_{\chi} + \rho_{Q_{\chi}}$ is independent of the cut-off function $\chi \in \mathcal{X}$.

\bibliographystyle{plain}
\bibliography{Junction}
\end{document}